\documentclass[12pt]{article}
\usepackage{graphicx}
\usepackage{subcaption}
\usepackage{bm,psfrag,afterpage,color,amsmath,amssymb,latexsym,amsthm,url,lscape}
\usepackage{mathtools}
\usepackage{booktabs}
\usepackage{longtable}
\usepackage{array}
\usepackage{multirow}
\usepackage{adjustbox}
\usepackage[margin=1in]{geometry}
\usepackage{algorithm}
\usepackage{algpseudocode}
\usepackage[hidelinks]{hyperref}
\usepackage[numbers,square,sort&compress]{natbib}
\usepackage{xcolor}
\usepackage{float}
\usepackage{thmtools}

\renewcommand{\thefootnote}{\fnsymbol{footnote}}

\newtheorem{proposition}{Proposition}
\newtheorem{lemma}{Lemma}
\newtheorem{corollary}{Corollary}

\numberwithin{equation}{section}

\title{Outcome-Adaptive Lasso with a Mixed-Effects Model under Partial Interference}
\author{Satoshi Nakashima$^1$, Akira Okazaki$^2$, Shuichi Kawano$^3$\\
$^1$Joint Graduate School of Mathematics for Innovation, Kyushu University,\\
Fukuoka, Japan\\
$^2$Institute of Statistical Mathematics, Tokyo, Japan\\
$^3$Faculty of Mathematics, Kyushu University, Fukuoka, Japan
}
\date{}

\begin{document}

\maketitle

\begingroup
\renewcommand{\thefootnote}{}
\footnotetext{
    \textbf{Corresponding author:}
    Satoshi Nakashima,
    nakashima.satoshi.912@s.kyushu-u.ac.jp
}
\endgroup

\begin{abstract}
Interference occurs when one individual's treatment or exposure affects another individual's outcome.
In particular, we assume partial interference, where individuals are divided into groups such that there is no interference between individuals in different groups.
In observational studies, inverse probability weighting (IPW) based on propensity scores is often used for the estimation of causal effects.
However, under partial interference, the group-level propensity score must be estimated and it is more likely to take extreme values than the individual-level propensity score.
Consequently, IPW estimators have large variances.
This problem can become more serious when many covariates are available.
In this study, we propose the Outcome-Adaptive Lasso using a mixed-effects logistic regression model.
The proposed method performs variable selection and estimation in the propensity score model simultaneously while accounting for unobserved heterogeneity through random effects.
Under regularity conditions, we establish the local oracle property and other asymptotic properties of the proposed method.
Through Monte Carlo simulations, we show that the proposed method generally selects covariates relevant for causal effect estimation and improves the finite-sample efficiency of IPW estimators.
We evaluate the performance of the proposed method using the Democratic Republic of the Congo Demographic and Health Survey.
\end{abstract}

\noindent\textbf{keywords:}
Causal inference; Partial interference; Propensity score; Inverse probability weighting; Variable selection.


\section{Introduction}

In causal inference, the Stable Unit Treatment Value Assumption (SUTVA) \citep{Rubin1980}, which consists of no interference and no multiple versions of treatment, has been widely used.
However, this assumption does not always hold.
For example, in infectious disease epidemiology, vaccination of one individual can affect the disease risk of other individuals \citep{Halloran1995, Perez2014}.
In addition, in many fields, it has been pointed out that an intervention or exposure of one individual can affect the outcomes of others \citep{Sobel2006, Rosenbaum2007, Hong2006}.
Thus, when interference is present, it is important to examine not only direct effects but also indirect effects or spillover effects.

In recent years, causal inference under interference has been actively studied.
Various approaches, including counterfactual allocation strategies, exposure mapping, and interference graph approaches, have been proposed
\citep{Hudgens2008, Hoshino2024, Hu2022}.
These frameworks make it possible to define and estimate direct and indirect effects separately.
Many of these studies assume that the network structure among individuals is known.
However, in practice, the network structure is often unknown, and it is difficult to identify the network structure \citep{Bhattacharya2020}.
Partial interference, although a relatively restrictive framework, provides a way to avoid this difficulty \citep{Sobel2006}.
Under partial interference, individuals are divided into multiple non-overlapping groups, and interference is assumed to occur within the same group but not between different groups.
In this assumption, it is easy to interpret and practically useful in settings with natural group structures, such as households, schools, regions, and hospitals.
Hudgens and Halloran~\citep{Hudgens2008} defined the direct, indirect/spillover, total, and overall effect in partial interference.
Since then, causal inference under partial interference has been actively studied \citep{Tchetgen2012,Liu2019,Imai2021, Kilpatrick2025}.

In observational studies, confounding adjustment is an important issue because treatment groups are not directly comparable.
Inverse probability weighting (IPW) based on propensity scores is widely used to adjust for confounding.
Tchetgen Tchetgen and VanderWeele~\citep{Tchetgen2012} proposed an IPW estimator based on the group-level propensity score.
They also proposed estimating the group-level propensity score using a mixed-effects logistic regression model.
However, propensity score estimation under partial interference can be unstable because the group-level propensity score involves a product of individual-level treatment probabilities.
Consequently, the variance of the IPW estimator can increase.

The choice of covariates included in the propensity score model can also affect the bias and variance of IPW estimators.
In particular, including covariates not related to the outcome can increase the bias or variance, whereas including covariates related to the outcome can improve estimation efficiency \citep{Lunceford2004,Brookhart2006,Rotnitzky2010}.
Motivated by these findings, Shortreed and Ertefaie ~\citep{Shortreed2017} proposed the Outcome-Adaptive Lasso (OAL), which uses outcome information for variable selection in the propensity score model.
OAL aims to improve statistical efficiency by selecting outcome predictors and confounders, while excluding covariates not related to the outcome, such as instrumental variables.
However, the OAL is based on a logistic regression model, and its theoretical results on consistency in variable selection and asymptotic normality, known as the oracle property, do not directly apply to mixed-effects logistic regression models.

In this study, we propose the OAL using a mixed-effects logistic regression model.
The proposed method simultaneously performs variable selection and propensity score estimation while accounting for unobserved heterogeneity through the random effects.
Although the objective function is nonconvex, we establish the local oracle property under regularity conditions.
Furthermore, we establish the consistency and asymptotic normality of the IPW estimators under the proposed method.

The remainder of this paper is organized as follows.
In Section~\ref{sec:prelim}, we introduce the notation under partial interference, define the estimands, and introduce the IPW estimators.
In Section~\ref{sec:propose}, we propose OAL based on a mixed-effects logistic regression model and describe how to select the tuning parameter.
In Section~\ref{sec:tp}, we describe the theoretical properties of the proposed method.
In Section~\ref{sec:sim}, we evaluate the finite-sample performance of the proposed method through numerical experiments.
In Section~\ref{sec:app}, we evaluate the performance of the proposed method using the Democratic Republic of the Congo Demographic and Health Survey (DRC-DHS).
In Section~\ref{sec:conc}, we give discussion and conclusions.

\section{Preliminaries}\label{sec:prelim}

In Section~\ref{subsec:nota}, we introduce the notation used throughout this paper.
In Section~\ref{subsec:est}, we define the causal estimands based on \citep{Hudgens2008,Tchetgen2012}.
In Section~\ref{subsec:ipw}, we describe inverse probability weighting estimators for partial interference.

\subsection{Notation}\label{subsec:nota}

We consider observational data consisting of a total of $n$ individuals and $m$ groups, where each individual belongs to one of the groups.
Let $n_i$ denote the size of group $i$, so that $\sum_{i=1}^m n_i=n$.
For individual $j$ in group $i$, we observe $O_{ij}=(X_{ij},A_{ij},Y_{ij})$ for $i=1,\dots,m,j=1,\dots,n_i$, where $X_{ij} \in \mathbb{R}^p$ is a covariate vector, $A_{ij} \in \{0,1\}$ is a binary treatment assignment (1: treated, 0: untreated), and $Y_{ij} \in \mathbb{R}$ is an outcome of interest.
Let $X_i=(X_{i1},\dots,X_{in_i})$, $A_i=(A_{i1},\dots,A_{in_i})$, and $Y_i=(Y_{i1},\dots,Y_{in_i})$.
Assume that $O_i=(X_i,A_i,Y_i)$ for $i=1,\dots,m$ are a random sample from a superpopulation of groups. 
Let $A_{i(-j)}=\{A_{i1},\dots,A_{i,j-1},A_{i,j+1},\dots,A_{in_i}\}$ be the treatment assignment vector for group $i$ excluding individual $j$.
Let $a_{ij}$, $a_{i(-j)}$, and $a_i$ denote realizations of $A_{ij}$, $A_{i(-j)}$, and $A_i$, respectively.
Let $\mathbb{A}(n_i)=\{0,1\}^{n_i}$ be the set of possible treatment assignment vectors for a group of size $n_i$, such that $a_i \in \mathbb{A}(n_i)$.
We assume partial interference: there is interference within the same group, but not between different groups.
Thus, the potential outcomes of individuals depend only on the treatment assignment within the same group.

\subsection{Estimands}\label{subsec:est}

Under the partial interference assumption, the potential outcome of an individual may depend on the treatment assignments of all individuals in the same group, but not on the treatment assignments in other groups.
Therefore, we denote the potential outcome of individual $j$ in group $i$ by $Y_{ij}(a_i)=Y_{ij}(a_{ij},a_{i(-j)})$ for $a_i \in \mathbb{A}(n_i)$.
The potential outcomes for group $i$ are defined as $Y_i(a_i)=(Y_{i1}(a_i),\dots,Y_{in_i}(a_i))$.
Assume causal consistency: $Y_{ij}=\sum_{a_i \in \mathbb{A}(n_i)}\mathbf{1}\{A_i=a_i\}Y_{ij}(a_i)$ for the observed outcome $Y_{ij}$, where $\mathbf{1}\{\cdot\}$ denotes the indicator function.

The causal effect of treatment is defined as the difference between counterfactual scenarios that correspond to different allocation strategies or different treatment assignments.
Following \citep{Tchetgen2012}, we consider a treatment allocation strategy where each individual independently receives treatment with probability $\theta \in (0,1)$.
Under allocation strategy $\theta$, the probability that group $i$ has treatment assignment $A_i=a_i$ is given by $\pi(a_i;\theta)=\prod_{j=1}^{n_i}\theta^{a_{ij}}(1-\theta)^{1-a_{ij}}$.
Similarly, the probability of the treatment assignment $A_{i(-j)}=a_{i(-j)}$ is given by $\pi(a_{i(-j)};\theta)=\prod_{k\neq j}\theta^{a_{ik}}(1-\theta)^{1-a_{ik}}$.
Under allocation strategy $\theta$, the average potential outcome for individual $j$ in group $i$ when an individual receives treatment $a$ is defined as $\overline{Y}_{ij}(a;\theta)=\sum_{a_{i(-j)} \in \mathbb{A}(n_i-1)} Y_{ij}(a_{ij}=a,a_{i(-j)})\pi(a_{i(-j)};\theta)$.
The group average potential outcome for group $i$ is defined as $\overline{Y}_{i}(a;\theta)=n_i^{-1}\sum_{j=1}^{n_i} \overline{Y}_{ij}(a;\theta)$.
The population average potential outcome is defined as $\mu_{a\theta}=\mathbb{E}[\overline{Y}_i(a;\theta)]$, where the expectation is taken over the superpopulation of groups.
Similarly, under allocation strategy $\theta$, the marginal average potential outcome for individual $j$ in group $i$ is defined as $\overline{Y}_{ij}(\theta)=\sum_{a_i \in \mathbb{A}(n_i)} Y_{ij}(a_{ij},a_{i(-j)})\pi(a_i;\theta)$.
The group average potential outcome for group $i$ is defined as $\overline{Y}_{i}(\theta)=n_i^{-1}\sum_{j=1}^{n_i} \overline{Y}_{ij}(\theta)$.
The population average potential outcome is defined as $\mu_\theta=\mathbb{E}[\overline{Y}_i(\theta)]$.

Following \citep{Hudgens2008}, we define the direct, indirect,
total, and overall effects as follows.
For a fixed allocation strategy $\theta$, the direct effect (DE)
is defined as the difference in the population average potential
outcomes between receiving and not receiving treatment: $\overline{\mathrm{DE}}(\theta)=\mu_{1\theta}-\mu_{0\theta}$.
For two allocation strategies $\theta_0,\theta_1\in(0,1)$, the indirect effect (IE) is defined as the difference in the population average potential outcomes among untreated individuals: $\overline{\mathrm{IE}}(\theta_1,\theta_0)=\mu_{0\theta_1}-\mu_{0\theta_0}$.
The total effect (TE) compares receiving treatment under the allocation strategy $\theta_1$ with not receiving treatment under the allocation strategy $\theta_0$: $\overline{\mathrm{TE}}(\theta_1,\theta_0)=\mu_{1\theta_1}-\mu_{0\theta_0}$.
The overall effect (OE) is defined as the difference in the population marginal average potential outcomes under two allocation strategies: $\overline{\mathrm{OE}}(\theta_1,\theta_0)=\mu_{\theta_1}-\mu_{\theta_0}$.

\subsection{Inverse Probability Weighting}\label{subsec:ipw}

When estimating causal effects in observational studies, it is necessary to adjust for confounding due to differences in baseline characteristics between the treatment and control groups.
Inverse probability weighting (IPW) using the propensity score is a common method for adjusting for confounding. 
IPW estimators under partial interference use the group-level propensity score, which is defined as $f(A_i \mid X_i;\psi)=\mathbb{P}(A_i=a_i \mid X_i=x_i)$, where $\psi$ denotes the parameter vector of the propensity score model.
To identify causal effects from the observational data, assume the following identification conditions.
First is the positivity assumption: $\mathbb{P}(A_i=a_i \mid X_i=x_i)>0$ holds for all $x_i$ in the support of $X_i$ and $a_i \in \mathbb{A}(n_i)$.
Second is the conditional exchangeability assumption: $Y_i(a_i) \perp\!\!\!\perp A_i \mid X_i$.

For the population average potential outcome under the allocation strategy $\theta$, the H\'{a}jek-type IPW estimator \citep{Hajek1971, Tchetgen2012} is defined as follows:
\begin{equation}
\hat{Y}(a;\theta)=\left.\sum_{i=1}^m\sum_{j=1}^{n_i}\frac{Y_{ij}\mathbf{1}\{A_{ij}=a\}\pi(A_{i(-j)};\theta)}{n_if(A_i \mid X_i;\psi)} \middle/ \sum_{i=1}^m \sum_{j=1}^{n_i}\frac{\mathbf{1}\{A_{ij}=a\}\pi(A_{i(-j)};\theta)}{n_if(A_i \mid X_i;\psi)}\right..
\label{eq:ipw_haj}
\end{equation}
Similarly, the IPW estimator for the marginal average potential outcome is defined as follows:
\[
\hat{Y}(\theta)=\left.\sum_{i=1}^m\sum_{j=1}^{n_i}\frac{Y_{ij}\pi(A_i;\theta)}{n_if(A_i \mid X_i;\psi)} \middle/ \sum_{i=1}^m\sum_{j=1}^{n_i}\frac{\pi(A_i;\theta)}{n_if(A_i \mid X_i;\psi)}\right..
\]
We define $\widehat{\mathrm{DE}}(\theta)=\hat{Y}(1;\theta)-\hat{Y}(0;\theta)$ as the direct effect estimator.
Similarly, the indirect, total, and overall effect estimators are defined as $\widehat{\mathrm{IE}}(\theta_1,\theta_0)=\hat{Y}(0;\theta_1)-\hat{Y}(0;\theta_0)$, $\widehat{\mathrm{TE}}(\theta_1,\theta_0)=\hat{Y}(1;\theta_1)-\hat{Y}(0;\theta_0)$, and $\widehat{\mathrm{OE}}(\theta_1,\theta_0)=\hat{Y}(\theta_1)-\hat{Y}(\theta_0)$, respectively.

In observational studies, the propensity score is generally unknown and must be estimated. 
When the number of covariates is large, nonparametric estimation of the propensity score is generally impractical due to the curse of dimensionality.
Therefore, propensity scores are often estimated using parametric models such as logistic regression models.
Tchetgen Tchetgen and VanderWeele~\citep{Tchetgen2012} proposed estimating the propensity score using the following mixed-effects logistic regression model:
\[
f(A_i \mid X_i;\psi)=\int\prod_{j=1}^{n_i}p_{ij}(b_i;\beta)^{A_{ij}}\{1-p_{ij}(b_i;\beta)\}^{1-A_{ij}}\phi(b_i;\sigma^2) \, db_i,
\]
where $\psi=(\beta,\sigma^2)$ is the parameter vector, $p_{ij}(b_i;\beta)=\mathbb{P}(A_{ij}=1 \mid X_{ij},b_i)=\mathrm{expit}(\sum_{k=1}^p\beta_k X_{ijk}+b_i)$, $b_i$ is a random effect of group $i$, and $\phi(b_i;\sigma^2)$ is the density function of $b_i$ where $b_i$ is assumed to follow $N(0,\sigma^2)$.
A maximum likelihood estimator of $\psi$ is obtained by maximizing the log-likelihood function
\begin{equation}
\sum_{i=1}^m\log f(A_i \mid X_i;\psi).
\label{loglike}
\end{equation}
Under regularity conditions, the maximum likelihood estimator is $\sqrt{m}$-consistent. 
This model can account for unobserved heterogeneity in treatment assignment through the random effects.

\begin{figure}[t]
    \centering
    \includegraphics[width=0.88\textwidth]{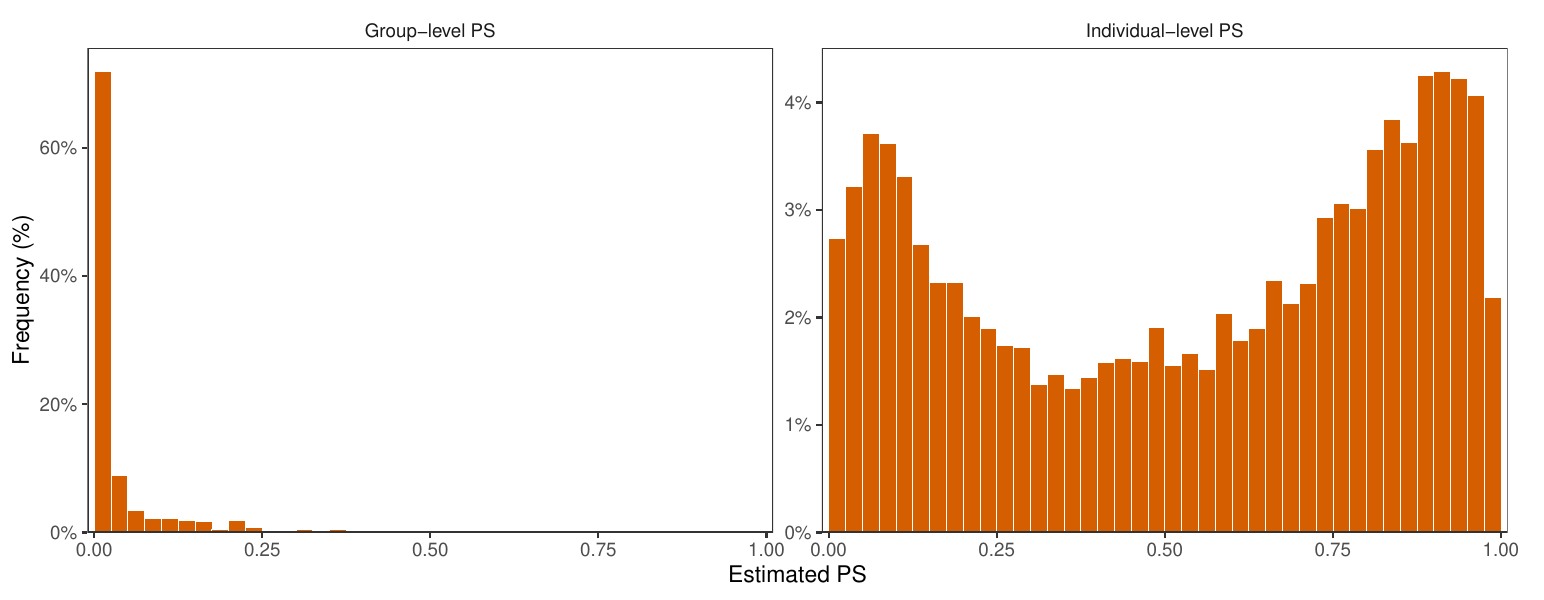}
    \caption{
    Histograms of estimated propensity scores in the DRC-DHS data analyzed in Section~\ref{sec:app}.
    The analysis included 531 groups, with a median group size of 14 (IQR: 11--17).
    The left and right panels show the distributions of the estimated
    group-level and individual-level propensity scores, respectively.
    The horizontal axis represents the estimated propensity score, and the vertical axis represents the relative frequency (\%).
    }
    \label{fig:2.1}
\end{figure}

The group-level propensity score involves a product of individual-level treatment probabilities, which can lead to extremely small values.
Figure~\ref{fig:2.1} shows the estimated group-level and individual-level propensity scores in the DRC-DHS data analyzed in Section~\ref{sec:app}.
From the left panel of the figure, we observe that the estimated group-level propensity scores are concentrated near zero. 
This problem can be more serious when the candidate covariates include variables such as instrumental variables, which do not reduce confounding bias but can increase the variance of IPW estimators.
However, identifying such variables is difficult when many candidate covariates are available.
Therefore, data-driven variable selection is desirable.

\section{Proposed Method}\label{sec:propose}

In this section, we propose a method for stably estimating the group-level propensity score required for estimating the causal effect under partial interference, even when the number of covariates is relatively large.
Specifically, we extend the OAL defined by the logistic regression model of \citep{Shortreed2017} to a mixed-effects logistic regression model.
In Section~\ref{subsec:oal}, we introduce the mixed-effects OAL based on the mixed-effects logistic regression model, and establish the asymptotic normality and consistency in variable selection of the proposed estimator.
In Section~\ref{subsec:sel}, we describe how to select the tuning parameter.

\subsection{Outcome-Adaptive Lasso with a Mixed-Effects Model}\label{subsec:oal}

Candidate covariates are classified into four sets:
$\mathcal{C}$ (confounders), variables related to both the outcome and the treatment; $\mathcal{P}$ (outcome predictors), variables related only to the outcome; $\mathcal{I}$ (instrumental variables), variables related only to the treatment; and $\mathcal{S}$ (spurious variables), variables related to neither the outcome nor the treatment.
From the viewpoint of causal effect estimation, appropriately adjusting for confounders $X_\mathcal{C}$ reduces bias, and including outcome predictors $X_\mathcal{P}$ may also improve estimation efficiency.
On the other hand, including instrumental variables $X_\mathcal{I}$ or spurious variables $X_\mathcal{S}$ may increase the variance or bias of the causal effect estimator.
Therefore, it is desirable to use a variable selection method that selects confounders and outcome predictors, while excluding instrumental variables and spurious variables from the propensity score model.

In this study, we call the method that applies the outcome-adaptive lasso to the mixed-effects logistic regression model the mixed-effects OAL.
Following the log-likelihood function \eqref{loglike}, the mixed-effects OAL estimator is given by
\begin{equation}
\hat\psi=\arg\min_{\psi} \left[-\sum_{i=1}^m \log f(A_i \mid X_i;\psi)+\lambda_m\sum_{k=1}^p \hat\omega_k|\beta_k|\right],
\label{eq:mixed_oal}
\end{equation}
where $\lambda_m \ (>0)$ is a regularization parameter, $\hat\omega_k = |\hat\alpha_k|^{-\zeta}$ with $\zeta>1$ and $\hat{\alpha}_k$ is an estimate of the regression coefficient for covariate $X_k$ in the outcome regression model.
For example, $\hat{\alpha}_k$ can be taken to be the ordinary least squares estimates obtained from the linear regression model
\begin{equation*}
Y_{ij}=\alpha_0+\alpha_t A_{ij}+\alpha_{int}(n_i-1)^{-1}\sum_{k \neq j} A_{ik}+\sum_{k=1}^{p}\alpha_k X_{ijk}+\varepsilon_{ij}.
\end{equation*}
Note that the penalty is imposed only on the fixed-effect parameters $\beta$, and not on the random-effect parameters.

The mixed-effects OAL estimator is numerically obtained by modifying the algorithm implemented in the R package \texttt{glmmLasso} \citep{Groll2014}.
Specifically, the ordinary $L_1$ penalty in \texttt{glmmLasso} is replaced with the outcome-adaptive penalty $\lambda_m\sum_{k=1}^p\hat\omega_k|\beta_k|$.
The fixed-effect parameters are updated using the gradient descent algorithm, which switches to Newton--Raphson updates when appropriate. The random-effect variance parameter is estimated using the EM-type algorithm.

We next establish the asymptotic properties of the mixed-effects OAL estimator.
Without loss of generality, we write $\mathcal{C} \cup \mathcal{P}=\{1,\dots,p_0\}$ with $p_0<p=|\mathcal{C}|+|\mathcal{P}|+|\mathcal{I}|+|\mathcal{S}|$.
Assume the propensity score model based on $X_{\mathcal{C}}$ and $X_\mathcal{P}$ is correctly specified.
Let $\psi^\ast$ be the true parameter value of the reduced propensity score model:
\[
\mathbb{P}(A_{ij}=1 \mid X_{ij},b_i)=\mathrm{expit}\left(\sum_{k \in \mathcal{C}}\beta_k X_{ijk}+\sum_{k \in \mathcal{P}}\beta_k X_{ijk}+b_i\right).
\]
We define the Fisher information matrix:
\[
I(\psi^\ast)=
\begin{pmatrix}
I_{11} & I_{12}\\
I_{21} & I_{22}
\end{pmatrix}.
\]
Note that $I_{11}$ is the Fisher information matrix corresponding to $\psi_{\mathcal{C} \cup \mathcal{P}}=(\beta_1,\dots,\beta_{p_0},\sigma^2)$.
Since the penalized negative log-likelihood function $-\sum_{i=1}^m\log f(A_i \mid X_i;\psi)+\lambda_m\sum_{k=1}^p\hat\omega_k|\beta_k|$ is not necessarily globally convex, all the results below are local results in a neighborhood of $\psi^\ast$.
Following the theoretical strategy for nonconcave penalized likelihood estimators in \citep{Fan2001}, we first establish the $\sqrt{m}$-consistency of a local minimizer and then derive its asymptotic normality and consistency in variable selection. 

\begin{lemma}
\label{lem:rootm_consistency}
Assume the conditions (A)--(E) in Appendix \ref{app:RC}.
If $\lambda_m/\sqrt{m} \to 0$ as $m \to \infty$, then there exists a local minimizer $\hat\psi_{OAL}$ of the penalized negative log-likelihood function $-\sum_{i=1}^m\log f(A_i \mid X_i;\psi)+\lambda_m\sum_{k=1}^p\hat\omega_k|\beta_k|$ such that
\[
\|\hat\psi_{OAL}-\psi^\ast\|_2=O_p(m^{-1/2}).
\]
\end{lemma}

\begin{proposition}[Oracle Property]
\label{prop:oracle}
Assume the conditions (A)--(E) in Appendix \ref{app:RC}. 
If $\lambda_m/\sqrt{m} \to 0$ and $\lambda_m m^{\zeta/2-1} \to \infty$ as $m \to \infty$, the $\sqrt{m}$-consistent local minimizer $\hat\psi_{OAL}$ in Lemma \ref{lem:rootm_consistency} satisfies the following properties:
\begin{itemize}
\item[(i)] Consistency in variable selection:
\[
\lim_{m \to \infty}\mathbb{P}(\hat{\beta}_k=0 \mid k \in \mathcal{I} \cup \mathcal{S})=1.
\]
\item[(ii)] Asymptotic normality:
\[
\sqrt{m}(\hat\psi_{OAL}-\psi_{\mathcal{C} \cup \mathcal{P}}^\ast) \overset{d}{\longrightarrow} N(0,I_{11}^{-1}).
\]
\end{itemize}
\end{proposition}
\noindent The proofs of Lemma \ref{lem:rootm_consistency} and Proposition \ref{prop:oracle} are provided in Appendices \ref{proof:lem:rootm_consistency} and \ref{proof:prop:oracle}, respectively. 

Proposition \ref{prop:oracle} establishes the local oracle property of the mixed-effects OAL estimator.
Specifically, covariates unrelated to the outcome are excluded with probability tending to one, while the estimators of the active parameters have the same asymptotic distribution as the maximum likelihood estimator under the reduced model known in advance.

\subsection{Tuning Parameter Selection}\label{subsec:sel}

In existing Lasso methods, prediction-based criteria such as cross-validation are often used to select the tuning parameter.
However, our goal is not to maximize the prediction accuracy of treatment assignment, but to adjust covariate balance for causal effect estimation.
Therefore, we use the weighted absolute mean difference (wAMD) proposed in~\citep{Shortreed2017} as a measure of covariate balance after weighting.
The wAMD is the sum of the absolute differences in weighted covariate means between treated and untreated individuals, with each difference weighted by the estimated strength of association between the corresponding covariate and the outcome.
The wAMD based on the IPW estimator \eqref{eq:ipw_haj} is defined as
\[
\begin{aligned}
\mathrm{wAMD}(\lambda_m,\theta)
=&\sum_{k=1}^{p}|\hat{\alpha}_k|\left|\sum_{i=1}^m\sum_{j=1}^{n_i}\frac{A_{ij}\pi(A_{i(-j)};\theta)X_{ijk}}{n_i f(A_i\mid X_i;\hat{\psi}_{\lambda_m})} \middle/ \sum_{i=1}^m\sum_{j=1}^{n_i}
\frac{A_{ij}\pi(A_{i(-j)};\theta)}
{n_i f(A_i\mid X_i;\hat{\psi}_{\lambda_m})}\right.\\
&\quad -\left.\sum_{i=1}^m\sum_{j=1}^{n_i}\frac{(1-A_{ij})\pi(A_{i(-j)};\theta)X_{ijk}}{n_i f(A_i\mid X_i;\hat{\psi}_{\lambda_m})} \middle/ \sum_{i=1}^m\sum_{j=1}^{n_i}\frac{(1-A_{ij})\pi(A_{i(-j)};\theta)}{n_i f(A_i\mid X_i;\hat{\psi}_{\lambda_m})}\right|,
\end{aligned}
\]
where $\hat\psi_{\lambda_m}$ is the mixed-effects OAL estimator under $\lambda_m$.

Under partial interference, the wAMD depends on the allocation strategy $\theta$.
Therefore, we select the tuning parameter $\lambda_m$ by minimizing
\[
\mathrm{wAMD}(\lambda_m)=|\Theta|^{-1}\sum_{\theta \in \Theta} \mathrm{wAMD}(\lambda_m,\theta).
\]
Here, $\Theta$ denotes a prespecified set of allocation strategies (e.g., $\Theta=\{0.1,0.2,\ldots,0.9\}$).
This criterion evaluates covariate balance across multiple allocation strategies, rather than at a single value of $\theta$.

\section{Theoretical Property}\label{sec:tp}

In this section, we show that adjusting for not only confounders $X_{\mathcal{C}}$ but also outcome predictors $X_{\mathcal{P}}$ can improve the efficiency of IPW estimators under partial interference.
In Section~\ref{subsec:conf}, We consider the propensity score model based only on $X_\mathcal{C}$.
In Section~\ref{subsec:op}, we examine the effect of additional covariates $X_\mathcal{P}$ in the propensity score model.

\subsection{Adjustment for Confounders}\label{subsec:conf}

For notational simplicity, let $X_i = X_{\mathcal{C},i}$ and $V_i = X_{\mathcal{P},i}$.
We assume the conditional exchangeability and positivity assumptions in terms of $X_i$: $Y_i(a_i) \perp\!\!\!\perp A_i \mid X_i$ and $\mathbb{P}(A_i=a_i \mid X_i=x_i)>0$.
We first consider the estimating equation for the IPW estimator.
The IPW estimator defined in Section~\ref{subsec:ipw} is a solution to the estimating equation $\sum_{i=1}^m U_\mu(X_i,A_i,Y_i;\mu)=0$ for $\mu$, where
\[
U_\mu(X_i,A_i,Y_i;\mu)=\sum_{j=1}^{n_i}\left\{\frac{Y_{ij}\mathbf{1}\{A_{ij}=a\}\pi(A_{i(-j)};\theta)}{n_if(A_i \mid X_i;\beta,\sigma^2)}-\mu\frac{\mathbf{1}\{A_{ij}=a\}\pi(A_{i(-j)};\theta)}{n_if(A_i \mid X_i;\beta,\sigma^2)}\right\}.
\]
We denote a solution of the equation $\sum_{i=1}^m U_\mu(X_i,A_i,Y_i;\mu)=0$ as $\hat\mu$ and a solution of the equation $\mathbb{E}[U_\mu(X_i,A_i,Y_i;\mu)]=0$ as $\mu^\ast$.
It is straightforward to show that $\mu^\ast = \mathbb{E}[n_i^{-1}\sum_{j=1}^{n_i}\overline{Y}_{ij}(a;\theta)]$.
Under regularity conditions for M-estimation, $\hat\mu \overset{p}{\longrightarrow} \mu^\ast$ and $\sqrt{m}(\hat\mu-\mu^\ast)$ converges in distribution to $N(0,\Sigma_{\mathrm{known}})$, where $\Sigma_{\mathrm{known}}=Q_0^{-1}\Omega_0 Q_0^{-1}$ with $Q_0=\mathbb{E}[-\partial U_\mu(X_i,A_i,Y_i;\mu^\ast)/\partial\mu]$ and $\Omega_0=\mathbb{E}[U_\mu(X_i,A_i,Y_i;\mu^\ast)^2]$ \citep{Stefanski2002, Perez2014}.

Next, we consider the estimating equations for the propensity score model.
We denote the log-likelihood for the mixed-effects logistic model as $l(X_i,A_i;\beta,\sigma^2)=\log f(A_i \mid X_i;\beta,\sigma^2)$.
The maximum likelihood estimators $\hat\beta$, $\hat\sigma^2$ are solutions to the score equations
\[
\sum_{i=1}^m S_\beta(X_i,A_i;\beta,\sigma^2)=0 \quad \mathrm{and} \quad \sum_{i=1}^m S_{\sigma^2}(X_i,A_i;\beta,\sigma^2)=0,
\]
where $S_\beta(X_i,A_i;\beta,\sigma^2)=\partial l(X_i,A_i;\beta,\sigma^2)/\partial\beta$, and $S_{\sigma^2}(X_i,A_i;\beta,\sigma^2)=\partial l(X_i,A_i;\beta,\sigma^2)/\partial\sigma^2$.

Therefore, when the IPW estimator and the propensity score parameters are estimated simultaneously, the corresponding estimator $\hat\eta=(\hat\mu,\hat\beta,\hat\sigma^2)$ is a solution to the stacked estimating equation $\sum_{i=1}^m U(X_i,A_i,Y_i;\eta)=0$, where 
\[
U(X_i,A_i,Y_i;\eta)=
\begin{pmatrix}
U_\mu(X_i,A_i,Y_i;\mu,\psi)\\
S_\psi(X_i,A_i;\psi)
\end{pmatrix}
\]
with $\psi=(\beta,\sigma^2)$ and $S_\psi(X_i,A_i;\beta,\sigma^2)=(S_\beta(X_i,A_i;\beta,\sigma^2),S_{\sigma^2}(X_i,A_i;\beta,\sigma^2))^\top$.
We denote a solution of the equation $\mathbb{E}[U(X_i,A_i,Y_i;\eta)]=0$ as $\eta^\ast$.

\begin{lemma}
\label{Lemma2}
Suppose that the propensity score model is correctly specified and the standard regularity conditions for M-estimation hold. 
Then, the following holds:
\[
\sqrt{m}\left(\hat\mu-\mu^\ast\right) \overset{d}{\longrightarrow} N(0,\Sigma_{\mathrm{est}}) \quad \mathrm{where} \quad \Sigma_{\mathrm{est}}=\Omega_{11}-\Omega_{12}\Omega_{22}^{-1}\Omega_{12}^\top.
\]
with $\Omega_{11}=\mathbb{E}\left[U_\mu(X_i,A_i,Y_i;\mu^\ast,\psi^\ast)^2\right]$, $\Omega_{12}=\mathbb{E}\left[U_\mu(X_i,A_i,Y_i;\mu^\ast,\psi^\ast) S_\psi(X_i,A_i;\psi^\ast)^\top\right]$, $\Omega_{22}=\mathbb{E}\left[S_\psi(X_i,A_i;\psi^\ast)S_\psi(X_i,A_i;\psi^\ast)^\top\right]$.

Furthermore, the following holds:
\[
\Sigma_{\mathrm{est}} \leq \Sigma_{\mathrm{known}}.
\]
\end{lemma}
\noindent
The proof of Lemma~\ref{Lemma2} is provided in Appendix~\ref{app:variance1}.

Lemma~\ref{Lemma2} shows that accounting for the uncertainty in propensity score estimation can reduce the asymptotic variance of the IPW estimator compared with treating the propensity score as known.

\subsection{Additional Adjustment for Outcome Predictors}\label{subsec:op}

Next, we consider the effect of additionally adjusting for outcome predictors $V_i$ in the propensity score model. 
Recall that $V_i$ is not related to treatment assignment,  but is associated with outcome, so that $\mathbb{P}(A_i \mid X_i,V_i)=\mathbb{P}(A_i \mid X_i)$.
We also assume conditional exchangeability given both $X_i$ and $V_i$: $Y_i(a_i) \perp\!\!\!\perp A_i \mid X_i,V_i$.
Thus, we can consider the propensity score model $\mathbb{P}(A_i=a_i \mid X_i,V_i)=f(A_i \mid X_i,V_i;\psi,\gamma)$, where $\psi=(\beta,\sigma^2)$ and $\gamma$ is an additional parameter vector
corresponding to $V_i$. 
We assume that this model reduces to the true propensity score model depending only on $X_i$ and $\psi$ when $\gamma=0$: $f(A_i \mid X_i,V_i;\psi,0)=f(A_i \mid X_i;\psi)$.
We further assume that
\[
\left.\frac{\partial}{\partial\psi}f(A_i \mid X_i,V_i;\psi,\gamma)\right|_{\gamma=0}=\frac{\partial}{\partial\psi}f(A_i \mid X_i;\psi).
\]
When $V_i$ is included in the propensity score model, $\gamma$ is estimated simultaneously with $\eta=(\mu,\beta,\sigma^2)$.
Thus, the effect of including $V_i$ on the asymptotic distribution can be examined by replacing $f(A_i \mid X_i;\psi)$ with $f(A_i \mid X_i,V_i;\psi,\gamma)$ in the preceding estimating equations and by adding the estimating equation
\[
\sum_{i=1}^m S_\gamma(X_i,V_i,A_i;\psi,\gamma )=0,
\]
where $S_\gamma(X_i,V_i,A_i;\psi,\gamma)=\partial l(X_i,V_i,A_i;\psi,\gamma)/\partial\gamma$.

Let $\eta_V=(\mu,\beta,\gamma,\sigma^2)$.
Then, we denote the solution of equation $\sum_{i=1}^m U(X_i,V_i,A_i,Y_i;\eta_V)=0$ as $\hat\eta_V$ and denote the solution of equation $\mathbb{E}[U(X_i,V_i,A_i,Y_i;\eta_V)]=0$ as $\eta_V^\ast$.

\begin{lemma}
\label{Lemma3}
Suppose that the propensity score model including the outcome predictors $V_i$ is correctly specified and that the standard regularity conditions for M-estimation hold. 
Then, the following holds:
\[
\sqrt{m}\left(\hat\mu^V-\mu^\ast\right) \overset{d}{\longrightarrow} N\left(0,\Sigma_{\mathrm{est}}^V\right),
\]
where
\[
\Sigma_{\mathrm{est}}^V = \Sigma_{\mathrm{est}}-(\Omega_{13}-\Omega_{12}\Omega_{22}^{-1}\Omega_{23})(\Omega_{33}-\Omega_{23}^\top\Omega_{22}^{-1}\Omega_{23})^{-1}(\Omega_{13}-\Omega_{12}\Omega_{22}^{-1}\Omega_{23})^\top.
\]
with $\Omega_{13}=\mathbb{E}\left[U_\mu(X_i,V_i,A_i,Y_i;\mu^\ast,\psi^\ast,\gamma^\ast) S_\gamma(X_i,V_i,A_i;\psi^\ast,\gamma^\ast)^\top\right]$,\\
$\Omega_{23}=\mathbb{E}[S_\psi(X_i,V_i,A_i;\psi^\ast,\gamma^\ast)S_\gamma(X_i,V_i,A_i;\psi^\ast,\gamma^\ast)^\top]$ and\\
$\Omega_{33}=\mathbb{E}[S_\gamma(X_i,V_i,A_i;\psi^\ast,\gamma^\ast)S_\gamma(X_i,V_i,A_i;\psi^\ast,\gamma^\ast)^\top]$.

Furthermore, the following holds:
\[
\Sigma_{\mathrm{est}}^{V} \le \Sigma_{\mathrm{est}}.
\]
\end{lemma}
\noindent
The proof of Lemma~\ref{Lemma3} is provided in Appendix~\ref{app:variance2}.

Lemma~\ref{Lemma3} shows that additional covariates such as outcome predictors in the propensity score model can improve asymptotic variance of the IPW estimator.
Similar result was obtained in \citep{Lunceford2004} in settings without interference.

The proposed method excludes covariates in $\mathcal{I}$ and $\mathcal{S}$ with probability tending to one.
Therefore, we focus on the propensity score model containing only confounders and outcome predictors.
However, the mixed-effects OAL estimator does not exactly satisfy the usual score equations because of the penalty.
The following proposition shows that the effect of the penalty is asymptotically negligible.

\begin{proposition}
\label{prop2}
Assume the conditions (A)--(E) in Appendix \ref{app:RC}. 
If $\lambda_m/\sqrt{m} \to 0$ as $m \to \infty$, the mixed-effects OAL estimator $\hat\psi_{OAL}=(\hat\beta_{OAL},\hat\gamma_{OAL},\hat\sigma_{OAL}^2)$ satisfies the following properties:
\[
\frac{1}{m}\sum_{i=1}^m S_\beta(X_i,V_i,A_i;\hat\beta_{OAL},\hat\gamma_{OAL},\hat\sigma^2_{OAL})=o_p(m^{-1/2}),
\]
\[
\frac{1}{m}\sum_{i=1}^m S_\gamma(X_i,V_i,A_i;\hat\beta_{OAL},\hat\gamma_{OAL},\hat\sigma^2_{OAL})=o_p(m^{-1/2})
\]
and
\[
\frac{1}{m}\sum_{i=1}^m S_{\sigma^2}(X_i,V_i,A_i;\hat\beta_{OAL},\hat\gamma_{OAL},\hat\sigma^2_{OAL}) = 0.
\]
\end{proposition}
\noindent
The proof of Proposition~\ref{prop2} is provided in Appendix~\ref{proof:prop2}.

By Proposition~\ref{prop2}, the mixed-effects OAL estimator approximately satisfies the usual score equations.
Therefore, we can apply standard M-estimation theory to the IPW estimator using the proposed method.

\section{Simulation Studies}\label{sec:sim}

In this section, we conduct Monte Carlo simulations to evaluate the IPW estimator based on the mixed-effects OAL.
In Section~\ref{subsec:setup}, we describe the simulation setup.
In Section~\ref{subsec:results}, we give the simulation results.

\subsection{Setup}\label{subsec:setup}

We generated the data following the simulation setups of \citep{Liu2016,Shortreed2017}.
For each data set, we generated $m=500$ groups.
The group size $n_i$ was generated as the discrete uniform distribution on the set \{2, 3, 4, 5, 6\}.
The covariate vector $X_{ij}=(X_{ij1},\ldots,X_{ij20})^\top$ was generated as $X_{ij} \sim N_{20}(0,\Sigma)$.
We considered four correlation structures: independent covariates ($\rho=0$), moderately correlated covariates ($\rho=0.2$), strongly correlated covariates ($\rho=0.5$), and highly correlated covariates ($\rho=0.75$).
We generated the treatment assignments and outcomes as follows.
\[
A_{ij} \mid X_{ij},b_i \sim Ber\{\mathrm{expit}(\sum_{k=1}^{20}\beta_k X_{ijk}+b_i)\}, \quad b_i\sim N(0,1),
\]
\[
Y_{ij}=3A_{ij}+2(n_i-1)^{-1}\sum_{k \neq j}A_{ik}+\sum_{k=1}^{20}\alpha_k X_{ijk}+\varepsilon_{ij}, \quad \varepsilon_{ij} \sim N(0,1).
\]
In this setting, the true direct effect was $\overline{\mathrm{DE}}(\theta)=3$, and the true indirect effect was $\overline{\mathrm{IE}}(\theta_1,\theta_0)=2(\theta_1-\theta_0)$.

We considered eight scenarios by varying the correlation among covariates and the strength of the association between the confounders and treatment assignment. 
The coefficient vector in the outcome model was fixed at $\alpha = (0.6,0.6,0.6,0.6,0,\ldots,0)$.
For the strong confounding setting, we set $\beta=(1,1,0,0,1,1,0,\ldots,0)$.
For the weak confounding setting, we set $\beta=(0.4,0.4,0,0,1,1,0,\ldots,0)$.
In all scenarios, $\mathcal{C}=\{1,2\}$, $\mathcal{P}=\{3,4\}$, $\mathcal{I}=\{5,6\}$ and $\mathcal{S}=\{7,\dots,20\}$.
We considered the following candidate values for the regularization
parameter: $m  \{m^{-10},\ m^{-5},\ m^{-1},\ m^{-0.75},\ m^{-0.5},m^{-0.25},\ m^{0.25},\ m^{0.49}\}$ and $\zeta$ such that $\lambda_m m^{\zeta/2-1} = m^2$ for each $\lambda_m$ value.
The regularization parameter was selected by minimizing wAMD.
We evaluated the performance of IPW estimator under several allocation strategies.
Specifically, we considered $\theta_1 \in \{0.1,0.3,0.7,0.9\}$ and $\theta_0=0.5$, and we calculated the direct and indirect effects.

The proposed method (Proposed) was compared with four approaches:
mixed-effects logistic regression models including all covariates (Full), only confounders and outcome predictors (Target), or only confounders (Conf), and an IPW estimator using the true propensity score (knownPS).
We repeated the simulation 1,000 times for each scenario.
We summarized the distributions of the effect estimates using boxplots.
To evaluate the variable selection performance, we used the proportion of times that each covariate was selected for inclusion in the propensity score model (tolerance was $10^{-8}$).

\subsection{Results}\label{subsec:results}

\begin{figure}[p]
    \centering

    \begin{subfigure}[t]{0.47\textwidth}
        \centering
        \includegraphics[width=\linewidth]{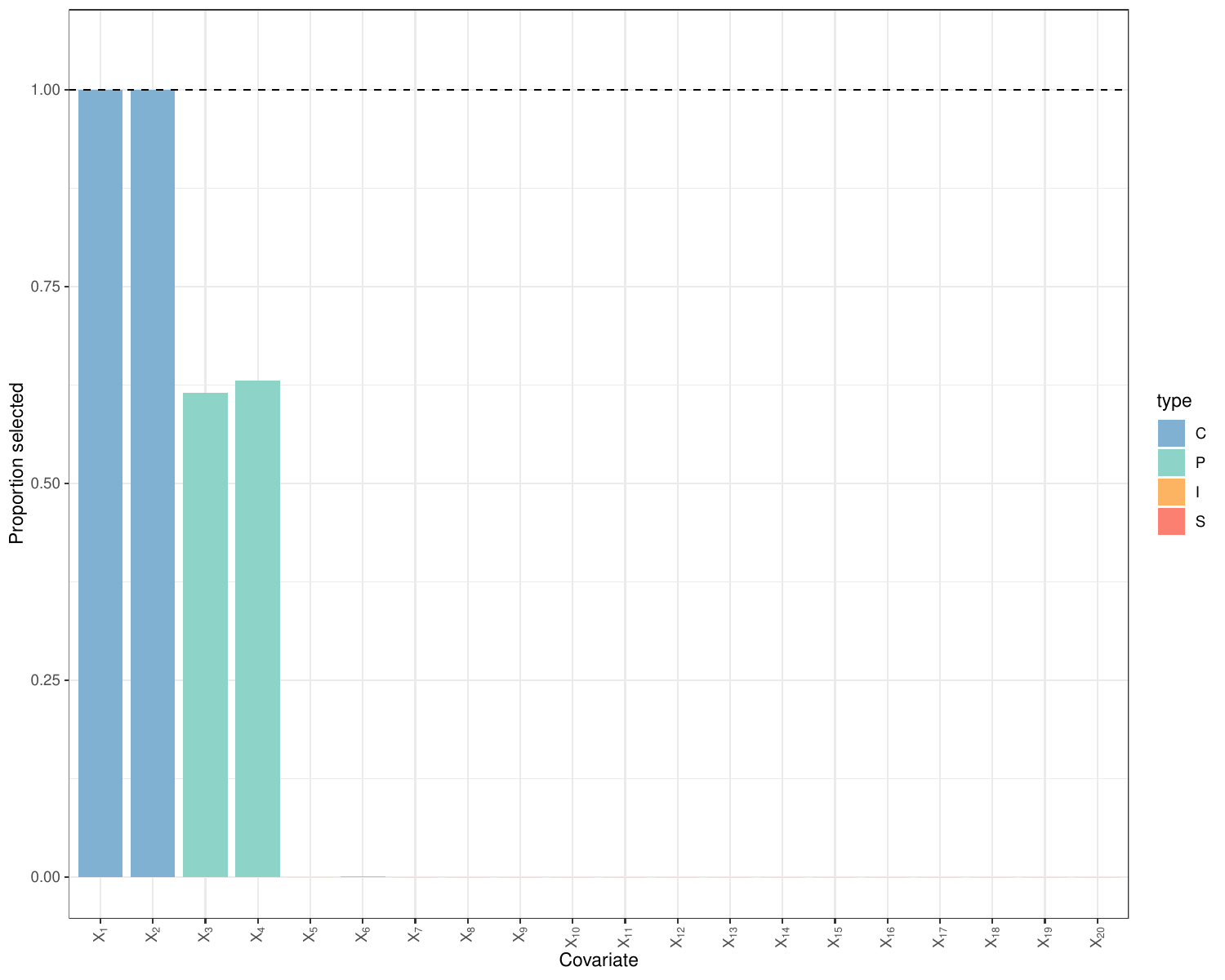}
        \caption{$\rho=0$}
    \end{subfigure}
    \hspace{0.03\textwidth}
    \begin{subfigure}[t]{0.47\textwidth}
        \centering
        \includegraphics[width=\linewidth]{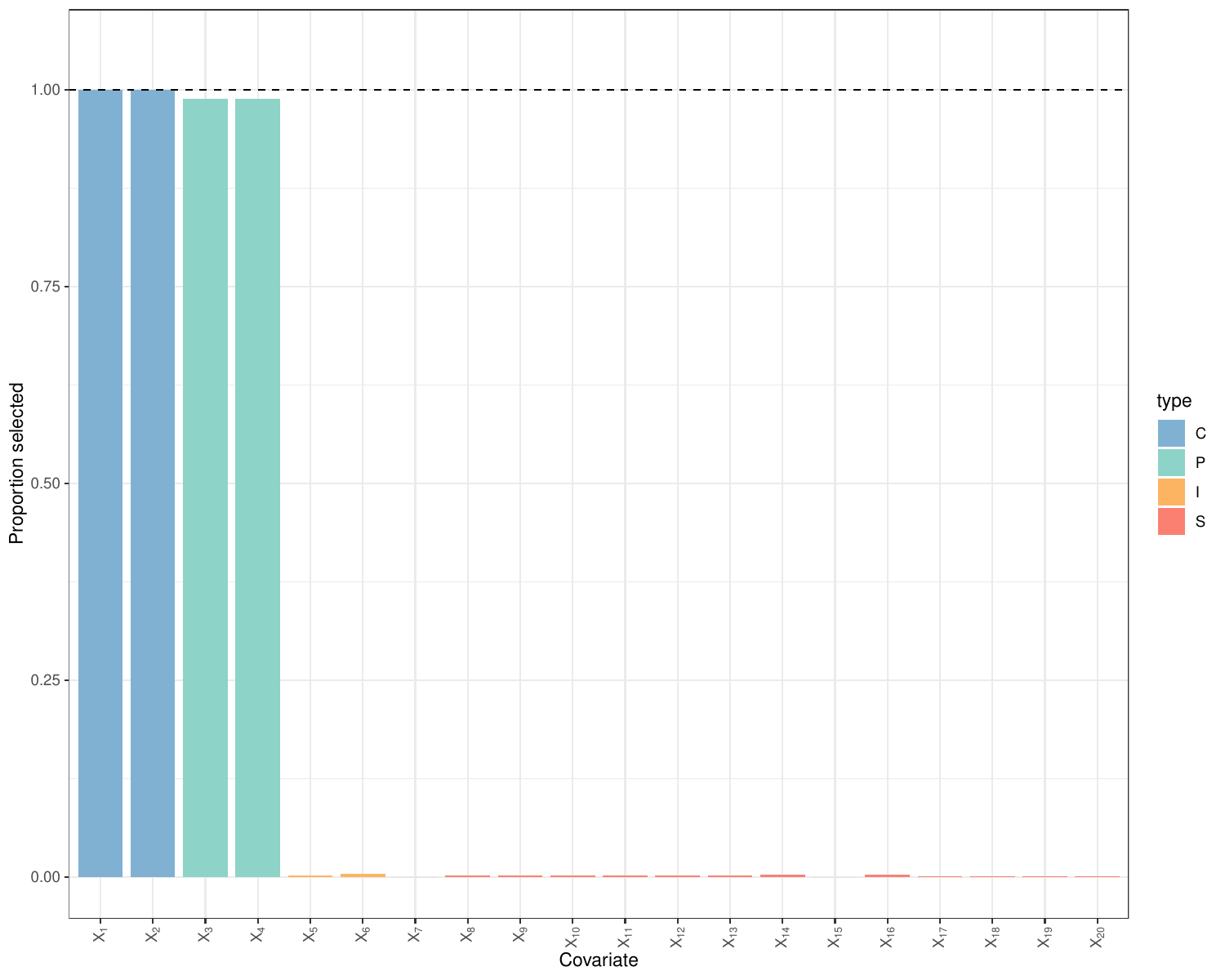}
        \caption{$\rho=0.2$}
    \end{subfigure}

    \vspace{0.3cm}

    \begin{subfigure}[t]{0.47\textwidth}
        \centering
        \includegraphics[width=\linewidth]{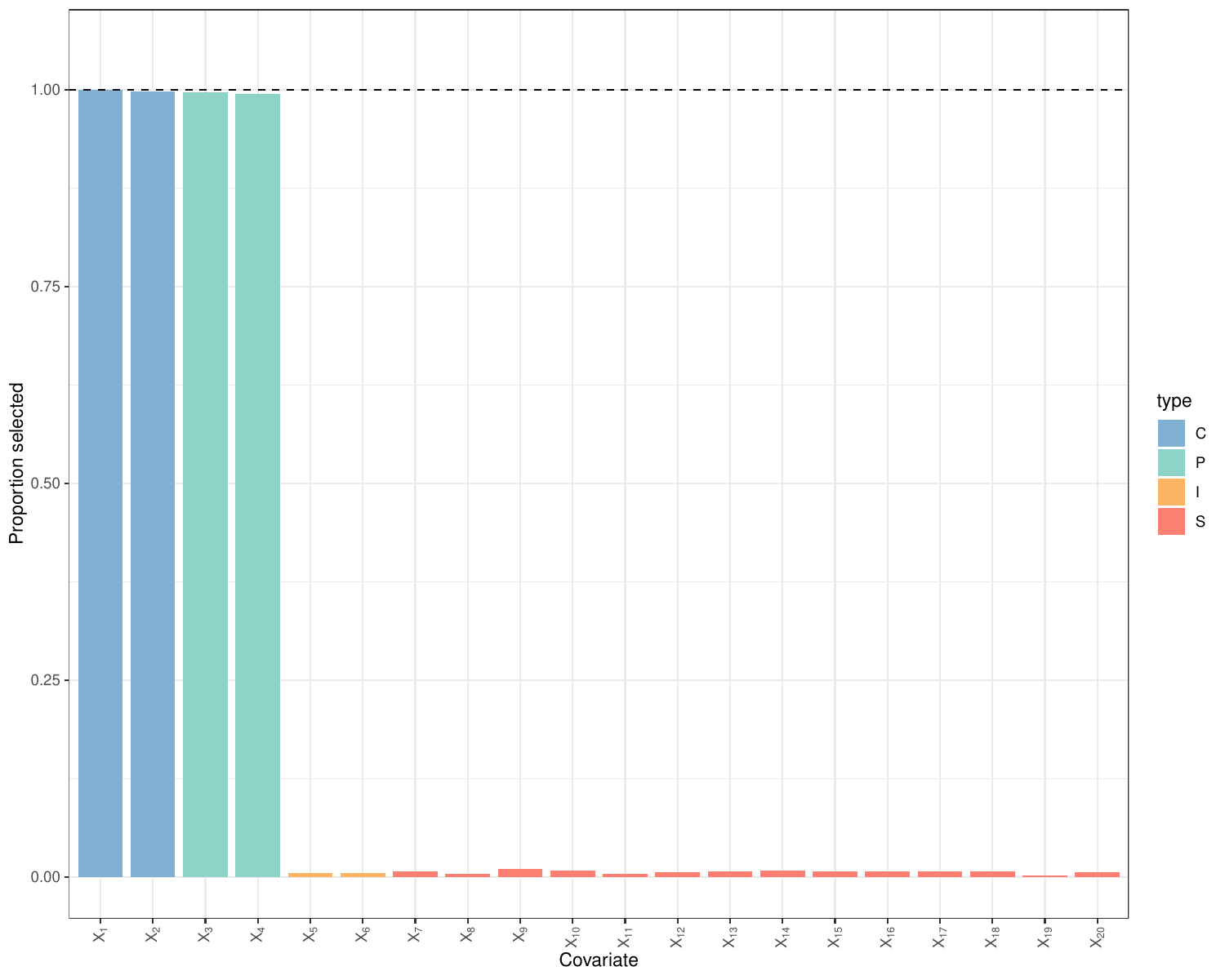}
        \caption{$\rho=0.5$}
    \end{subfigure}
    \hspace{0.03\textwidth}
    \begin{subfigure}[t]{0.47\textwidth}
        \centering
        \includegraphics[width=\linewidth]{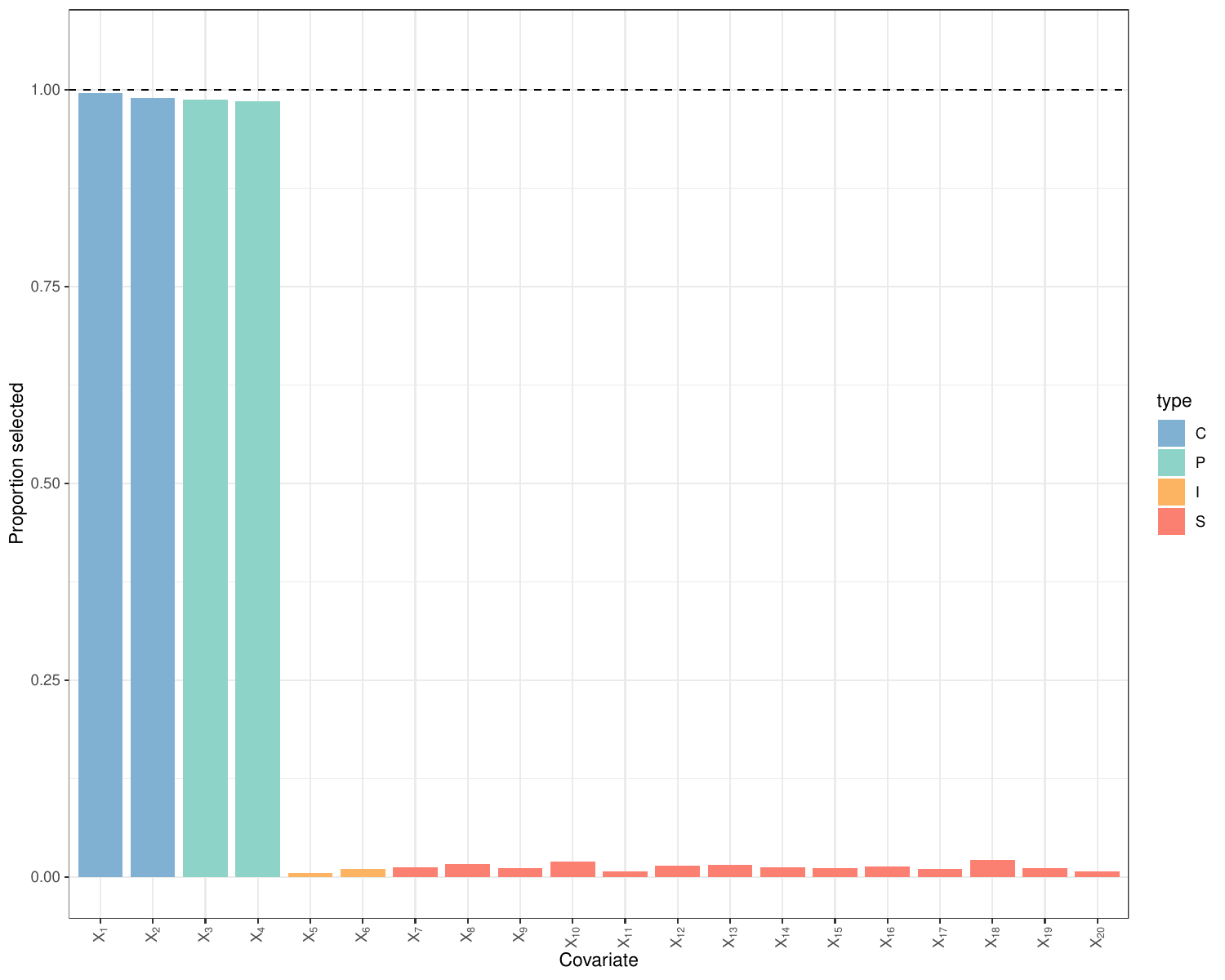}
        \caption{$\rho=0.75$}
    \end{subfigure}

    \caption{
    Variable selection frequencies under the four correlation structures in the strong confounding settings.
    Subfigures (1), (2), (3), and (4) correspond to $\rho=0$, $0.2$, $0.5$, and $0.75$, respectively.
    The bars show the proportion of times that each covariate was selected over Monte Carlo replications.
    Covariates $X_1$ and $X_2$ are confounders, $X_3$ and $X_4$ are outcome predictors, $X_5$ and $X_6$ are instrumental variables, and $X_7$--$X_{20}$ are spurious variables.
    }
    \label{fig:covariate_selection_strong}
\end{figure}

\begin{figure}[p]
    \centering

    \begin{subfigure}[t]{0.47\textwidth}
        \centering
        \includegraphics[width=\linewidth]{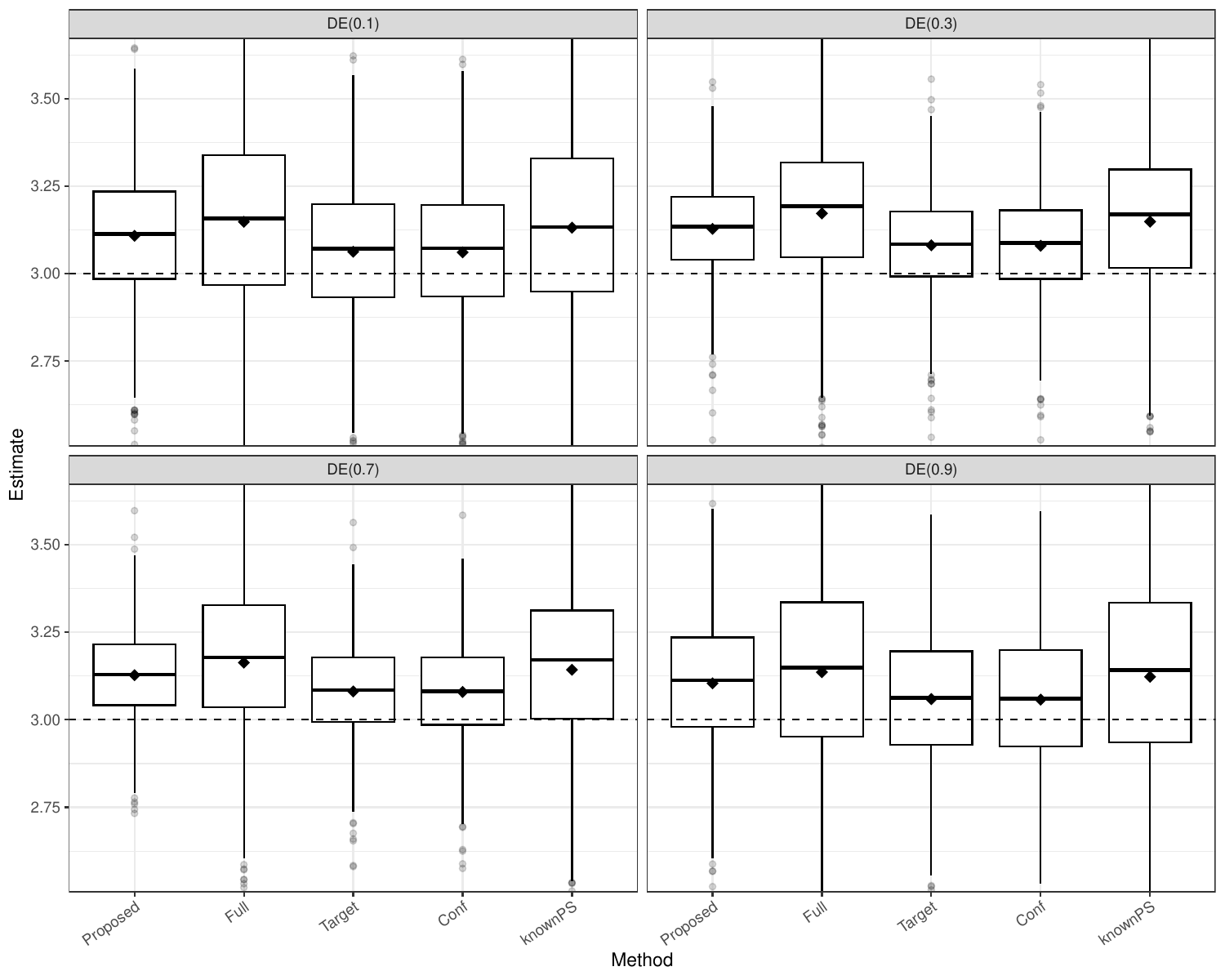}
        \caption{$\rho=0$}
    \end{subfigure}
    \hspace{0.03\textwidth}
    \begin{subfigure}[t]{0.47\textwidth}
        \centering
        \includegraphics[width=\linewidth]{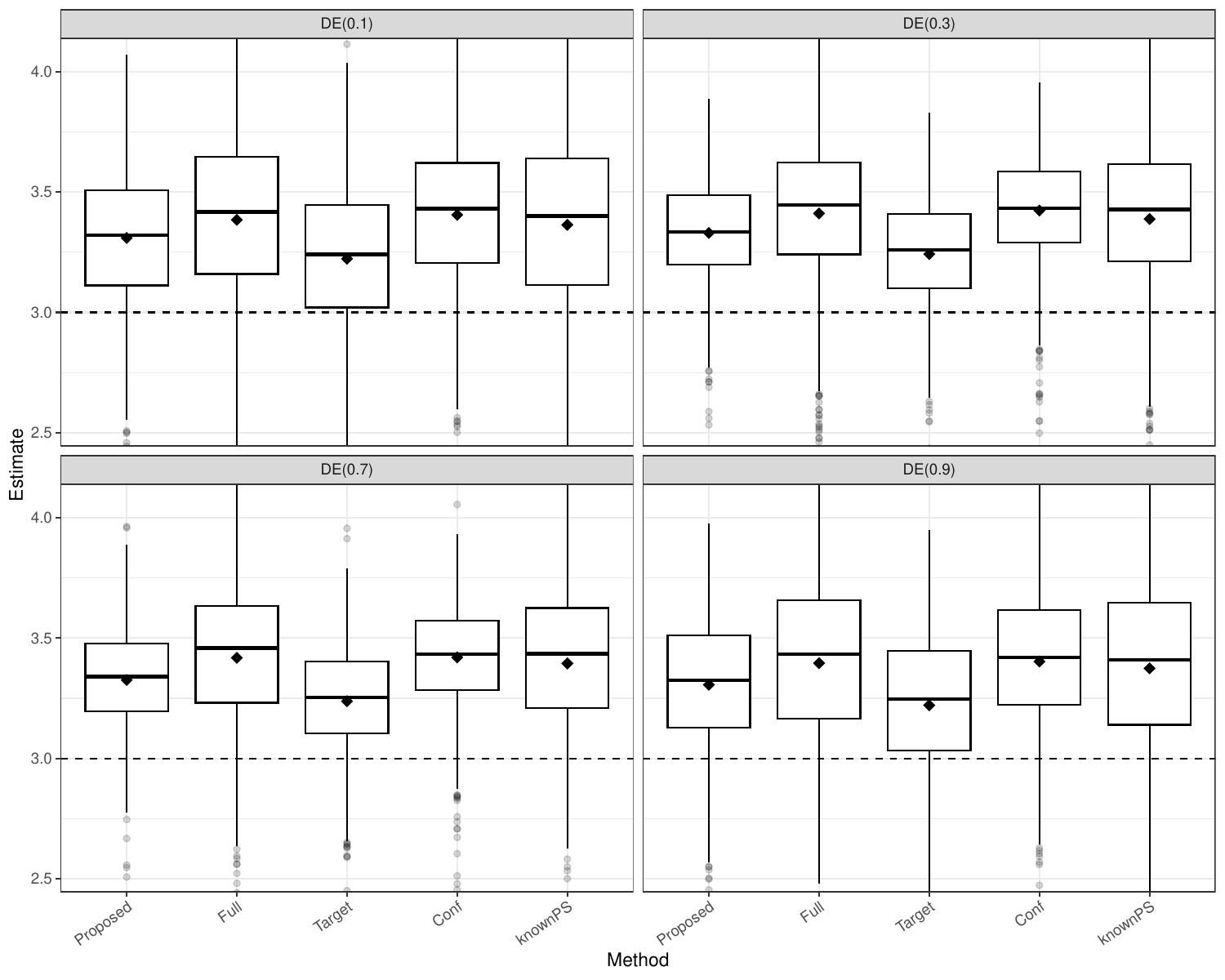}
        \caption{$\rho=0.2$}
    \end{subfigure}

    \vspace{0.3cm}

    \begin{subfigure}[t]{0.47\textwidth}
        \centering
        \includegraphics[width=\linewidth]{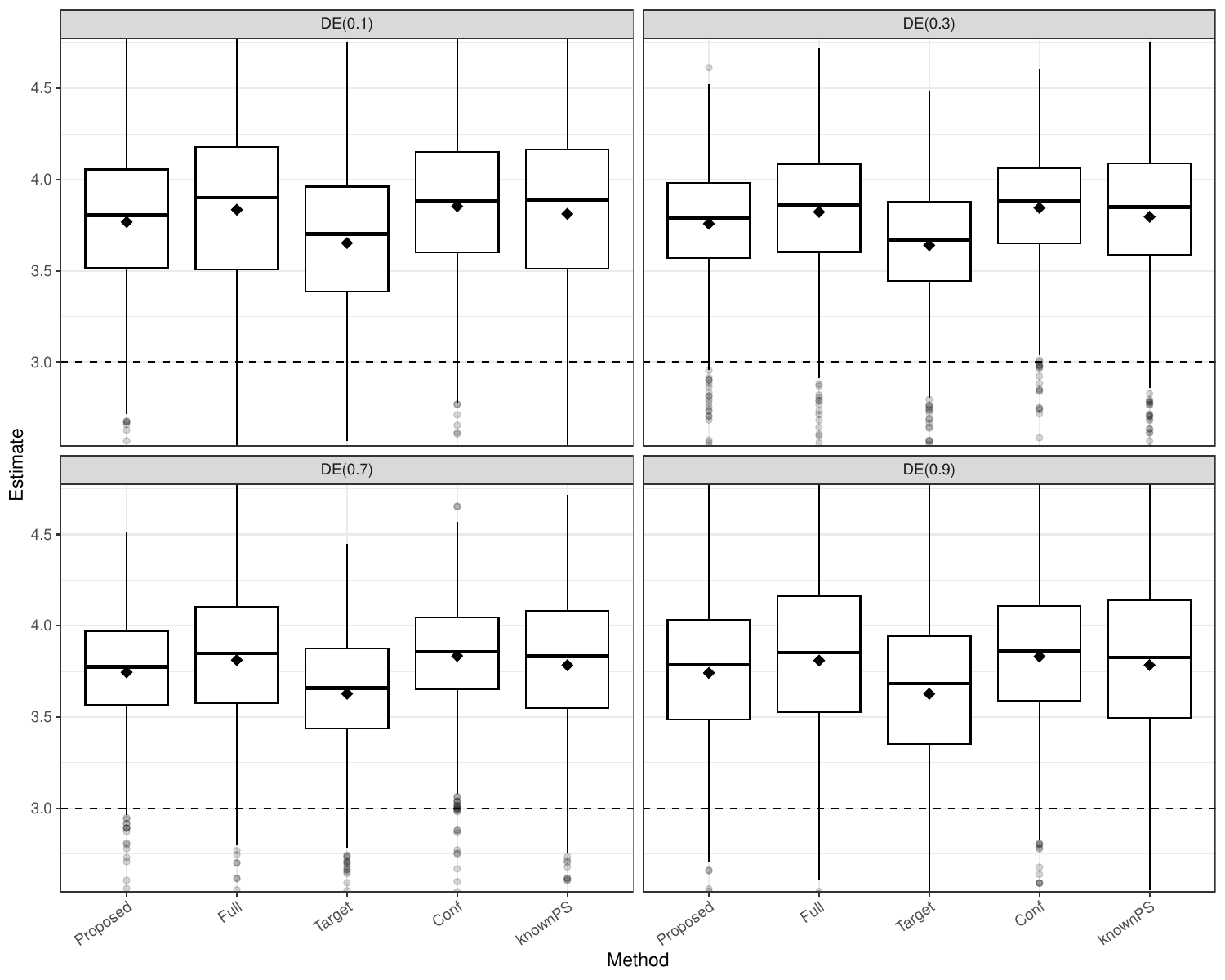}
        \caption{$\rho=0.5$}
    \end{subfigure}
    \hspace{0.03\textwidth}
    \begin{subfigure}[t]{0.47\textwidth}
        \centering
        \includegraphics[width=\linewidth]{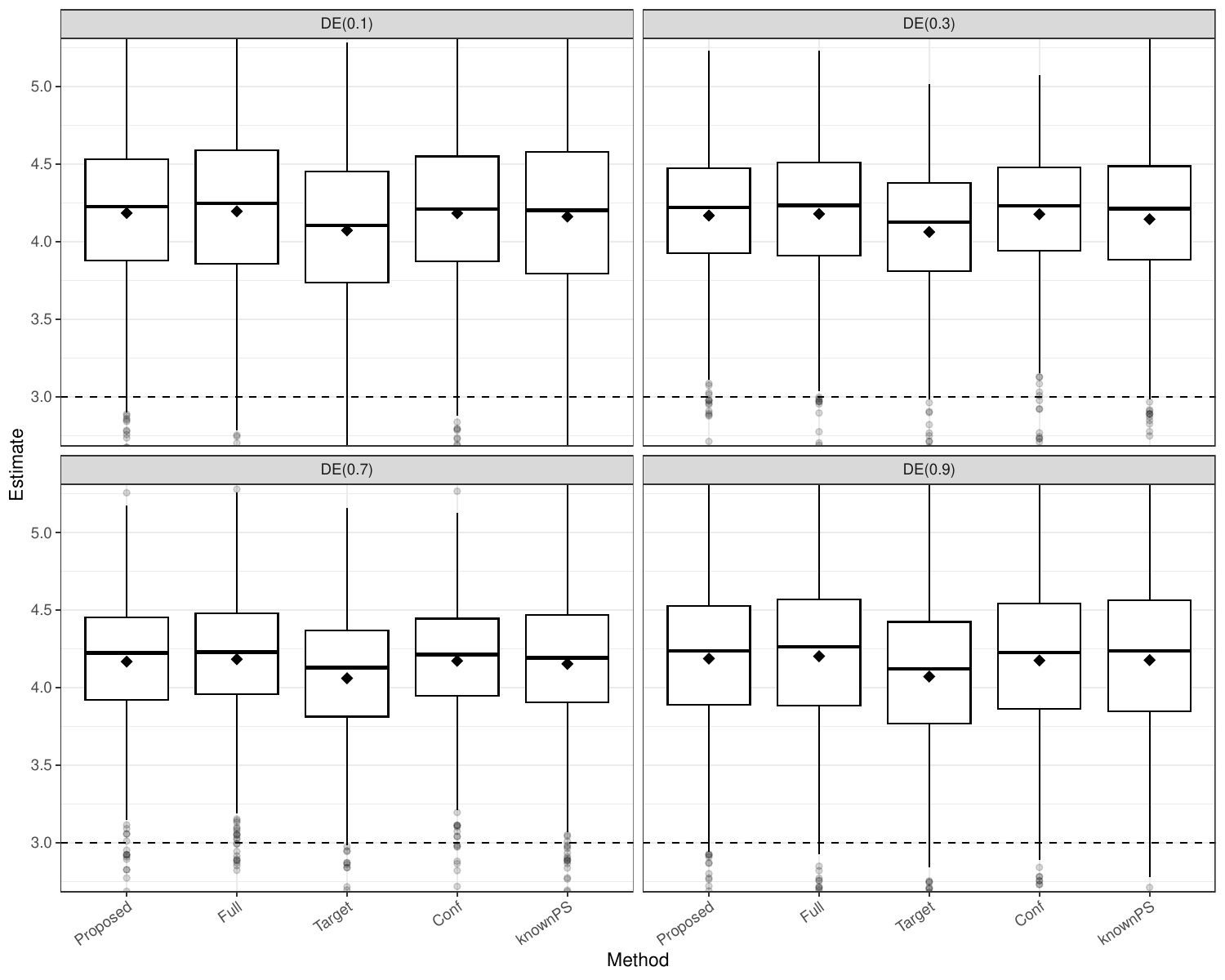}
        \caption{$\rho=0.75$}
    \end{subfigure}

    \caption{
    Boxplots of the direct effect estimates under the four correlation structures in the strong confounding settings.
    Subfigures (1), (2), (3) and (4) correspond to $\rho=0$, $0.2$, $0.5$ and $0.75$, respectively.
    Within each subfigure, the four panels correspond to $\theta=0.1$, $0.3$, $0.7$, and $0.9$, respectively.
    The dashed line represents the true direct effect, $\overline{\mathrm{DE}}(\theta)=3$.
    }
    \label{fig:direct_effect_strong}
\end{figure}

\begin{figure}[p]
    \centering

    \begin{subfigure}[t]{0.47\textwidth}
        \centering
        \includegraphics[width=\linewidth]{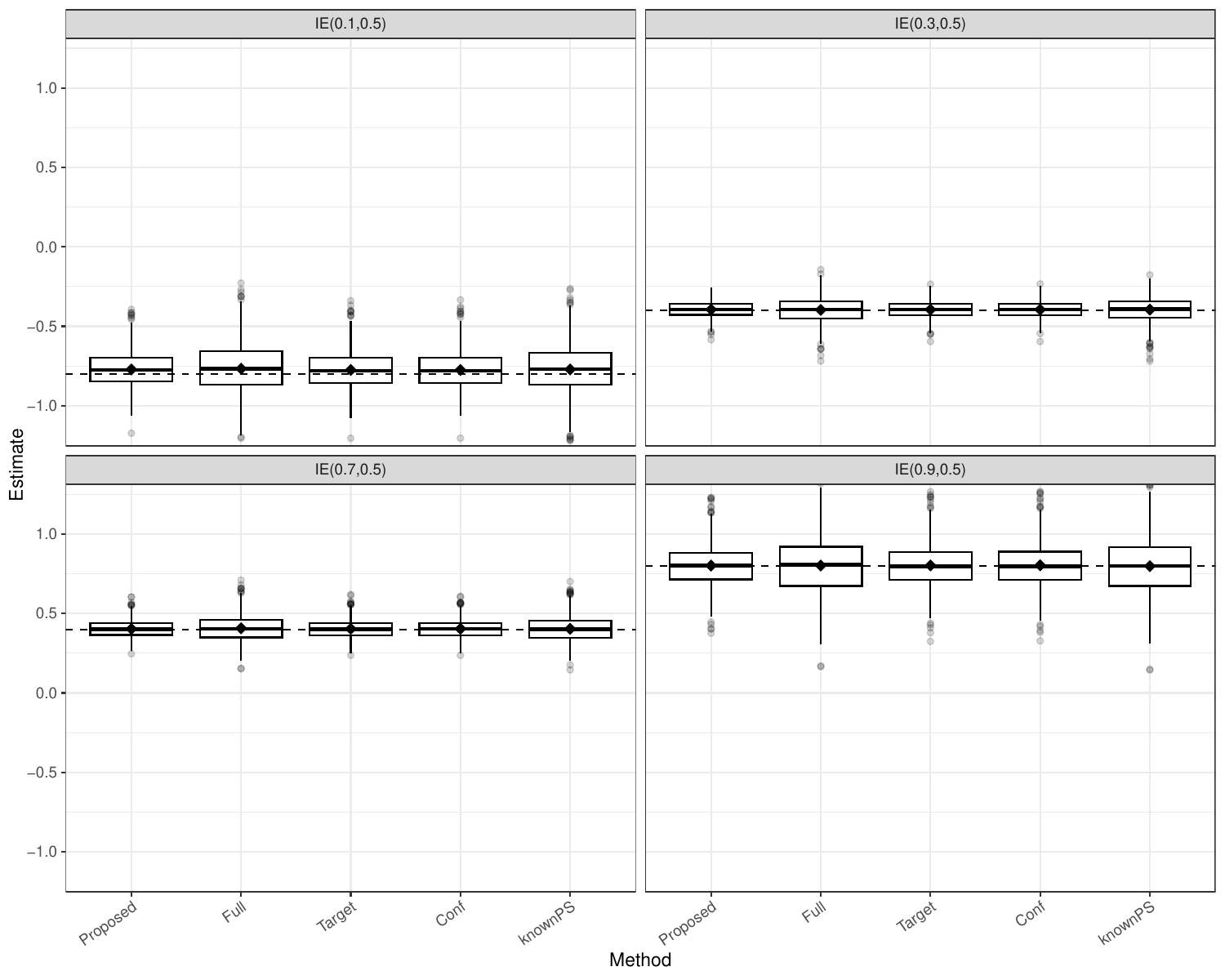}
        \caption{$\rho=0$}
    \end{subfigure}
    \hspace{0.03\textwidth}
    \begin{subfigure}[t]{0.47\textwidth}
        \centering
        \includegraphics[width=\linewidth]{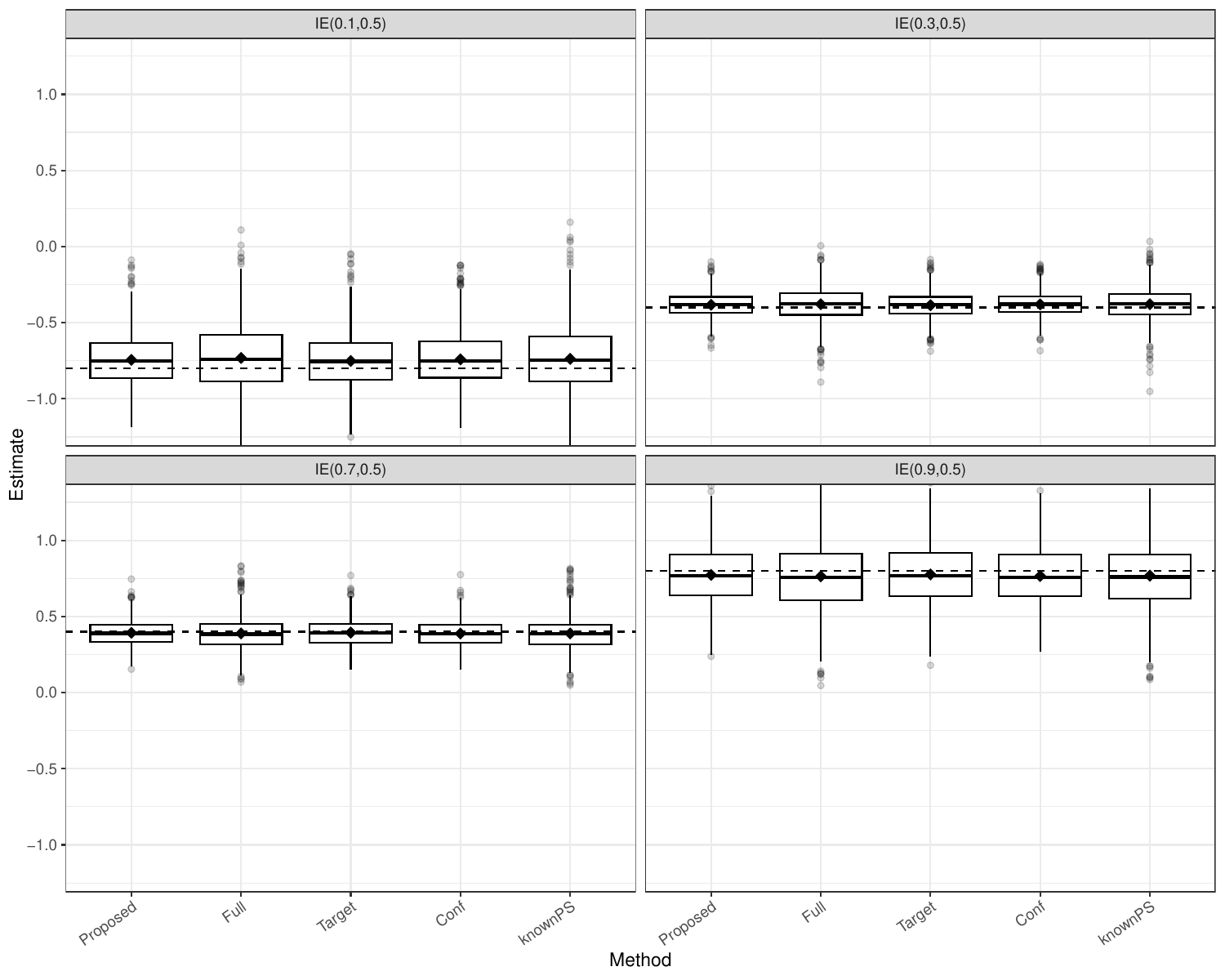}
        \caption{$\rho=0.2$}
    \end{subfigure}

    \vspace{0.3cm}

    \begin{subfigure}[t]{0.47\textwidth}
        \centering
        \includegraphics[width=\linewidth]{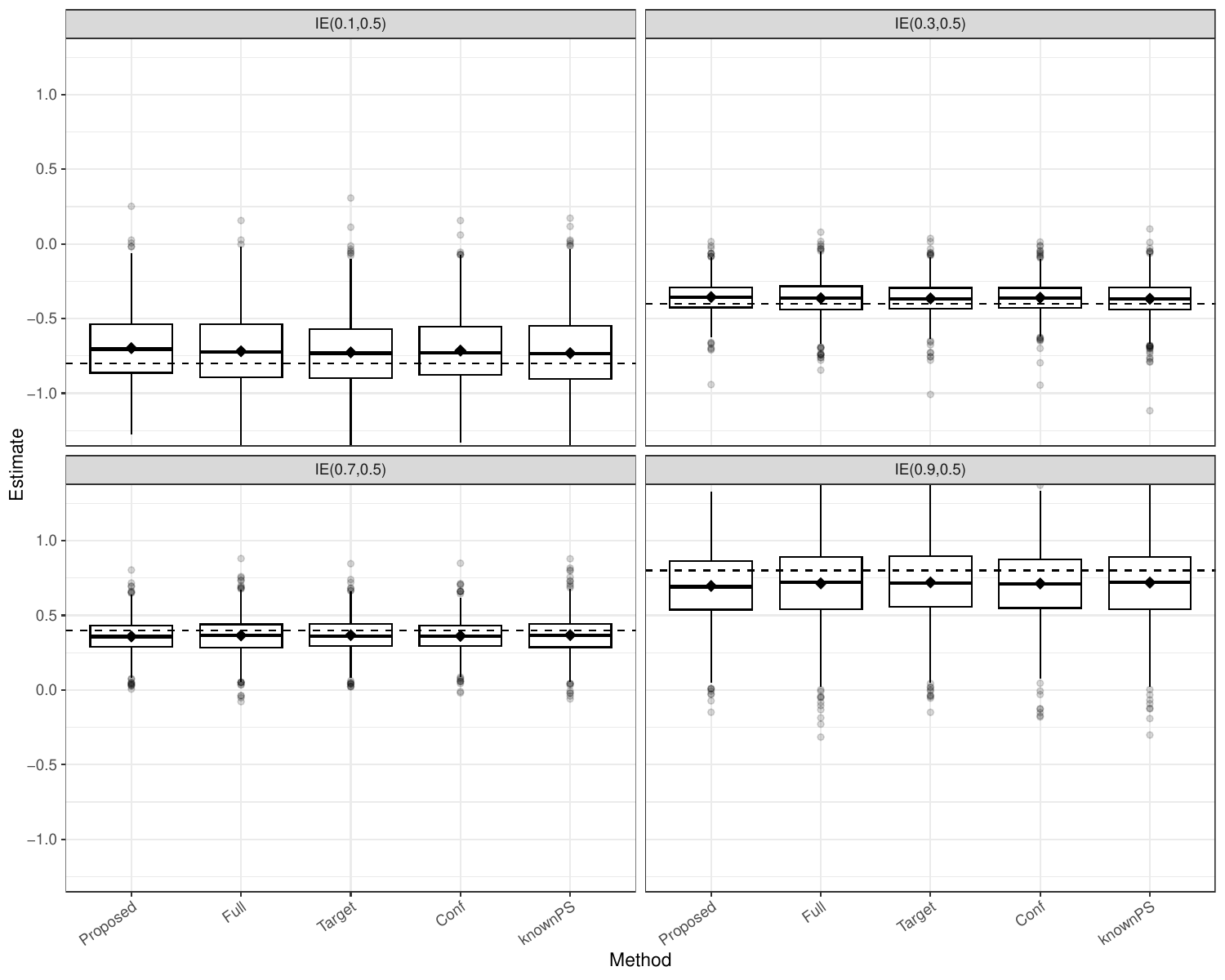}
        \caption{$\rho=0.5$}
    \end{subfigure}
    \hspace{0.03\textwidth}
    \begin{subfigure}[t]{0.47\textwidth}
        \centering
        \includegraphics[width=\linewidth]{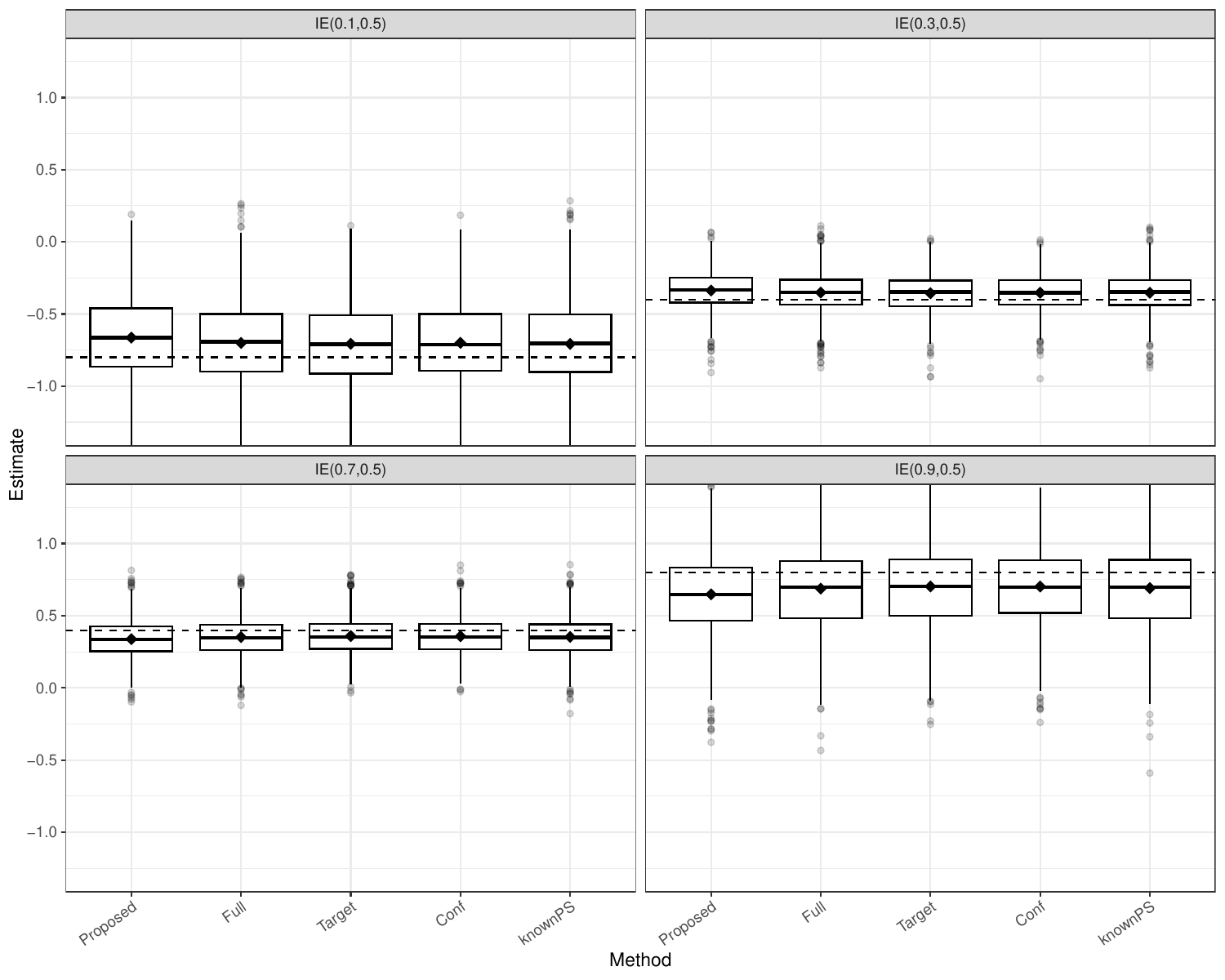}
        \caption{$\rho=0.75$}
    \end{subfigure}
    
    \caption{
    Boxplots of the indirect effect estimates under the four correlation structures in the strong confounding settings.
    Subfigures (1), (2), (3) and (4) correspond to $\rho=0$, $0.2$, $0.5$ and $0.75$, respectively.
    Within each subfigure, the four panels correspond to $\theta_1=0.1$, $0.3$, $0.7$, and $0.9$, respectively, with $\theta_0=0.5$ as the reference allocation strategy.
    The dashed lines represent the true indirect effects, $\overline{\mathrm{IE}}(0.1,0.5)=-0.8$, $\overline{\mathrm{IE}}(0.3,0.5)=-0.4$, $\overline{\mathrm{IE}}(0.7,0.5)=0.4$, and $\overline{\mathrm{IE}}(0.9,0.5)=0.8$.
    }
    \label{fig:indirect_effect_strong}
\end{figure}

Figure \ref{fig:covariate_selection_strong} shows the variable selection frequencies under the four correlation structures in the strong confounding settings.
The proposed method successfully selected confounders and outcome predictors with high probability, while excluding instrumental variables and spurious variables in most settings.
As the correlation among covariates increased, some irrelevant covariates were selected with low frequencies.

Figures \ref{fig:direct_effect_strong} and \ref{fig:indirect_effect_strong} show the distributions of the direct and indirect effect estimators, respectively.
For both the direct and indirect effects, the proposed method showed lower variability than Full and KnownPS, with results comparable to those of Targ and Conf.
For the direct effect, substantial bias was observed across all methods as the correlation among covariates increased.
This bias may partly reflect the finite-sample bias of the H\'{a}jek-type IPW estimator.
In contrast, the indirect effect estimators showed negligible bias regardless of the correlation among covariates.
The weak confounding setting showed similar results to the strong confounding setting.
The results for the weak confounding setting are provided in Appendix~\ref{app:results}.

\section{Application to the DRC-DHS Data}\label{sec:app}

We used data from the 2013--2014 Democratic Republic of the Congo Demographic and Health Survey (DRC-DHS) \citep{DHS2014}. 
The DRC-DHS is a nationally representative household survey that collected information on bed net use, malaria biomarkers, and individual and household-level characteristics.
The analysis included children aged 6–59 months.
Insecticide-treated net (ITN) use on the previous night was treated as the exposure, and malaria positivity based on blood smear testing was treated as the outcome.
After excluding individuals with missing values, the final analytic sample comprised 7,598 children, including 3,510 ITN non-users and 4,088 ITN users.
The DHS survey cluster number (\texttt{hv001}) was used to define the partial interference groups, resulting in 531 groups with a median group size of 14 (IQR:11--17).

A total of 39 candidate covariates was considered, including age, sex, weight, height, place of residence, maternal education, and others.
Propensity scores were estimated using the proposed method and the mixed-effects logistic regression model including all candidate covariates (Full model).
Under Bernoulli allocation strategies with $\theta$ ranging from 0.1 to 0.9, we estimated the direct effect (DE), indirect effect (IE), total effect (TE), and overall effect (OE), using $\theta=0.5$ as the reference allocation strategy for the IE, TE, and OE. 
Pointwise 95\% confidence intervals were calculated by sandwich variance estimators that accounted for uncertainty in propensity score estimation.
Bootstrap sampling was used only to assess the stability of covariate selection, and selection frequencies were calculated from 1,000 bootstrap samples, with a covariate regarded as selected when the absolute value of its estimated coefficient exceeded $10^{-8}$.

\begingroup
\footnotesize
\setlength{\tabcolsep}{2.5pt}
\renewcommand{\arraystretch}{0.9}

\begin{table}[htbp]
\centering
\caption{
Covariate distributions according to ITN use and the selection probabilities.
In the ITN use columns, continuous variables are presented as mean (standard deviation), and categorical variables as number (\%).
The last column reports the percentage of 1,000 bootstrap samples in which each covariate was selected for the propensity score model.
Ref denotes the reference category.
}
\label{tab:realdata_summary}
\begin{tabular}{
    @{}
    p{0.46\textwidth}
    >{\centering\arraybackslash}p{0.17\textwidth}
    >{\centering\arraybackslash}p{0.17\textwidth}
    >{\centering\arraybackslash}p{0.12\textwidth}
    @{}
}
\toprule
& \multicolumn{2}{c}{ITN use} & \\
\cmidrule(lr){2-3}
Covariates
& Non-user
& User
& Selected (\%) \\
& $N=3,510$
& $N=4,088$
& \\
\midrule
Age (months)
    & 34.48 (15.44)
    & 31.16 (15.57)
    & 5.0 \\

\textbf{Sex}
    & & & \\
\quad Male
    & 1,744 (49.7)
    & 2,017 (49.3)
    & Ref \\
\quad Female
    & 1,766 (50.3)
    & 2,071 (50.7)
    & 0.0 \\

Weight (kg)
    & 11.75 (2.80)
    & 11.34 (2.79)
    & 14.5 \\

Height (cm)
    & 85.92 (10.70)
    & 84.22 (10.91)
    & 84.2 \\

Height/Age standard deviation
    & \(-1.91\) (1.80)
    & \(-1.75\) (1.76)
    & 57.0 \\

Weight/Age standard deviation
    & \(-1.25\) (1.26)
    & \(-1.17\) (1.24)
    & 92.4 \\

Weight/Height standard deviation
    & \(-0.21\) (1.27)
    & \(-0.23\) (1.25)
    & 59.3 \\

BMI/Age standard deviation
    & 0.00 (1.33)
    & \(-0.03\) (1.32)
    & 19.2 \\

\textbf{Place of residence}
    & & & \\
\quad Capital, large city
    & 466 (13.3)
    & 473 (11.6)
    & Ref \\
\quad Small city
    & 112 (3.2)
    & 160 (3.9)
    & 0.0 \\
\quad Town
    & 426 (12.1)
    & 583 (14.3)
    & 0.0 \\
\quad Countryside
    & 2,506 (71.4)
    & 2,872 (70.3)
    & 8.5 \\

\textbf{Woman's highest educational level}
    & & & \\
\quad No education
    & 774 (22.1)
    & 757 (18.5)
    & Ref \\
\quad Primary
    & 1,439 (41.0)
    & 1,675 (41.0)
    & 0.0 \\
\quad Secondary
    & 888 (25.3)
    & 1,338 (32.7)
    & 0.0 \\
\quad Higher
    & 20 (0.6)
    & 47 (1.1)
    & 0.0 \\
\quad Don't know
    & 389 (11.1)
    & 271 (6.6)
    & 0.0 \\
\bottomrule
\end{tabular}
\end{table}
\endgroup

Table~\ref{tab:realdata_summary} summarizes the selected covariate distributions and selection frequencies; the results for the other covariates are provided in the Appendix~\ref{app:results}.
When the proposed method was applied to the full data set, Height (cm) and Weight-for-Age Z-score (WAZ) were selected.
In the bootstrap analysis, several anthropometric measures were selected relatively frequently, including Height (84.2\%), Height-for-Age Z-score (57.0\%), WAZ (92.4\%), and Weight-for-Height Z-score (59.3\%).
In contrast, child age was selected infrequently (5.0\%), although previous studies have identified age as an important predictor of childhood malaria infection in the DRC \citep{Emina2021,Ma2017}.
This may partly reflect correlations between age and anthropometric measures, particularly height ($\rho=0.83$) and weight ($\rho=0.77$), which were selected more frequently.
Maternal education and place of residence also showed low selection frequencies, although previous studies have identified these factors as predictors of childhood malaria infection in the DRC \citep{Ma2017,Emina2021}.

\begin{figure}[t]
    \centering
    \includegraphics[width=0.88\textwidth]{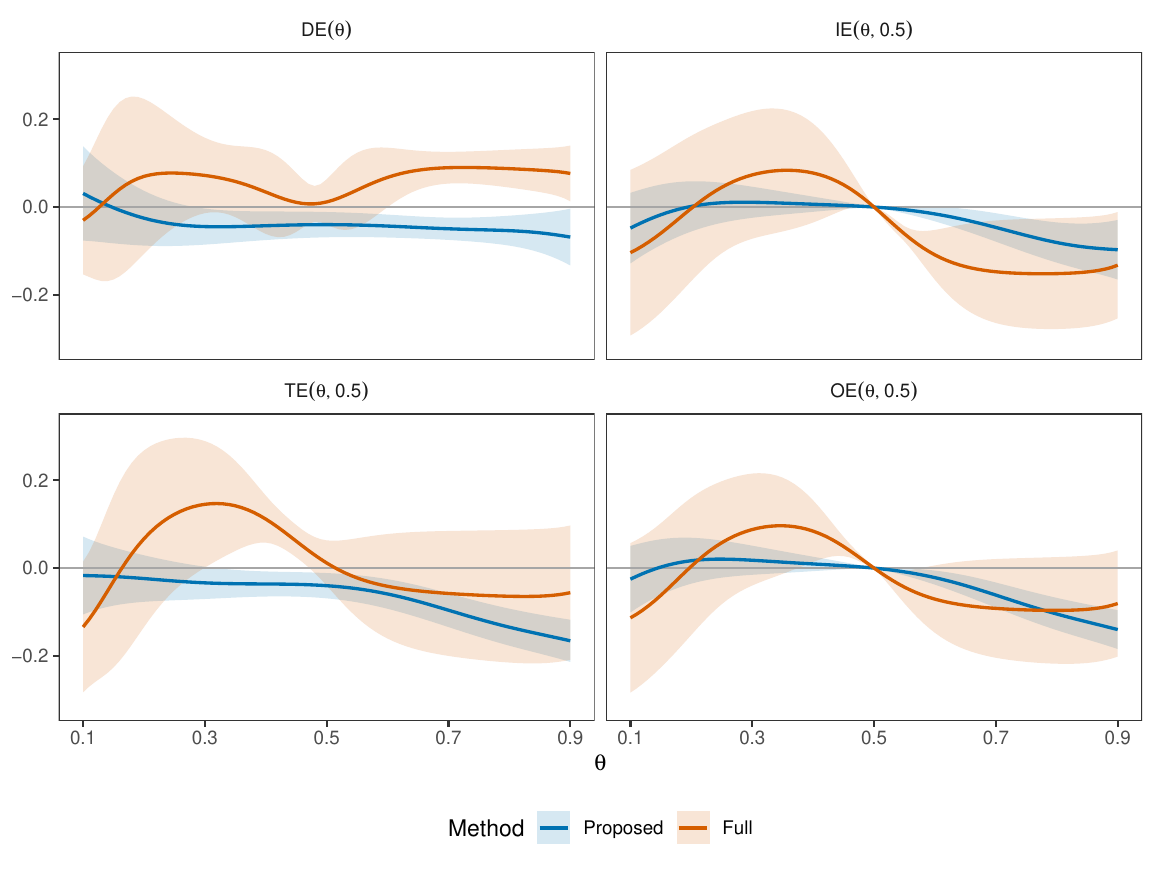}
    \caption{
    Estimated direct, indirect, total, and overall effects of ITN use on malaria using the full sample.
    The solid lines represent the point estimates, and the shaded regions represent pointwise 95\% confidence intervals based on the sandwich variance estimator.
    For the indirect, total, and overall effects, $\theta=0.5$ is the reference allocation strategy.
}
    \label{fig:6.1}
\end{figure}

Figure~\ref{fig:6.1} shows the results of the DE, IE, TE, and OE estimation.
For the proposed method, the DE estimate at $\theta=0.5$ was $-0.040$ (95\% CI: $-0.069,-0.011$).
In contrast, for the Full model, the DE estimate at $\theta=0.5$ was $0.012$ (95\% CI: $-0.040,0.063$), with the confidence interval including zero.
Results for other specific allocation strategies are provided in Table~\ref{tab:realdata_effects}.
Overall, for DE, the proposed method estimated negative effects over a broad range of allocation strategies, whereas the Full model generally estimated positive effects.
For the IE, TE, and OE, the proposed method consistently estimated negative effects with confidence intervals excluding zero at higher allocation strategies, whereas the Full model also estimated negative effects but its confidence intervals often included zero.
These results suggest that variable selection can contribute to more stable IPW estimation under partial interference by excluding unnecessary covariates in the propensity score model.

\section{Discussion and Conclusions}\label{sec:conc}

In this study, we proposed the Outcome-Adaptive Lasso using the mixed-effects logistic regression model, a variable selection method for propensity score models under partial interference.
The proposed method simultaneously performs variable selection and estimation while accounting for heterogeneity through the random effects.
Under regularity conditions, we established the local oracle property of the proposed method, including consistency in variable selection and asymptotic normality, as well as the consistency and asymptotic normality of the IPW estimators.

The simulation results showed that the proposed method frequently selected confounders and outcome predictors while generally excluding instrumental and spurious variables.
Compared with the model including all candidate covariates, the proposed method generally reduced the variability of IPW estimators.
However, substantial bias was observed in the direct effect estimators as the correlation among covariates increased.
In the DRC-DHS application, the proposed method selected several anthropometric measures and produced causal effect estimates and confidence intervals that differed from those obtained using the full model.

Several limitations and directions remain for future research.
First, we did not account for the uncertainty arising from model selection when constructing confidence intervals in our real-data application.
Statistical inference that accounts for this uncertainty is known as post-selection inference \citep{Kuchibhotla2022}.
Second, the proposed method relies on an outcome regression model to construct the adaptive weights, and its variable selection performance can be affected by misspecification of the outcome model.
Moreover, when a correctly specified outcome model is available, augmented IPW (AIPW) may be preferable to IPW because it can attain the semiparametric efficiency bound if the propensity score model is also correctly specified.
In this context, Cho and Yang~\citep{Cho2025} proposed a conservative variable selection approach that separately selects covariates for the propensity score and outcome models and uses their union to help preserve the double robustness of AIPW estimators.
Third, extending the proposed method to the generalized Outcome-Adaptive Lasso \citep{Balde2023} can improve variable selection under correlated covariates.
However, since the R package \texttt{glmmLasso} does not support elastic-net-type penalties, further modifications to its algorithm would be required.

\section*{Acknowledgments}
S. N. was supported by JGMI of Kyushu University.
A. O. was supported by JSPS KAKENHI Grant Number JP25K24377.
S. K. was supported by JSPS KAKENHI Grant Numbers JP23K11008, JP23H00809, and JP23H03352.

\section*{Data Availability Statement}

The data used in this study are from the 2013--2014 Democratic Republic of the Congo Demographic and Health Survey (DHS).
The DHS datasets are available upon registration and approval through the DHS Program website (\url{https://dhsprogram.com}).

\bibliographystyle{vancouver}
\bibliography{references}

@article{Balde2023,
  author  = {Bald{\'e}, Isma{\"i}la and Yang, Yi Archer and Lefebvre, Genevi{\`e}ve},
  title   = {Reader reaction to ``Outcome-adaptive lasso: Variable selection for causal inference'' by Shortreed and Ertefaie (2017)},
  journal = {Biometrics},
  volume  = {79},
  number  = {1},
  pages   = {514--520},
  year    = {2023},
  doi     = {10.1111/biom.13683}
}

@article{Balde2025,
  title   = {The oracle property of the generalized outcome-adaptive lasso},
  author  = {Bald{\'e}, Ismaila},
  journal = {Stat Probabil Lett},
  volume  = {221},
  pages   = {110379},
  year    = {2025},
  doi     = {10.1016/j.spl.2025.110379}
}

@inproceedings{Bhattacharya2020,
  author    = {Bhattacharya, Rohit and Malinsky, Daniel and Shpitser, Ilya},
  title     = {Causal Inference Under Interference And Network Uncertainty},
  booktitle = {Proceedings of the 35th Conference on Uncertainty in Artificial Intelligence},
  series    = {Proceedings of Machine Learning Research},
  volume    = {115},
  pages     = {1028--1038},
  year      = {2020},
  publisher = {PMLR},
  url       = {https://proceedings.mlr.press/v115/bhattacharya20a.html}
}

@article{Brookhart2006,
  author  = {Brookhart, M. Alan and Schneeweiss, Sebastian and Rothman, Kenneth J. and Glynn, Robert J. and Avorn, Jerry and St{\"u}rmer, Til},
  title   = {Variable selection for propensity score models},
  journal = {Am J Epidemiol},
  year    = {2006},
  volume  = {163},
  number  = {12},
  pages   = {1149--1156}
}

@article{Cho2025,
  author  = {Cho, Eunah and Yang, Shu},
  title   = {Variable Selection for Doubly Robust Causal Inference},
  journal = {Stat Interface},
  volume  = {18},
  number  = {1},
  pages   = {93--105},
  year    = {2025},
  doi     = {10.4310/SII.241023040813}
}

@book{DHS2014,
  author = {{DHS}},
  title = {{R\'epublique D\'emocratique du Congo
            Enqu\^ete D\'emographique et de Sant\'e
            (EDS-RDC) 2013--2014 [Dataset]}},
  year = {2014},
  address = {Rockville, Maryland, USA},
  publisher = {Minist\`ere du Plan et Suivi de la Mise
               en oeuvre de la R\'evolution de la
               Modernit\'e (MPSMRM),
               Minist\`ere de la Sant\'e Publique (MSP),
               and ICF International},
  note = {CDPR61FL}
}

@article{Emina2021,
  title={Profiling malaria infection among under-five children in the Democratic Republic of Congo},
  author={Emina, Jacques B. O. and Doctor, Henry V. and Ye, Yazoume},
  journal={PLoS ONE},
  volume={16},
  number={5},
  pages={e0250550},
  year={2021},
  doi={10.1371/journal.pone.0250550}
}

@article{Fan2001,
  author  = {Fan, Jianqing and Li, Runze},
  title   = {Variable Selection via Nonconcave Penalized Likelihood and Its Oracle Properties},
  journal = {J Am Stat Assoc},
  year    = {2001},
  volume  = {96},
  number  = {456},
  pages   = {1348--1360},
  doi     = {10.1198/016214501753382273}
}

@article{Groll2014,
  author  = {Groll, Andreas and Tutz, Gerhard},
  title   = {Variable Selection for Generalized Linear Mixed Models by L1-Penalized Estimation},
  journal = {Stat Comput},
  year    = {2014},
  volume  = {24},
  number  = {2},
  pages   = {137--154},
  doi     = {10.1007/s11222-012-9359-z}
}

@article{Halloran1995,
  author  = {Halloran, M. Elizabeth and Struchiner, Claudio J.},
  title   = {Causal inference in infectious diseases},
  journal = {Epidemiology},
  volume  = {6},
  pages   = {142--151},
  year    = {1995}
}

@incollection{Hajek1971,
  author    = {H{\'a}jek, Jaroslav},
  title     = {Comment on {``An Essay on the Logical Foundations of Survey Sampling''} by {D}. {B}asu},
  booktitle = {Foundations of Statistical Inference},
  editor    = {Godambe, V. P. and Sprott, D. A.},
  publisher = {Holt, Rinehart and Winston},
  year      = {1971},
  pages     = {236}
}

@article{Hong2006,
  author  = {Hong, Guanglei and Raudenbush, Stephen W.},
  title   = {Evaluating kindergarten retention policy: A case study of causal inference for multilevel observational data},
  journal = {J Am Stat Assoc},
  volume  = {101},
  pages   = {901--910},
  year    = {2006}
}

@article{Hoshino2024,
  author  = {Hoshino, Tadao and Yanagi, Takahide},
  title   = {Causal Inference with Noncompliance and Unknown Interference},
  journal = {J Am Stat Assoc},
  year    = {2024},
  volume  = {119},
  number  = {548},
  pages   = {2869--2880},
  doi     = {10.1080/01621459.2023.2284413}
}

@article{Hudgens2008,
  author  = {Hudgens, M. G. and Halloran, M. Elizabeth},
  title   = {Toward causal inference with interference},
  journal = {J Am Stat Assoc},
  volume  = {103},
  pages   = {832--842},
  year    = {2008}
}

@article{Hu2022,
  author  = {Hu, Yuchen and Li, Shuangning and Wager, Stefan},
  title   = {Average Direct and Indirect Causal Effects under Interference},
  journal = {Biometrika},
  year    = {2022},
  volume  = {109},
  number  = {4},
  pages   = {1165--1172},
  doi     = {10.1093/biomet/asac008}
}

@article{Imai2021,
  author  = {Imai, Kosuke and Jiang, Zhichao and Malani, Anup},
  title   = {Causal Inference With Interference and Noncompliance in Two-Stage Randomized Experiments},
  journal = {J Am Stat Assoc},
  year    = {2021},
  volume  = {116},
  number  = {534},
  pages   = {632--644}
}

@article{Kilpatrick2025,
  author  = {Kilpatrick, Kayla W. and Saul, Bradley C. and Hudgens, Michael},
  title   = {G-Estimation With Partial Interference},
  journal = {Stat},
  year    = {2025},
  volume  = {14},
  number  = {1},
  pages   = {e70059},
  doi     = {10.1002/sta4.70059}
}

@article{Kuchibhotla2022,
  author  = {Kuchibhotla, Arun K. and Kolassa, John E. and Kuffner, Todd A.},
  title   = {Post-Selection Inference},
  journal = {Annu Rev Stat Appl},
  year    = {2022},
  volume  = {9},
  pages   = {505--527},
  doi     = {10.1146/annurev-statistics-100421-044639}
}

@article{Liu2016,
  title   = {On inverse probability-weighted estimators in the presence of interference},
  author = {Liu, L. and Hudgens, M. G. and Becker-Dreps, S.},
  journal = {Biometrika},
  year    = {2016},
  volume  = {103},
  number  = {4},
  pages   = {829--842}
}

@article{Liu2019,
  author  = {Liu, Lan and Hudgens, Michael G. and Saul, Bradley and Clemens, John D. and Ali, Mohammad and Emch, Michael E.},
  title   = {Doubly Robust Estimation in Observational Studies with Partial Interference},
  journal = {Stat},
  year    = {2019},
  volume   = {8},
  number   = {1},
  pages    = {e214},
  doi      = {10.1002/sta4.214},
  url      = {https://doi.org/10.1002/sta4.214}
}

@article{Lunceford2004,
  author  = {Lunceford, Jared K. and Davidian, Marie},
  title   = {Stratification and weighting via the propensity score in estimation of causal treatment effects: A comparative study},
  journal = {Stat Med},
  year    = {2004},
  volume  = {23},
  number  = {19},
  pages   = {2937--2960},
  doi     = {10.1002/sim.1903}
}

@article{Ma2017,
  title={Is maternal education a social vaccine for childhood malaria infection? A cross-sectional study from war-torn Democratic Republic of Congo},
  author={Ma, Cary and Claude, Kasereka Masumbuko
        and Kibendelwa, Zacharie Tsongo
        and Brooks, Hannah and Zheng, Xiaonan
        and Hawkes, Michael},
  journal={Pathog Glob Health},
  volume={111},
  number={2},
  pages={98--106},
  year={2017},
  publisher={Taylor \& Francis},
  doi={10.1080/20477724.2017.1288971}
}

@article{Perez2014,
  author  = {Perez-Heydrich, Carolina and Hudgens, Michael G. and Halloran, M. Elizabeth and Clemens, John D. and Ali, Mohammad and Emch, Michael E.},
  title   = {Assessing Effects of Cholera Vaccination in the Presence of Interference},
  journal = {Biometrics},
  year    = {2014},
  month   = sep,
  volume  = {70},
  number  = {3},
  pages   = {731--741},
  doi     = {10.1111/biom.12184}
}

@article{Rosenbaum2007,
  author  = {Rosenbaum, Paul R.},
  title   = {Interference Between Units in Randomized Experiments},
  journal = {J Am Stat Assoc},
  year    = {2007},
  volume  = {102},
  number  = {477},
  pages   = {191--200}
}

@article{Rotnitzky2010,
  author  = {Rotnitzky, Andrea and Li, Lingling and Li, Xiaochun},
  title   = {A note on overadjustment in inverse probability weighted estimation},
  journal = {Biometrika},
  year    = {2010},
  volume  = {97},
  number  = {4},
  pages   = {997--1001}
}

@article{Rubin1980,
  author  = {Rubin, Donald B.},
  title   = {Comment on ``Randomization Analysis of
             Experimental Data: The Fisher Randomization
             Test'' by D. Basu},
  journal = {J Am Stat Assoc},
  year    = {1980},
  volume  = {75},
  number  = {371},
  pages   = {591--593},
  doi     = {10.1080/01621459.1980.10477517}
}

@article{Shortreed2017,
  author  = {Shortreed, Susan M. and Ertefaie, Ashkan},
  title   = {Outcome-adaptive lasso: Variable Selection for causal inference},
  journal = {Biometrics},
  volume  = {73},
  number  = {4},
  pages   = {1111--1122},
  year    = {2017}
}

@article{Sobel2006,
  author  = {Sobel, Michael E.},
  title   = {What Do Randomized Studies of Housing Mobility Demonstrate? Causal Inference in the Face of Interference},
  journal = {J Am Stat Assoc},
  year    = {2006},
  volume  = {101},
  number  = {476},
  pages   = {1398--1407}
}

@article{Stefanski2002,
  author  = {Stefanski, Leonard A. and Boos, Dennis D.},
  title   = {The Calculus of M-Estimation},
  journal = {Am Stat},
  year    = {2002},
  volume  = {56},
  number  = {1},
  pages   = {29--38}
}

@article{Tchetgen2012,
  author  = {Tchetgen Tchetgen, Eric J. and VanderWeele, Tyler J.},
  title   = {On causal inference in the presence of interference},
  journal = {Stat Methods Med Res},
  volume  = {21},
  pages   = {55--75},
  year    = {2012}
}

\appendix

\setcounter{figure}{0}
\renewcommand{\thefigure}{A.\arabic{figure}}

\setcounter{table}{0}
\renewcommand{\thetable}{A.\arabic{table}}

\clearpage
\section{Regularity conditions}
\label{app:RC}

\begin{enumerate}
    \item[(A)] The group-level observations $O_1,\dots,O_m$ are independent and identically distributed.
    \item[(B)] The Fisher information matrix
    \[
    I(\psi)=\mathbb{E}\left[\left\{\frac{\partial}{\partial\psi}\log f(A_i \mid X_i;\psi)\right\}\left\{\frac{\partial}{\partial\psi}\log f(A_i \mid X_i;\psi)\right\}^{\top}\right]
    \]
    is finite and positive definite at $\psi=\psi^\ast$.
    \item[(C)] The reduced propensity score model is correctly specified.
    Furthermore, the first and second logarithmic derivatives of $f$ satisfy the equations
    \[
    \mathbb{E}\left[\frac{\partial}{\partial\psi_k}\log f(A_i \mid X_i;\psi)\right]=0 \quad \mathrm{for} \quad  k=1,\dots,p+1
    \]
    and
    \begin{align*}
    I_{k\ell}(\psi)
    &=\mathbb{E}\left[\frac{\partial}{\partial\psi_k}\log f(A_i \mid X_i;\psi)\frac{\partial}{\partial\psi_\ell}\log f(A_i \mid X_i;\psi)\right]\\
    &=\mathbb{E}\left[-\frac{\partial^2}{\partial\psi_k\partial\psi_\ell}\log f(A_i \mid X_i;\psi)\right] \quad \mathrm{for} \quad  k,\ell=1,\dots,p+1
    \end{align*}
    at $\psi=\psi^\ast$, the score has mean zero and the information identity holds.
    \item[(D)] There exists a neighborhood $\mathcal{N}$ of $\psi^\ast$ such that the second and third logarithmic derivatives of $f$ exist for all $\psi \in \mathcal{N}$.
    Furthermore, there exist functions $M_1(X_i)$ and $M_2(X_i)$ such that
    \[
    \left|\frac{\partial^2}{\partial\psi_k\partial\psi_\ell}\log f(A_i \mid X_i;\psi)\right|\leq M_1(X_i) \quad \mathrm{for} \quad  k,\ell=1,\dots,p+1
    \]
    and
    \[
    \left|\frac{\partial^3}{\partial\psi_k\partial\psi_\ell\partial\psi_r}\log f(A_i \mid X_i;\psi)\right| \leq M_2(X_i) \quad \mathrm{for} \quad  k,\ell,r=1,\dots,p+1
    \]
    where $\mathbb{E}[M_1(X_i)] < \infty$ and $\mathbb{E}[M_2(X_i)] < \infty$.
    \item[(E)] For the outcome regression coefficient estimator $\hat\alpha$, 
    \[
    \hat\alpha_k \overset{p}{\longrightarrow} \alpha_k^\ast \neq 0, \quad k \in \mathcal{C} \cup \mathcal{P},
    \]
    and
    \[
    \hat\alpha_k=O_p(m^{-1/2}), \quad k \in \mathcal{I} \cup \mathcal{S}.
    \]
\end{enumerate}

\section{Proof of Lemma \ref{lem:rootm_consistency}}
\label{proof:lem:rootm_consistency}

We prove the result following the ideas of \citep{Shortreed2017,Balde2025,Fan2001}.

\noindent
\textit{Proof}.
Let $\psi=\psi^\ast + u/\sqrt{m}$, $l_i(\psi)=\log f(A_i \mid X_i;\psi)$ and $l_m(\psi)=-\sum_{i=1}^m l_i(\psi)$, where $u$ belongs to a compact subset of $\mathbb{R}^{p+1}$.

We define $\hat{V}(u)=\mathcal{L}(u)-\mathcal{L}(0)$, where
\[
\mathcal{L}(u)=l_m\left(\psi^\ast+\frac{u}{\sqrt{m}}\right)+\lambda_m\sum_{k=1}^p\hat\omega_k\left|\beta_k^\ast+\frac{u_k}{\sqrt{m}}\right|.
\]
Thus, we have
\[
\hat{V}(u)=l_m\left(\psi^\ast+\frac{{u}}{\sqrt{m}}\right)-l_m(\psi^\ast)+\lambda_m\sum_{k=1}^p\hat\omega_k\left(\left|\beta_k^\ast+\frac{u_k}{\sqrt{m}}\right|-|\beta_k^\ast|\right).
\]
By applying a Taylor expansion to the first term, we have
\[
\begin{aligned}
\hat{V}(u)
&=\nabla l_m(\psi^\ast)^\top\left(\frac{u}{\sqrt{m}}\right)+\frac{1}{2}\left(\frac{u}{\sqrt{m}}\right)^\top\nabla^2l_m(\psi^\ast)\left(\frac{u}{\sqrt{m}}\right)\\
&\qquad+\frac{m^{-3/2}}{6}\sum_{k,\ell,r}\frac{\partial^3}{\partial\psi_k\partial\psi_\ell\partial\psi_r}l_m(\tilde\psi^\ast)u_ku_\ell u_r+\lambda_m\sum_{k=1}^p\hat\omega_k\left(\left|\beta_k^\ast+\frac{u_k}{\sqrt{m}}\right|-|\beta_k^\ast|\right)\\
&=-u^\top\frac{1}{\sqrt{m}}\sum_{i=1}^m\nabla l_i(\psi^\ast)-\frac{1}{2}u^\top\frac{1}{m}\sum_{i=1}^m\nabla^2l_i(\psi^\ast)u\\
&\qquad -\frac{m^{-3/2}}{6}\sum_{i=1}^m\sum_{k,\ell,r}\frac{\partial^3}{\partial\psi_k\partial\psi_\ell\partial\psi_r}l_i(\tilde\psi^\ast)u_k u_\ell u_r+\lambda_m\sum_{k=1}^p\hat\omega_k\left(\left|\beta_k^\ast+\frac{u_k}{\sqrt{m}}\right|-|\beta_k^\ast|\right)
\end{aligned}
\]
where $\tilde\psi^\ast$ is between $\psi^\ast$ and $\psi^\ast+u/\sqrt{m}$.

By applying the central limit theorem to the first term and the law of large numbers to the second term, we have
\begin{equation}
-u^\top\frac{1}{\sqrt{m}}\sum_{i=1}^m\nabla l_i(\psi^\ast) \overset{d}{\longrightarrow} -u^\top Z
\label{firstconv}
\end{equation}
and
\begin{equation}
-\frac{1}{2}u^\top\frac{1}{m}\sum_{i=1}^m\nabla^2l_i(\psi^\ast)u \overset{p}{\longrightarrow} \frac{1}{2}u^\top I(\psi^\ast)u,
\end{equation}
where $Z \sim N(0,I(\psi^\ast))$.

Under the condition (D), we have
\begin{align}
\left|\frac{m^{-3/2}}{6}\sum_{i=1}^m \sum_{k,\ell,r} \frac{\partial^3}{\partial\psi_k\partial\psi_\ell\partial\psi_r}l_i(\tilde\psi^\ast)u_k u_\ell u_r\right|
&=\left|\frac{1}{6\sqrt{m}}\frac{1}{m}\sum_{i=1}^m \sum_{k,\ell,r}\frac{\partial^3}{\partial\psi_k\partial\psi_\ell\partial\psi_r}l_i(\tilde\psi^\ast)u_k u_\ell u_r\right|
\notag\\
&\leq \frac{1}{6\sqrt{m}}\left\{\frac{1}{m}\sum_{i=1}^m M_2(X_i)\right\}\left(\sum_{k=1}^{p+1}|u_k|\right)^3 \overset{p}{\longrightarrow} 0.
\end{align}
The behavior of the fourth term depends on the covariate type.
If $k \in \mathcal{C} \cup \mathcal{P}$, that is confounders and outcome predictors, we have
\[
\lambda_m\hat\omega_k\left(\left|\beta_k^\ast+\frac{u_k}{\sqrt{m}}\right|-|\beta_k^\ast|\right)=\frac{\lambda_m}{\sqrt{m}}\hat\omega_k\sqrt{m}\left(\left|\beta_k^\ast+\frac{u_k}{\sqrt{m}}\right|-|\beta_k^\ast|\right)
\]
with $\lambda_m/\sqrt{m} \rightarrow 0$, $\hat\omega_k=O_p(1)$, and $\left|\sqrt{m}\left(\left|\beta_k^\ast+\frac{u_k}{\sqrt{m}}\right|-|\beta_k^\ast|\right)\right| \leq |u_k|$.
By Slutsky's theorem, we have
\begin{equation}
\lambda_m\hat\omega_k\left(\left|\beta_k^\ast+\frac{u_k}{\sqrt{m}}\right|-|\beta_k^\ast|\right) \overset{p}{\longrightarrow} 0.
\label{fourthconv}
\end{equation}
If $k \in \mathcal{I} \cup \mathcal{S}$, that is instrumental variables and spurious variables ($\beta_k^\ast=0$), we have
\[
\lambda_{m}\hat\omega_k\left(\left|\beta_k^\ast+\frac{u_k}{\sqrt{m}}\right|-|\beta_k^\ast|\right)=\lambda_m\hat\omega_k\left|\frac{u_k}{\sqrt{m}}\right| \geq 0.
\]

By the positive definiteness of the Fisher information matrix in condition (B), there exists a constant $c>0$ such that, on the sphere $\|u\|=M$,
\begin{align*}
\hat{V}(u)
&\geq l_m\left(\psi^\ast+\frac{u}{\sqrt{m}}\right)-l_m(\psi^\ast)+\lambda_m\sum_{k \in \mathcal{C} \cup \mathcal{P}} \hat\omega_k\left(\left|\beta_k^\ast+\frac{u_j}{\sqrt{m}}\right|-|\beta_k^\ast|\right)\\
&\geq -M\left\|\frac{1}{\sqrt{m}}\sum_{i=1}^m \nabla l_i(\psi^\ast)\right\|+\frac{c}{2}M^2 + o_p(1).
\end{align*}
By \eqref{firstconv}, we have
\[
\frac{1}{\sqrt{m}}\sum_{i=1}^m \nabla l_i(\psi^\ast)=O_p(1).
\]
Therefore, for every $\varepsilon>0$, we can choose a sufficiently large finite $M$ such that
\[
\mathbb{P}\left\{\inf_{\|u\|=M} \hat{V}(u)>0\right\} \geq 1-\varepsilon
\]
for all sufficiently large $m$.
Since $\hat{V}(0)=0$ and $\hat{V}(u)$ is continuous, a local minimizer exists in the interior of the closed ball $\{u;\|u\| \leq M\}$.
For sufficiently large $m$, the corresponding ball in the parameter space,
\[
\left\{\psi:\|\psi-\psi^\ast\| \leq \frac{M}{\sqrt{m}}\right\},
\]
is contained in the neighborhood of $\psi^\ast$.
Therefore, the local minimizer within this ball, which corresponds to the mixed-effects OAL estimator, is $\sqrt{m}$-consistent. \qed

\section{Proof of Proposition \ref{prop:oracle}}
\label{proof:prop:oracle}

\noindent
\textit{Proof}.
We first prove the consistency in variable selection.
Let 
\[
l_m^p(\beta_{\mathcal{C} \cup \mathcal{P}},\beta_{\mathcal{I} \cup \mathcal{S}},\sigma^2)=l_m(\beta_{\mathcal{C} \cup \mathcal{P}},\beta_{\mathcal{I} \cup \mathcal{S}},\sigma^2)+\lambda_m\sum_{k \in \mathcal{C} \cup \mathcal{P}}\hat\omega_k|\beta_k|+\lambda_m\sum_{k \in \mathcal{I} \cup \mathcal{S}}\hat\omega_k|\beta_k|
\]
be a penalized negative log-likelihood function.
Thus, to prove consistency in variable selection, it suffices to show that
\[
l_m^p(\beta_{\mathcal{C} \cup \mathcal{P}},\beta_{\mathcal{I} \cup \mathcal{S}},\sigma^2)-l_m^p(\beta_{\mathcal{C} \cup \mathcal{P}},0,\sigma^2)=l_m(\beta_{\mathcal{C} \cup \mathcal{P}},\beta_{\mathcal{I} \cup \mathcal{S}},\sigma^2)-l_m(\beta_{\mathcal{C} \cup \mathcal{P}},0,\sigma^2)+\lambda_m\sum_{k \in \mathcal{I} \cup \mathcal{S}}\hat\omega_k|\beta_k|>0
\]
with probability tending to one as $m \to \infty$.

By the mean value theorem, we have
\[
l_m(\beta_{\mathcal{C} \cup \mathcal{P}},\beta_{\mathcal{I} \cup \mathcal{S}},\sigma^2)-l_m(\beta_{\mathcal{C} \cup \mathcal{P}},0,\sigma^2)=\left\{\frac{\partial l_m(\beta_{\mathcal{C} \cup \mathcal{P}},\xi,\sigma^2)}{\partial\beta_{\mathcal{I} \cup \mathcal{S}}}\right\}^\top\beta_{\mathcal{I} \cup \mathcal{S}}
\]
for some $\|\xi\| \le \|\beta_{\mathcal{I} \cup \mathcal{S}}\|$.
By the mean value theorem again, we have

\begin{align*}
\left\|\frac{\partial l_m(\beta_{\mathcal{C} \cup \mathcal{P}},\xi,\sigma^2)}{\partial\beta_{\mathcal{I} \cup \mathcal{S}}}-\frac{\partial l_m(\beta_{\mathcal{C} \cup \mathcal{P}}^\ast,0,\sigma^{2\ast})}{\partial\beta_{\mathcal{I} \cup \mathcal{S}}}\right\|
&=\left\|\left(\frac{\partial l_m(\beta_{\mathcal{C} \cup \mathcal{P}},\xi,\sigma^2)}{\partial\beta_{\mathcal{I} \cup \mathcal{S}}}-\frac{\partial l_m(\beta_{\mathcal{C} \cup \mathcal{P}},0,\sigma^2)}{\partial\beta_{\mathcal{I} \cup \mathcal{S}}}\right)\right.\\
&\qquad +\left.\left(\frac{\partial l_m(\beta_{\mathcal{C} \cup \mathcal{P}},0,\sigma^2)}{\partial\beta_{\mathcal{I} \cup \mathcal{S}}}-\frac{\partial l_m(\beta_{\mathcal{C} \cup \mathcal{P}}^\ast,0,\sigma^{2\ast})}{\partial\beta_{\mathcal{I} \cup \mathcal{S}}}\right)\right\|\\
&\leq \left\|\frac{\partial l_m(\beta_{\mathcal{C} \cup \mathcal{P}},\xi,\sigma^2)}{\partial\beta_{\mathcal{I} \cup \mathcal{S}}}-\frac{\partial l_m(\beta_{\mathcal{C} \cup \mathcal{P}},0,\sigma^2)}{\partial\beta_{\mathcal{I} \cup \mathcal{S}}}\right\|\\
&\qquad + \left\|\frac{\partial l_m(\beta_{\mathcal{C} \cup \mathcal{P}},0,\sigma^2)}{\partial\beta_{\mathcal{I} \cup \mathcal{S}}}-\frac{\partial l_m(\beta_{\mathcal{C} \cup \mathcal{P}}^\ast,0,\sigma^{2\ast})}{\partial\beta_{\mathcal{I} \cup \mathcal{S}}}\right\|\\
&\leq \left(\sum_{i=1}^m M_{1}(X_i)\right)\|\xi\|+\left(\sum_{i=1}^m M_1(X_i)\right)\|\psi_{\mathcal{C} \cup \mathcal{P}}-\psi_{\mathcal{C} \cup \mathcal{P}}^\ast\|\\
&=(\|\xi\|+\|\psi_{\mathcal{C} \cup \mathcal{P}}-\psi_{\mathcal{C} \cup \mathcal{P}}^\ast\|)O_p(m)
\end{align*}
where $\psi_{\mathcal{C} \cup \mathcal{P}}=(\beta_{\mathcal{C} \cup \mathcal{P}},\sigma^2)$.
The limiting behavior of $(\|\xi\|+\|\psi_{\mathcal{C} \cup \mathcal{P}}-\psi_{\mathcal{C} \cup \mathcal{P}}^\ast\|)O_p(m)$ depends on whether $k \in \mathcal{I}$ or $k \in \mathcal{S}$.
For $k \in \mathcal{I}$, $\|\xi\| \leq \|\beta_{\mathcal{I}}\|=O_p(1)$.
Thus, we have
\[
\left\|\frac{\partial l_m(\beta_{\mathcal{C} \cup \mathcal{P}},\xi,\sigma^2)}{\partial\beta_{\mathcal{I} \cup \mathcal{S}}}-\frac{\partial l_m(\beta_{\mathcal{C} \cup \mathcal{P}}^\ast,0,\sigma^{2\ast})}{\partial\beta_{\mathcal{I} \cup \mathcal{S}}}\right\| \leq \left(\|\xi\|+\|\psi_{\mathcal{C} \cup \mathcal{P}}-\psi_{\mathcal{C} \cup \mathcal{P}}^\ast\|\right)O_p(m)=O_p(m).
\]
For $k \in \mathcal{S}$, $\|\xi\| \leq \|\beta_{\mathcal{S}}\|=O_p(m^{-1/2})$.
Thus, we have
\[
\left\|\frac{\partial l_m(\beta_{\mathcal{C} \cup \mathcal{P}},\xi,\sigma^2)}{\partial\beta_{\mathcal{I} \cup \mathcal{S}}}-\frac{\partial l_m(\beta_{\mathcal{C} \cup \mathcal{P}}^\ast,0,\sigma^{2\ast})}{\partial\beta_{\mathcal{I} \cup \mathcal{S}}}\right\| \le \left(\|\xi\|+\|\psi_{\mathcal{C} \cup \mathcal{P}}-\psi_{\mathcal{C} \cup \mathcal{P}}^\ast\|\right)O_p(m)=O_p(m^{1/2}).
\]
Therefore, for $k \in \mathcal{I} \cup \mathcal{S}$, we have
\[
\left\|\frac{\partial l_m(\beta_{\mathcal{C} \cup \mathcal{P}},\xi,\sigma^2)}{\partial\beta_{\mathcal{I} \cup \mathcal{S}}}-\frac{\partial l_m(\beta_{\mathcal{C} \cup \mathcal{P}}^\ast,0,\sigma^{2\ast})}{\partial\beta_{\mathcal{I} \cup \mathcal{S}}}\right\| \le \left(\|\xi\|+\|\psi_{\mathcal{C} \cup \mathcal{P}}-\psi_{\mathcal{C} \cup \mathcal{P}}^\ast\|\right)O_p(m)=O_p(m).
\]
Since the reduced propensity score model is correctly specified, the score evaluated at the true parameter has mean zero. 
By the central limit theorem, we have
\[
\frac{1}{\sqrt{m}}\frac{\partial l_m\left(\beta_{\mathcal{C} \cup \mathcal{P}}^\ast,0,\sigma^{2\ast}\right)}{\partial\beta_{\mathcal{I}\cup\mathcal{S}}}=O_p(1),
\]
Thus, we have
\[
\left\|\frac{\partial l_m\left(\beta_{\mathcal{C} \cup \mathcal{P}}^\ast,0,\sigma^{2\ast}\right)}{\partial\beta_{\mathcal{I} \cup \mathcal{S}}}\right\|=O_p(m^{1/2}).
\]
By using the triangle inequality, we have
\[
l_m(\beta_{\mathcal{C} \cup \mathcal{P}},\beta_{\mathcal{I} \cup \mathcal{S}},\sigma^2)-l_m(\beta_{\mathcal{C} \cup \mathcal{P}},0,\sigma^2) \geq -O_p(m)\sum_{k \in \mathcal{I} \cup \mathcal{S}}|\beta_k|.
\]
Then, we have
\begin{align}
l_m^p(\beta_{\mathcal{C} \cup \mathcal{P}},\beta_{\mathcal{I} \cup \mathcal{S}},\sigma^2)-l_m^p(\beta_{\mathcal{C} \cup \mathcal{P}},0,\sigma^2)
&\geq \sum_{k \in \mathcal{I} \cup \mathcal{S}} \{-|\beta_k|O_p(m)+\lambda_m\hat\omega_k|\beta_k|\}
\notag\\
&=\sum_{k \in \mathcal{I} \cup \mathcal{S}} \{-|\beta_k|O_p(m)+O_p(\lambda_m m^{\zeta/2})|\beta_k|\}.
\label{eq:prop1.1}
\end{align}
Let $(\hat\beta_{\mathcal{C} \cup \mathcal{P}},0,\hat\sigma^2)$ denote the minimizer of the penalized negative log-likelihood function $l_m^p(\beta_{\mathcal{C} \cup \mathcal{P}},0,\sigma^2)$.
Then, $l_m^p(\beta_{\mathcal{C} \cup \mathcal{P}},0,\sigma^2) \geq l_m^p(\hat\beta_{\mathcal{C} \cup \mathcal{P}},0,\hat\sigma^2)$ always holds.
Therefore, we have
\begin{align}
l_m^p(\beta_{\mathcal{C} \cup \mathcal{P}},\beta_{\mathcal{I} \cup \mathcal{S}},\sigma^2)-l_m^p(\hat\beta_{\mathcal{C} \cup \mathcal{P}},0,\hat\sigma^2)\notag
&=l_m^p(\beta_{\mathcal{C} \cup \mathcal{P}},\beta_{\mathcal{I} \cup \mathcal{S}},\sigma^2)-l_m^p(\beta_{\mathcal{C} \cup \mathcal{P}},0,\sigma^2)\\
&\qquad + l_m^p(\beta_{\mathcal{C} \cup \mathcal{P}},0,\sigma^2)-l_m^p(\hat\beta_{\mathcal{C} \cup \mathcal{P}},0,\hat\sigma^2)\label{eq:prop1.2}\\\notag
&\geq l_m^p(\beta_{\mathcal{C} \cup \mathcal{P}},\beta_{\mathcal{I} \cup \mathcal{S}},\sigma^2)-l_m^p(\beta_{\mathcal{C} \cup \mathcal{P}},0,\sigma^2).
\end{align}
By inequality \eqref{eq:prop1.1} and $\lambda_m m^{\zeta/2-1} \to \infty$ as $m \to \infty$, the right-hand side of the inequality \eqref{eq:prop1.2} is positive with probability tending to 1 as $m \to \infty$.
This completes the proof of the consistency in variable selection.

Next, we prove the asymptotic normality.
By the consistency in variable selection, with probability tending to one, $\widehat\beta_{\mathcal{I} \cup \mathcal{S}}=0$.
Moreover, by Lemma~\ref{lem:rootm_consistency}, there exists a $\sqrt{m}$-consistent local minimizer $(\hat\beta_{\mathcal{C} \cup \mathcal{P}},\hat\sigma^2)$ of the penalized negative log-likelihood $l_m^p(\beta_{\mathcal{C} \cup \mathcal{P}},0,\sigma^2)$ such that
\[
\left.\frac{\partial l_m^p(\beta_{\mathcal{C} \cup \mathcal{P}},\beta_{\mathcal{I} \cup \mathcal{S}},\sigma^2)}{\partial\beta_{\mathcal{C} \cup \mathcal{P}}}\right|_{(\beta_{\mathcal{C} \cup \mathcal{P}},\beta_{\mathcal{I} \cup \mathcal{S}},\sigma^2) = (\hat\beta_{\mathcal{C} \cup \mathcal{P}},0,\hat\sigma^2)}=0.
\]
Note that $\hat{\psi}_{OAL}=(\hat\beta_{\mathcal{C} \cup \mathcal{P}},\hat\sigma^2)$ is a consistent estimator and $\hat\omega_k=O_p(1)$ for $k \in \mathcal{C} \cup \mathcal{P}$, by the Taylor expansion,
\[
0=\frac{1}{\sqrt{m}}\nabla l_m\left(\beta_{\mathcal{C} \cup \mathcal{P}}^\ast,0,\sigma^{2\ast}\right)+\frac{1}{m}\nabla^2 l_m\left(\beta_{\mathcal{C}\cup\mathcal{P}}^\ast,0,\sigma^{2\ast}\right)\sqrt{m}\left(\hat\psi_{\mathrm{OAL}}-\psi_{\mathcal{C} \cup \mathcal{P}}^\ast\right)+o_p(1).
\]
Therefore, by the central limit theorem, the law of large
numbers and Slutsky's theorem,
\[
\sqrt{m}\left(\hat\psi_{\mathrm{OAL}}-\psi_{\mathcal{C} \cup \mathcal{P}}^\ast\right) \overset{d}{\longrightarrow} N(0,I_{11}^{-1}).
\]
This completes the proof of the asymptotic normality.\qed

\section{Asymptotic variance of the IPW estimator}\label{app:variance}

\begin{corollary}
\label{coro1}
Suppose that the propensity score model is correctly specified.
Then, the following hold:
\[
Q_{11}=\mathbb{E}\left[-\frac{\partial U_\mu(X_i,A_i,Y_i;\mu^\ast,\psi^\ast)}{\partial\mu}\right]=1.
\]
\end{corollary}

\noindent
\textit{Proof}.
By the law of iterated expectations,
\begin{align*}
\mathbb{E}\left[-\frac{\partial U_\mu(X_i,A_i,Y_i;\mu^\ast,\psi^\ast)}{\partial\mu}\right]
&=\mathbb{E}\left[\sum_{j=1}^{n_i}\frac{\mathbf{1}\{A_{ij}=a\}\pi(A_{i(-j)};\theta)}{n_i f(A_i\mid X_i;\psi)}\right]\\
&=\mathbb{E}\left[\mathbb{E}\left[\sum_{j=1}^{n_i}\left.\frac{\mathbf{1}\{A_{ij}=a\}\pi(A_{i(-j)};\theta)}{n_if(A_i\mid X_i;\psi^\ast)}\right|X_i\right]\right]\\
&=\mathbb{E}\left[\frac{1}{n_i}\sum_{j=1}^{n_i}\sum_{a_{i(-j)}}\pi(a_{i(-j)};\theta)\right].
\end{align*}
Since $\sum_{a_{i(-j)}}\pi(a_{i(-j)};\theta)=1$, we have
\[
Q_{11}=\mathbb{E}\left[-\frac{\partial U_\mu(X_i,A_i,Y_i;\mu^\ast,\psi^\ast)}{\partial\mu}\right]=\mathbb{E}\left[\frac{1}{n_i}\sum_{j=1}^{n_i}1\right]=1.
\]
\qed

\subsection{Proof of Lemma~\ref{Lemma2}}
\label{app:variance1}

\noindent
\textit{Proof}. 
Using block matrix notation, we write
\[
Q_1=
\begin{pmatrix}
Q_{11} & Q_{12}\\
0 & Q_{22}
\end{pmatrix}, \qquad \Omega_1=
\begin{pmatrix}
\Omega_{11} & \Omega_{12}\\
\Omega_{12}^\top & \Omega_{22}
\end{pmatrix},
\]
where $Q_{11}=\mathbb{E}\left[-\frac{\partial U_\mu(X_i,A_i,Y_i;\mu^\ast,\psi^\ast)}{\partial\mu}\right]$, $Q_{12}=\mathbb{E}\left[-\frac{\partial U_\mu(X_i,A_i,Y_i;\mu^\ast,\psi^\ast)}{\partial\psi}\right]$, and $Q_{22}=\mathbb{E}\left[-\frac{\partial S_\psi(X_i,A_i;\mu^\ast,\psi^\ast)}{\partial\psi}\right]$.
The inverse matrix of $Q_1$ is given by
\[
Q_1^{-1}=
\begin{pmatrix}
Q_{11}^{-1} & -Q_{11}^{-1}Q_{12}Q_{22}^{-1}\\
0 & Q_{22}^{-1}
\end{pmatrix}.
\]
Thus, we have
\[
Q_1^{-1}\Omega_1Q_1^{-1,\top}=
\begin{pmatrix}
Q_{11}^{-1}\Omega_{11}Q_{11}^{-1,\top}-Q_{11}^{-1}\Omega_{12}\Omega_{22}^{-1}\Omega_{12}^\top Q_{11}^{-1,\top} & 0\\
0 & \Omega_{22}^{-1}
\end{pmatrix}.
\]
For $e=(1,0)$, we have
\begin{align*}
\Sigma_{\mathrm{est}} 
&=eQ_1^{-1}\Omega_1Q_1^{-1,\top}e^\top\\
&=Q_{11}^{-1}\Omega_{11}Q_{11}^{-1,\top}-Q_{11}^{-1}\Omega_{12}\Omega_{22}^{-1}\Omega_{12}^\top Q_{11}^{-1,\top}\\
&=\Sigma_{\mathrm{known}}-\Omega_{12}\Omega_{22}^{-1}\Omega_{12}^\top. \qquad (\because \, \text{Corollary}~\ref{coro1})
\end{align*}
Since $\Omega_{22}$ is positive definite, $\Omega_{12}\Omega_{22}^{-1}\Omega_{12}^\top \geq 0$.
Therefore, $\Sigma_{\mathrm{est}} \leq \Sigma_{\mathrm{known}}$.
\qed

\subsection{Proof of Lemma~\ref{Lemma3}}
\label{app:variance2}

\noindent
\textit{Proof}. 
We define
\[
\tilde{\Omega}_{12}=
\begin{pmatrix}
\Omega_{12} & \Omega_{13}
\end{pmatrix}, \qquad \tilde{\Omega}_{22}=
\begin{pmatrix}
\Omega_{22} & \Omega_{23}\\
\Omega_{23}^\top & \Omega_{33}
\end{pmatrix}.
\]
By the same calculation as in Lemma~\ref{Lemma2}, we have
\[
\Sigma_{\mathrm{est}}^V=\Sigma_{\mathrm{known}}-\tilde{\Omega}_{12}\tilde{\Omega}_{22}^{-1}\tilde{\Omega}_{12}^\top.
\]
The inverse matrix of $\tilde{\Omega}_{22}$ is given by
\[
\tilde\Omega_{22}^{-1}=
\begin{pmatrix}
\Omega_{22}^{-1}+\Omega_{22}^{-1}\Omega_{23}
(\Omega_{33}-\Omega_{23}^\top\Omega_{22}^{-1}\Omega_{23})^{-1}\Omega_{23}^\top\Omega_{22}^{-1} & -\Omega_{22}^{-1}\Omega_{23}(\Omega_{33}-\Omega_{23}^\top\Omega_{22}^{-1}\Omega_{23})^{-1}\\
-(\Omega_{33}-\Omega_{23}^\top\Omega_{22}^{-1}\Omega_{23})^{-1}\Omega_{23}^\top\Omega_{22}^{-1} & (\Omega_{33}-\Omega_{23}^\top\Omega_{22}^{-1}\Omega_{23})^{-1}
\end{pmatrix}.
\]
Thus, we have
\[
\begin{aligned}
\Sigma_{\mathrm{est}}^V
&= \Omega_{11}-\Omega_{12}\Omega_{22}^{-1}\Omega_{21}-\left(\Omega_{13}-\Omega_{12}\Omega_{22}^{-1}\Omega_{23}\right)(\Omega_{33}-\Omega_{23}^\top\Omega_{22}^{-1}\Omega_{23})^{-1}\left(\Omega_{13}-\Omega_{12}\Omega_{22}^{-1}\Omega_{21}\right)^\top\\
&= \Sigma_{\mathrm{est}}-\left(\Omega_{13}-\Omega_{12}\Omega_{22}^{-1}\Omega_{23}\right)(\Omega_{33}-\Omega_{23}^\top\Omega_{22}^{-1}\Omega_{23})^{-1}\left(\Omega_{13}-\Omega_{12}\Omega_{22}^{-1}\Omega_{23}\right)^\top.
\end{aligned}
\]
Since $\Omega_{33}-\Omega_{23}^\top\Omega_{22}^{-1}\Omega_{23}$ is positive definite, we have
\[
\left(\Omega_{13}-\Omega_{12}\Omega_{22}^{-1}\Omega_{23}\right)(\Omega_{33}-\Omega_{23}^\top\Omega_{22}^{-1}\Omega_{23})^{-1}\left(\Omega_{13}-\Omega_{12}\Omega_{22}^{-1}\Omega_{23}\right)^\top \geq 0.
\]
Therefore, $\Sigma_{\mathrm{est}}^V \leq \Sigma_{\mathrm{est}}$.\qed

\section{Proof of Proposition~\ref{prop2}}
\label{proof:prop2}
By the consistency in variable selection established in Proposition~\ref{prop:oracle}, covariates in $\mathcal{I} \cup \mathcal{S}$ are excluded with probability tending to one. 
Therefore, it suffices to consider the reduced propensity score model containing only the covariates in $\mathcal{C} \cup \mathcal{P}$.
For notational convenience, we write the corresponding coefficients as $\beta$ and $\gamma$, respectively.\\
\textit{Proof}. 
Consider the penalized negative log-likelihood
\[
-\sum_{i=1}^m\log f(A_i \mid X_i,V_i;\beta,\gamma,\sigma^2)+\lambda_m\left\{\sum_{k\in\mathcal{C}}\hat\omega_{\beta_k}|\beta_k|+\sum_{k\in\mathcal{P}}\hat\omega_{\gamma_k}|\gamma_k|\right\}.
\]
The Karush--Kuhn--Tucker conditions with respect to $\beta$ and $\gamma$ are
\[
-\sum_{i=1}^m S_\beta(X_i,V_i,A_i;\hat\beta_{OAL},\hat\gamma_{OAL},\hat\sigma^2_{OAL})+\lambda_m\hat\omega_{\beta_k} d_{\beta_k}=0, \qquad k\in\mathcal{C}
\]
and
\[
-\sum_{i=1}^m S_\gamma(X_i,V_i,A_i;\hat\beta_{OAL},\hat\gamma_{OAL},\hat\sigma^2_{OAL})+\lambda_m\hat\omega_{\gamma_k} d_{\gamma_k}=0, \qquad k\in\mathcal{P}
\]
where $S_\beta(X_i,V_i,A_i;\beta,\gamma,\sigma^2)=\partial\log f(A_i \mid X_i,V_i;\beta,\gamma,\sigma^2)/\partial\beta$, $S_\gamma(X_i,V_i,A_i;\beta,\gamma,\sigma^2)=\partial\log f(A_i \mid X_i,V_i;\beta,\gamma,\sigma^2)/\partial\gamma$ and $d_{\beta_k}$ and $d_{\gamma_k}$ are elements of the subdifferential of the $L_1$ norm evaluated at $\beta_k$ or $\gamma_k$. 
That is,
\[
d_{\beta_k} \in
\begin{cases}
    \{-1\}, & \beta_k<0\\
    [-1,1], & \beta_k=0\\
    \{1\}, & \beta_k>0
\end{cases}
\quad \mathrm{and} \quad 
d_{\gamma_k} \in
\begin{cases}
    \{-1\}, & \gamma_k<0\\
    [-1,1], & \gamma_k=0\\
    \{1\}, & \gamma_k>0
\end{cases}.
\]
Since $\lambda_m/\sqrt{m} \to 0$, $\hat\omega_{\beta_k}=O_p(1)$, $\hat\omega_{\gamma_k}=O_p(1)$, $|d_{\beta_k}| \leq 1$ and $|d_{\gamma_k}| \leq 1$, it follows that
\[
\lambda_m\hat\omega_{\beta_k} d_{\beta_k}=o_p(m^{1/2}) \quad\mathrm{and}\quad \lambda_m\hat\omega_{\gamma_k} d_{\gamma_k}=o_p(m^{1/2}).
\]
Therefore,
\[
\frac{1}{m}\sum_{i=1}^mS_\beta(X_i,V_i,A_i;\hat\beta_{OAL},\hat\gamma_{OAL},\hat\sigma_{OAL}^2)=o_p(m^{-1/2}),
\]
and,
\[
\frac{1}{m}\sum_{i=1}^mS_\gamma(X_i,V_i,A_i;\hat\beta_{OAL},\hat\gamma_{OAL},\hat\sigma_{OAL}^2)=o_p(m^{-1/2}).
\]

For the unpenalized variance parameter,
\[
\sum_{i=1}^m S_{\sigma^2}(X_i,V_i,A_i;\hat\beta_{OAL},\hat\gamma_{OAL},\hat\sigma^2_{OAL})=0.
\]
 where $S_{\sigma^2}(X_i,V_i,A_i;\beta,\gamma,\sigma^2)=\partial\log f(A_i \mid X_i,V_i;\beta,\gamma,\sigma^2)/\partial\sigma^2$. 
 This completes the proof of Proposition~\ref{prop2}.
\qed

\section{Results of additional numerical studies}\label{app:results}

\begin{figure}[H]
    \centering

    \begin{subfigure}[t]{0.47\textwidth}
        \centering
        \includegraphics[width=\linewidth]{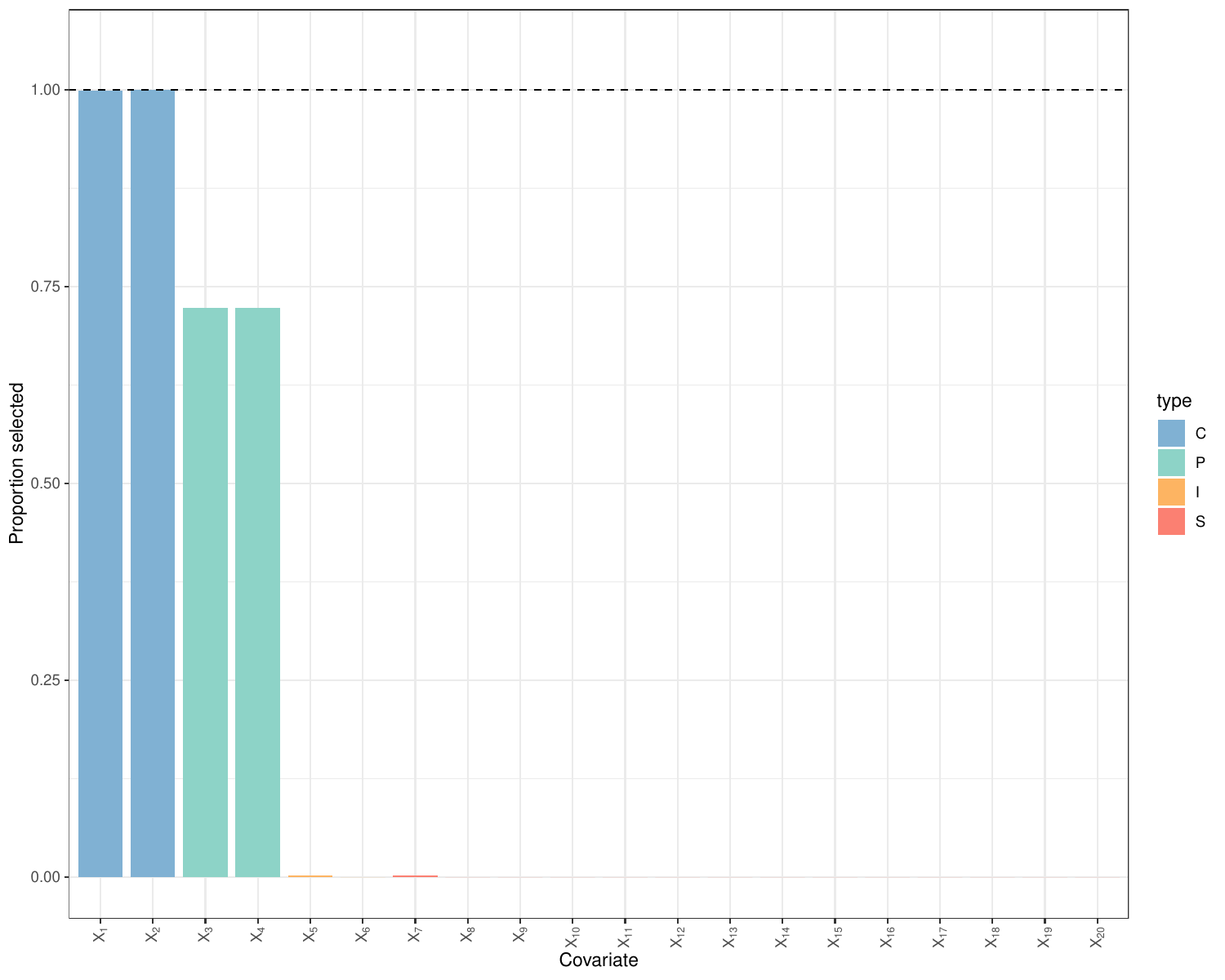}
        \caption{$\rho=0$}
    \end{subfigure}
    \hspace{0.03\textwidth}
    \begin{subfigure}[t]{0.47\textwidth}
        \centering
        \includegraphics[width=\linewidth]{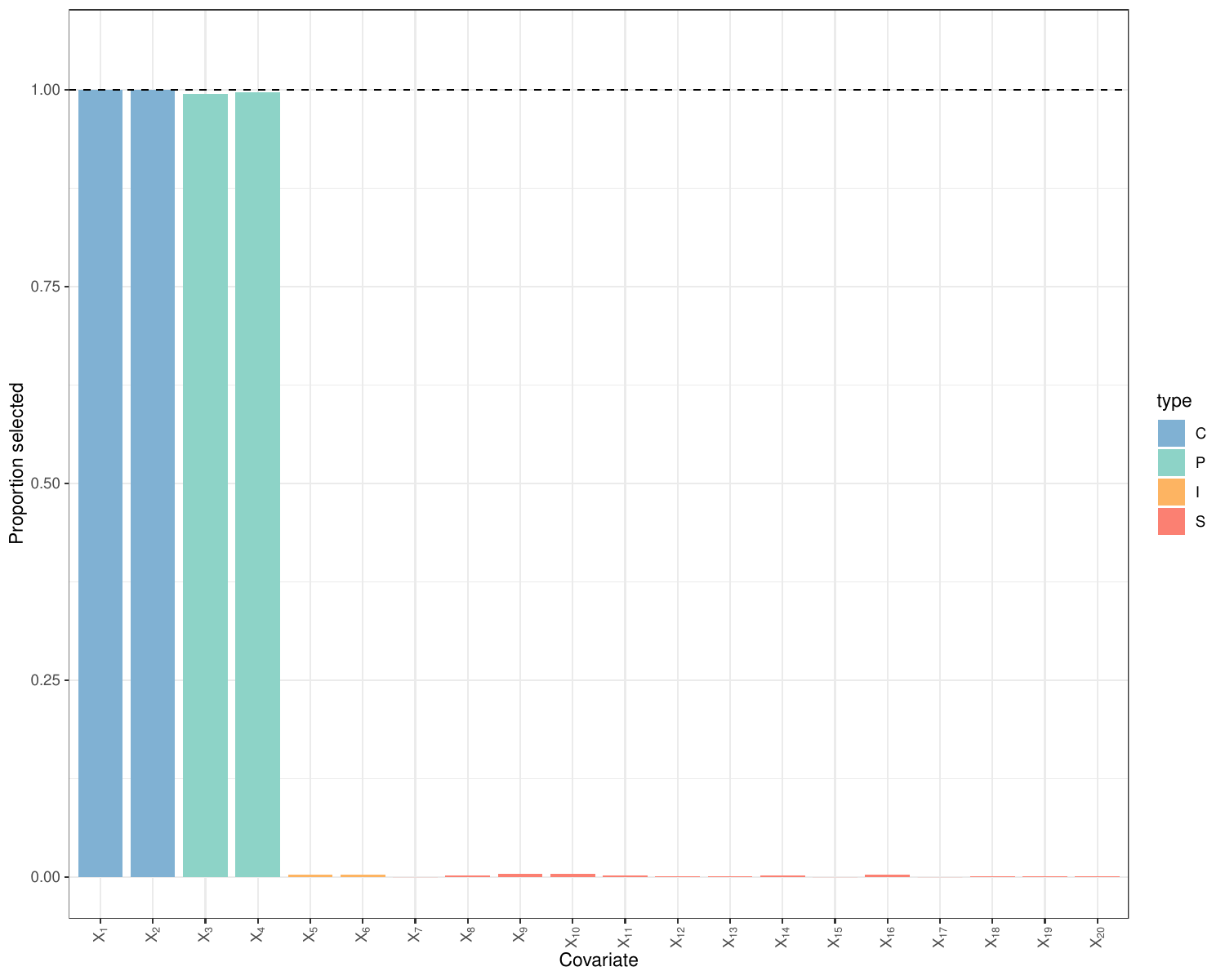}
        \caption{$\rho=0.2$}
    \end{subfigure}

    \vspace{0.3cm}

    \begin{subfigure}[t]{0.47\textwidth}
        \centering
        \includegraphics[width=\linewidth]{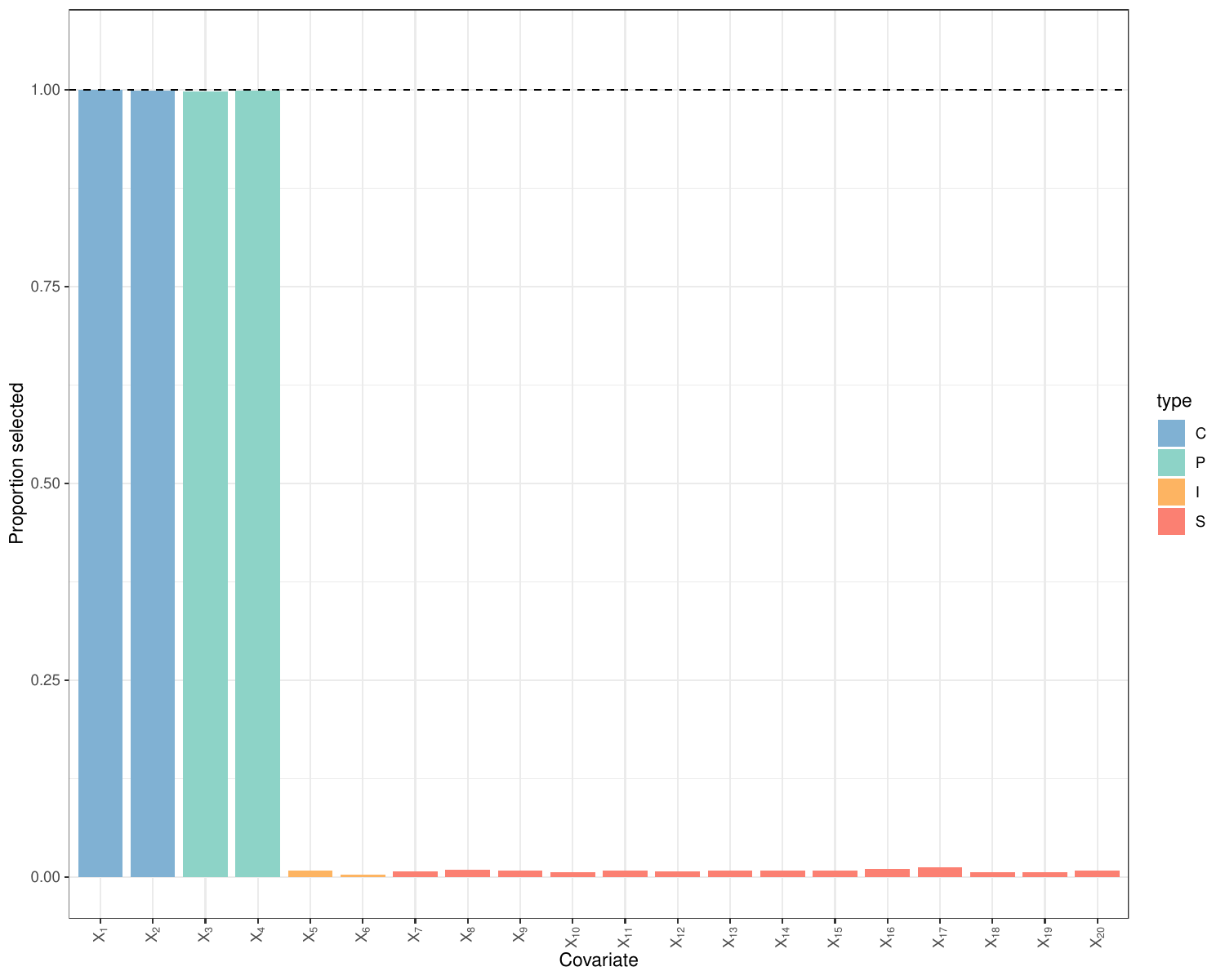}
        \caption{$\rho=0.5$}
    \end{subfigure}
    \hspace{0.03\textwidth}
    \begin{subfigure}[t]{0.47\textwidth}
        \centering
        \includegraphics[width=\linewidth]{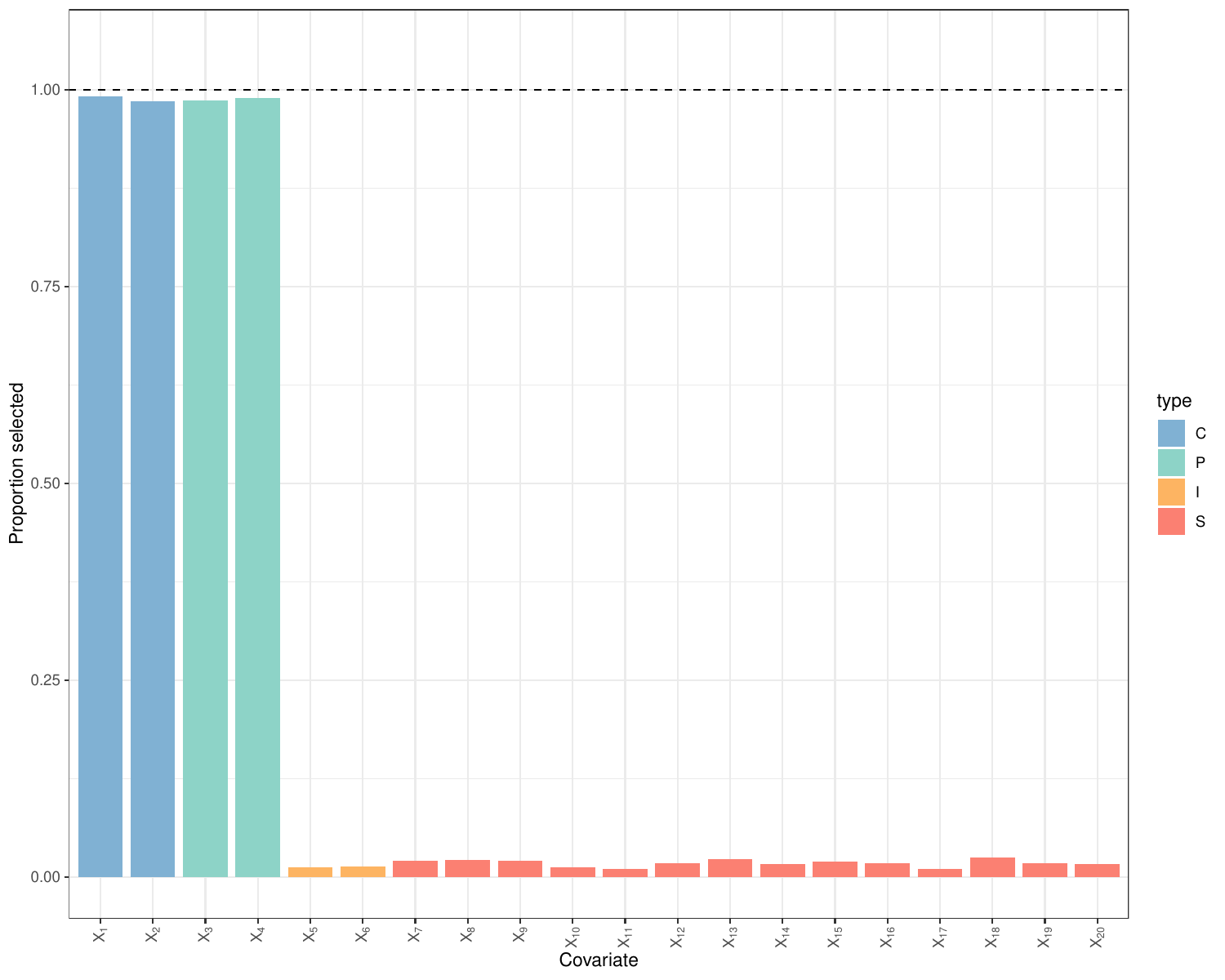}
        \caption{$\rho=0.75$}
    \end{subfigure}

    \caption{
    Variable selection frequencies under the four correlation structures in the weak confounding settings.
    Subfigures (1), (2), (3), and (4) correspond to $\rho=0$, $0.2$, $0.5$, and $0.75$, respectively.
    The bars show the proportion of times each covariate was selected over Monte Carlo replications.
    Covariates $X_1$ and $X_2$ are confounders, $X_3$ and $X_4$ are outcome predictors, $X_5$ and $X_6$ are instrumental variables, and $X_7$--$X_{20}$ are spurious variables.
    }
\end{figure}

\begin{figure}[p]
    \centering

    \begin{subfigure}[t]{0.47\textwidth}
        \centering
        \includegraphics[width=\linewidth]{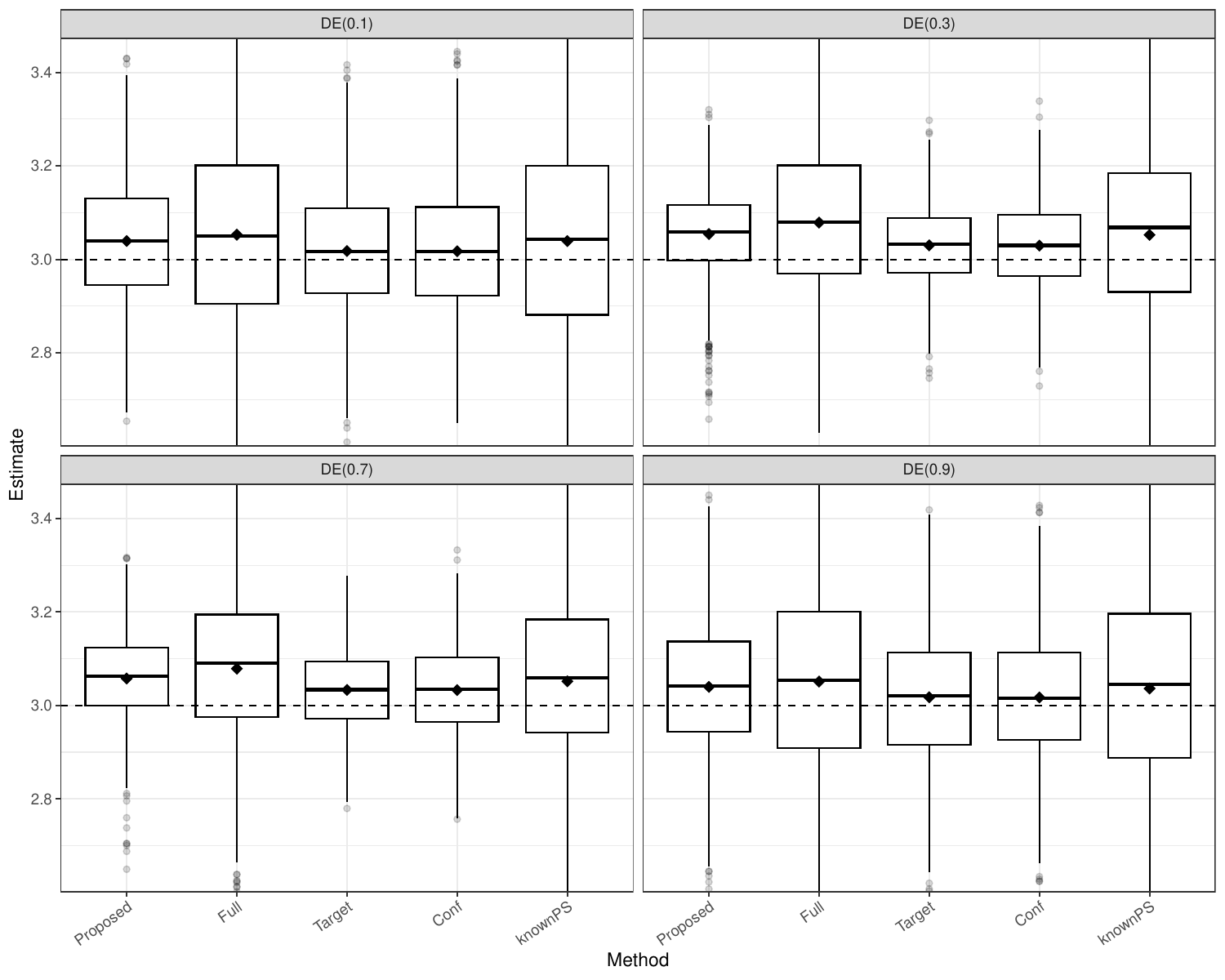}
        \caption{$\rho=0$}
    \end{subfigure}
    \hspace{0.03\textwidth}
    \begin{subfigure}[t]{0.47\textwidth}
        \centering
        \includegraphics[width=\linewidth]{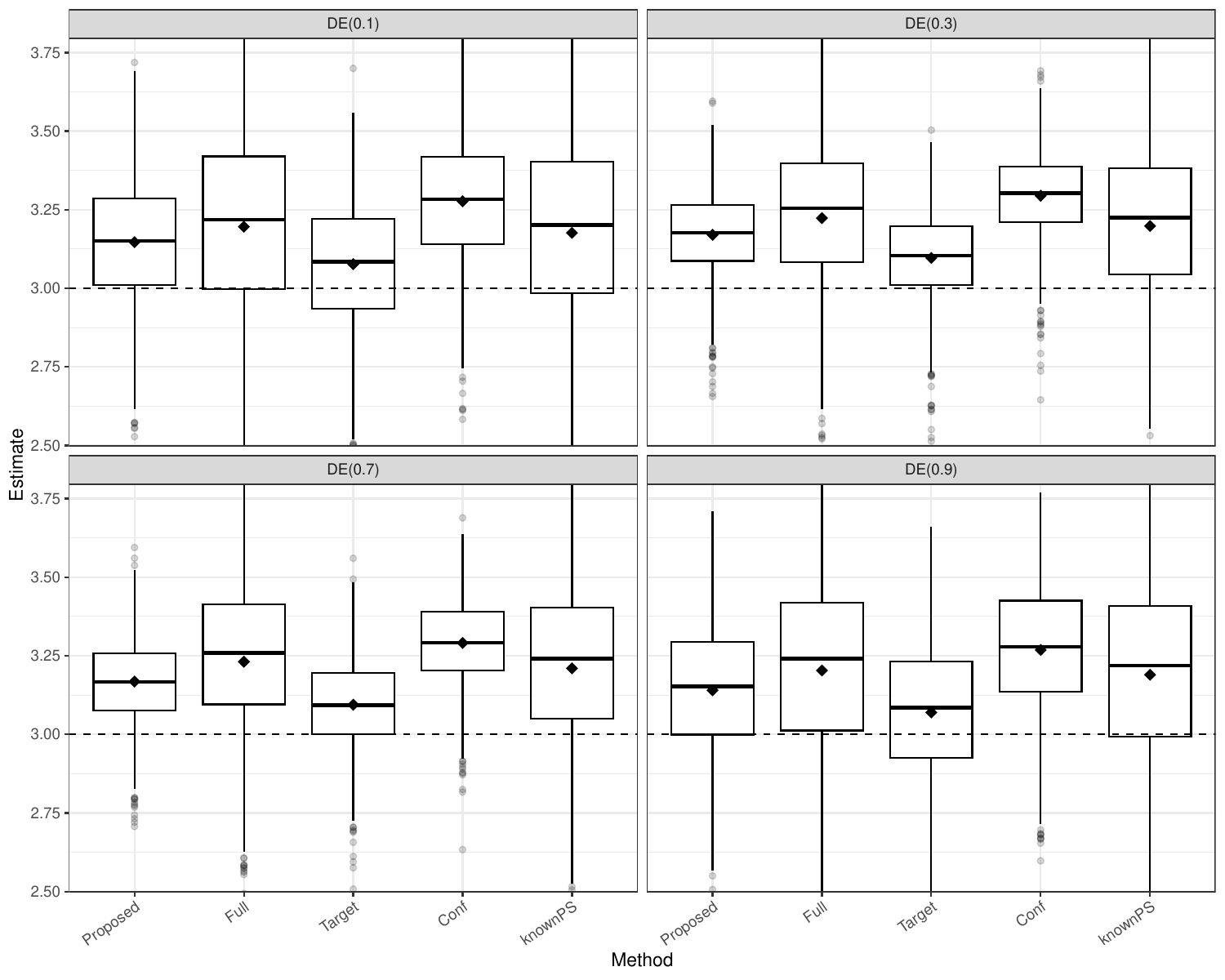}
        \caption{$\rho=0.2$}
    \end{subfigure}

    \vspace{0.3cm}

    \begin{subfigure}[t]{0.47\textwidth}
        \centering
        \includegraphics[width=\linewidth]{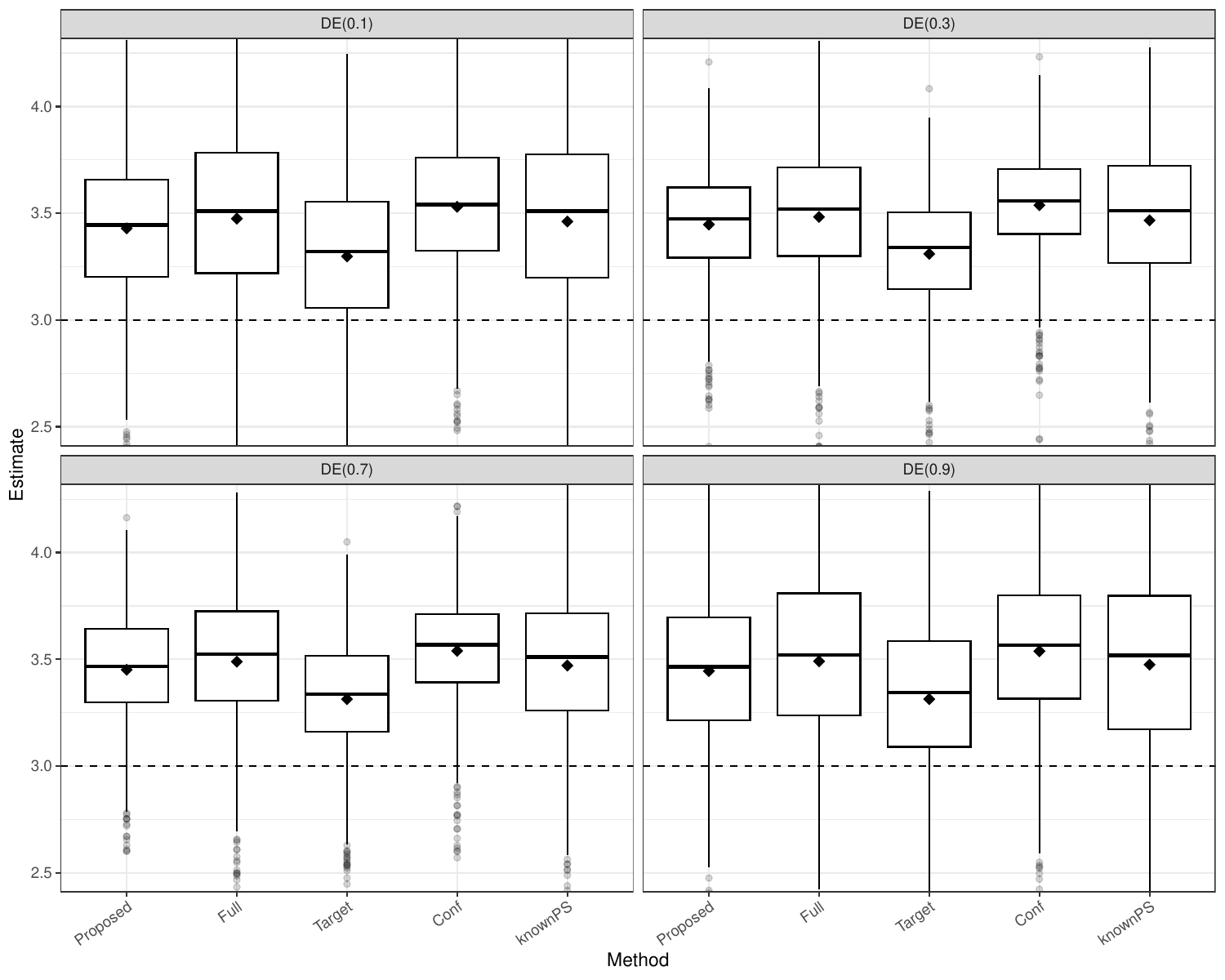}
        \caption{$\rho=0.5$}
    \end{subfigure}
    \hspace{0.03\textwidth}
    \begin{subfigure}[t]{0.47\textwidth}
        \centering
        \includegraphics[width=\linewidth]{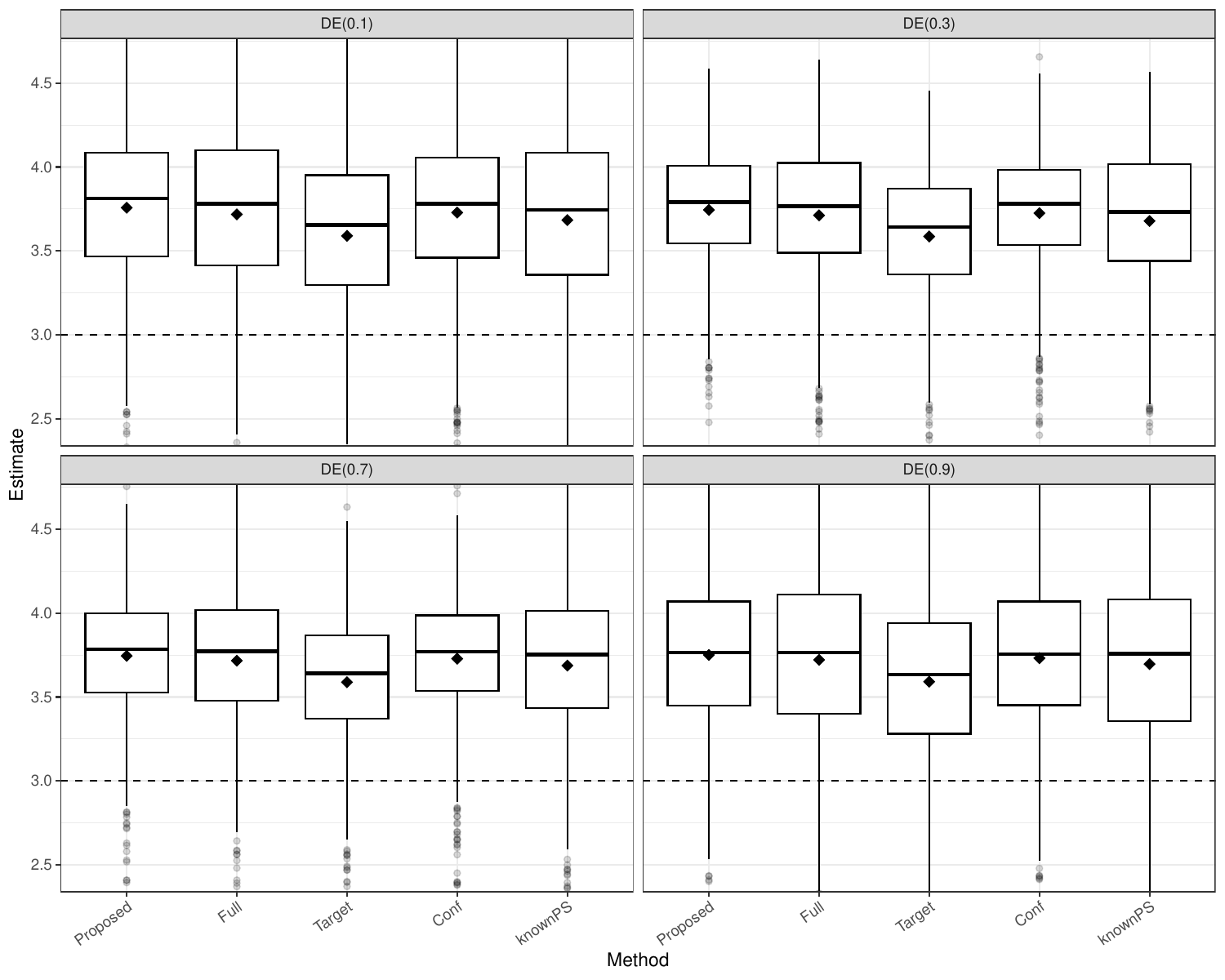}
        \caption{$\rho=0.75$}
    \end{subfigure}

    \caption{
    Boxplots of the direct effect estimators under the four correlation structures in the weak confounding settings.
    Subfigures (1), (2), (3) and (4) correspond to $\rho=0$, $0.2$, $0.5$ and $0.75$, respectively.
    Within each subfigure, the four panels correspond to $\theta=0.1$, $0.3$, $0.7$, and $0.9$, respectively.
    The dashed line represents the true direct effect, $\overline{\mathrm{DE}}(\theta)=3$.
    }
\end{figure}

\begin{figure}[p]
    \centering

    \begin{subfigure}[t]{0.47\textwidth}
        \centering
        \includegraphics[width=\linewidth]{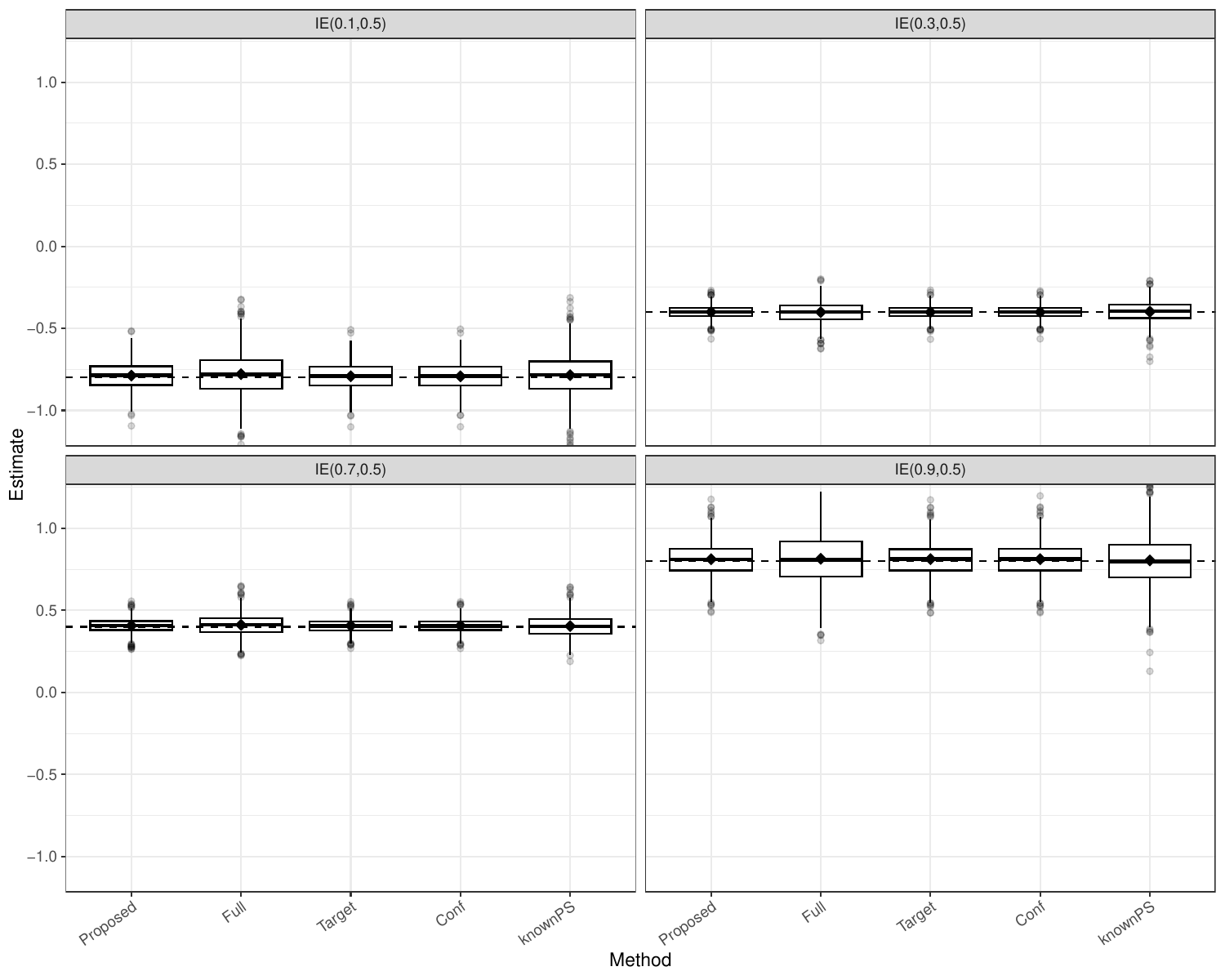}
        \caption{$\rho=0$}
    \end{subfigure}
    \hspace{0.03\textwidth}
    \begin{subfigure}[t]{0.47\textwidth}
        \centering
        \includegraphics[width=\linewidth]{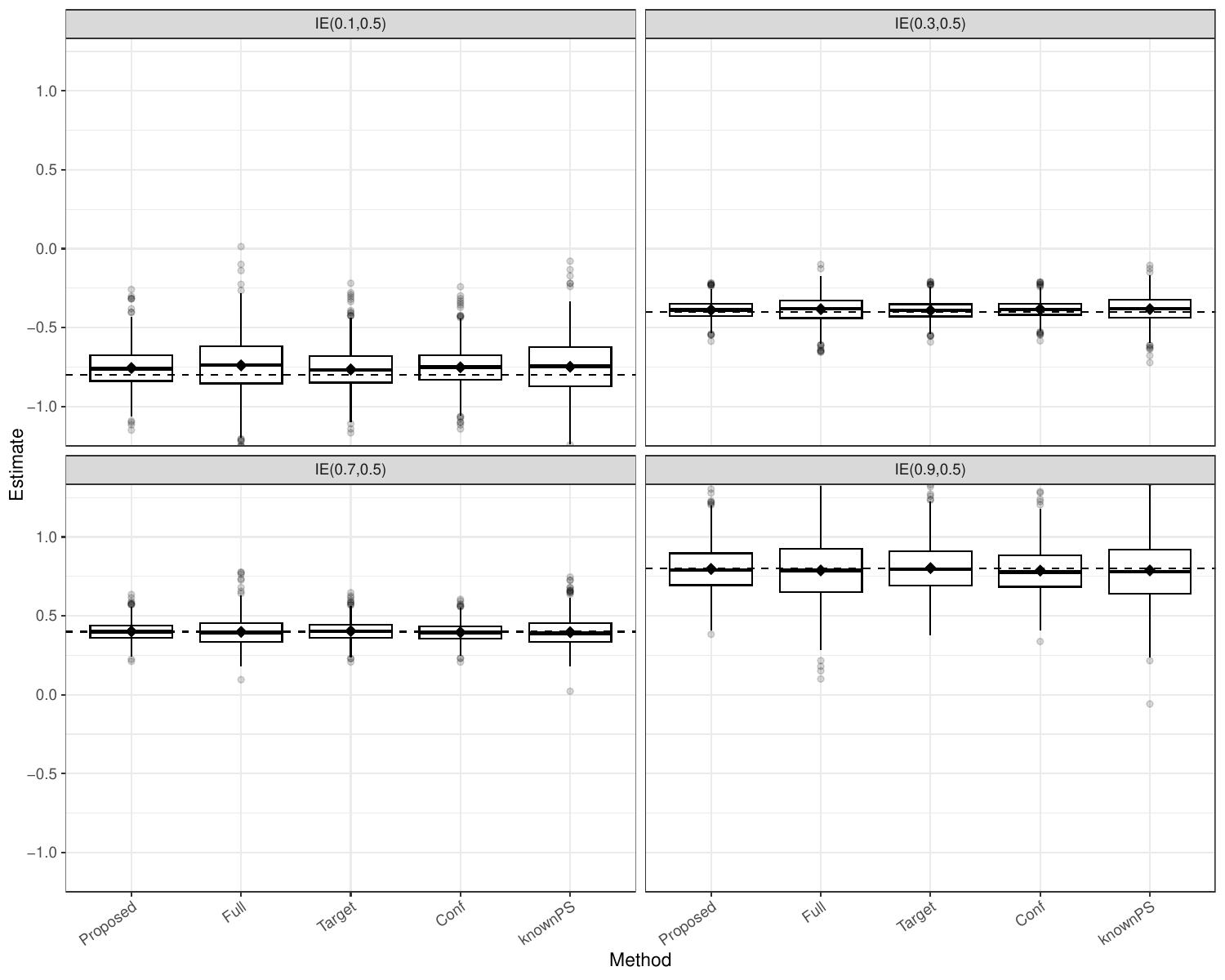}
        \caption{$\rho=0.2$}
    \end{subfigure}

    \vspace{0.3cm}

    \begin{subfigure}[t]{0.47\textwidth}
        \centering
        \includegraphics[width=\linewidth]{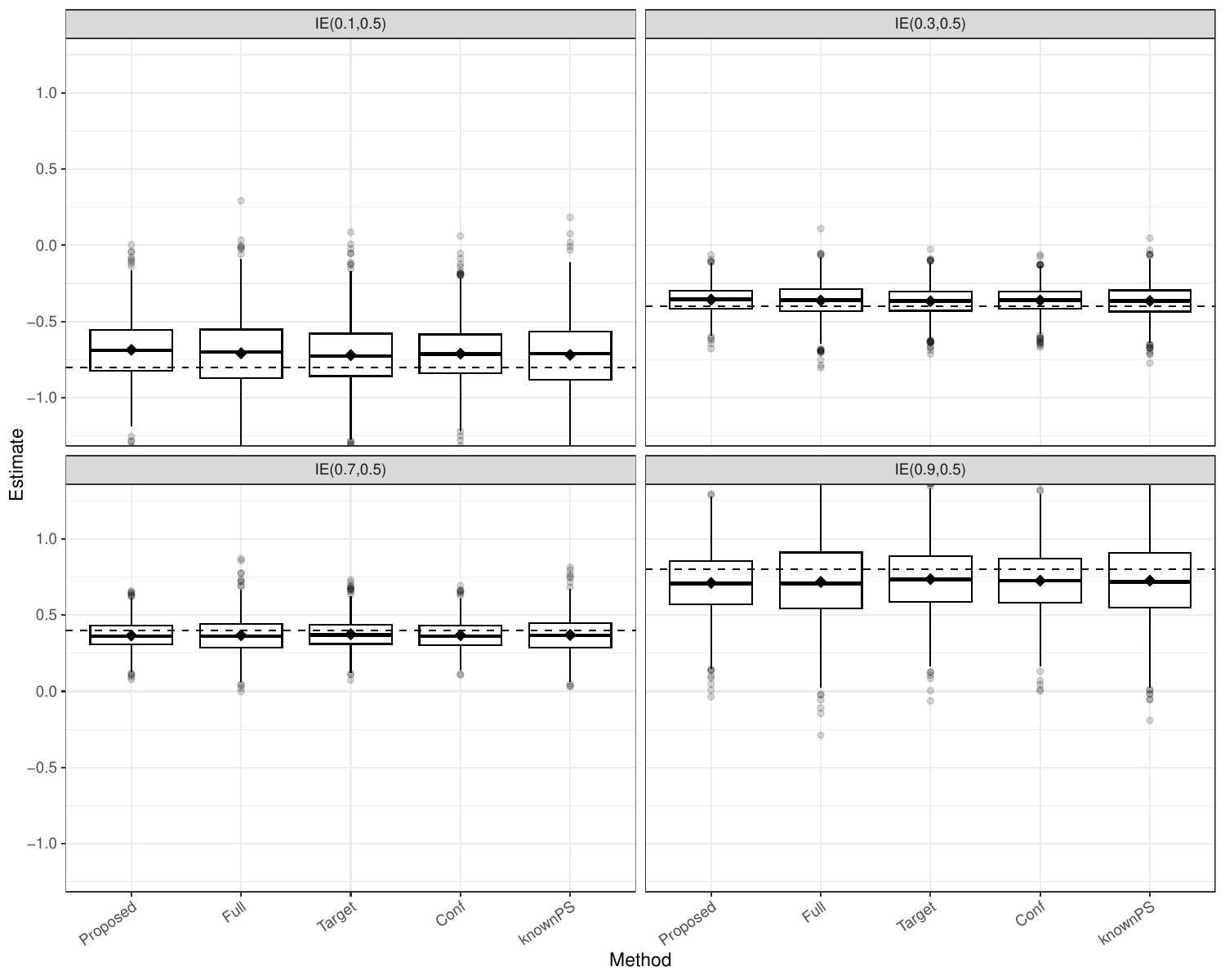}
        \caption{$\rho=0.5$}
    \end{subfigure}
    \hspace{0.03\textwidth}
    \begin{subfigure}[t]{0.47\textwidth}
        \centering
        \includegraphics[width=\linewidth]{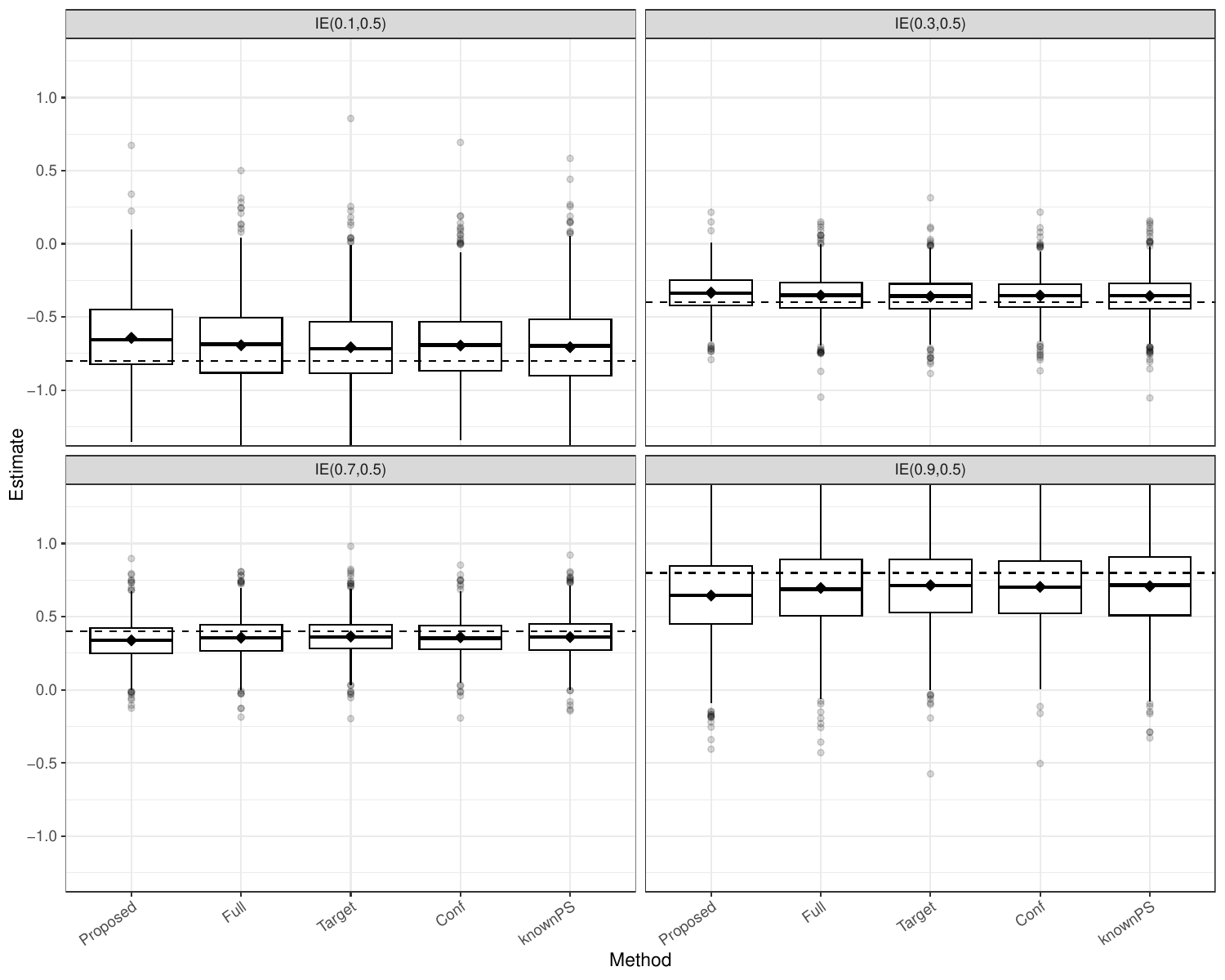}
        \caption{$\rho=0.75$}
    \end{subfigure}
    
    \caption{
    Boxplots of the indirect effect estimators under the four correlation structures in the weak confounding settings.
    Subfigures (1), (2), (3) and (4) correspond to $\rho=0$, $0.2$, $0.5$ and $0.75$, respectively.
    Within each subfigure, the four panels correspond to $\theta_1=0.1$, $0.3$, $0.7$, and $0.9$, respectively, with $\theta_0=0.5$ as the reference allocation strategy.
    The dashed lines represent the true indirect effects, $\overline{\mathrm{IE}}(0.1,0.5)=-0.8$, $\overline{\mathrm{IE}}(0.3,0.5)=-0.4$, $\overline{\mathrm{IE}}(0.7,0.5)=0.4$ and $\overline{\mathrm{IE}}(0.9,0.5)=0.8$.
    }
\end{figure}

\begingroup
\footnotesize

\setlength{\tabcolsep}{2.5pt}
\renewcommand{\arraystretch}{0.8}
\setlength{\LTpre}{0pt}
\setlength{\LTpost}{0pt}

\captionsetup{
    font=footnotesize,
    skip=3pt
}

\begin{longtable}{
    @{}
    p{0.46\textwidth}
    >{\centering\arraybackslash}p{0.17\textwidth}
    >{\centering\arraybackslash}p{0.17\textwidth}
    >{\centering\arraybackslash}p{0.12\textwidth}
    @{}
}
\caption{
Covariate distributions according to ITN use.
Continuous variables are presented as mean (standard deviation), and categorical
variables as number (\%).
The last column reports the percentage of 1,000 bootstrap samples in which each covariate was selected for the propensity score model.
Ref denotes the reference category.
}
\label{tab:realdata_summary2}\\

\toprule
& \multicolumn{2}{c}{ITN use} & \\
\cmidrule(lr){2-3}
Covariates
& Non-user
& User
& Selected (\%) \\
& $N=3,510$
& $N=4,088$
& \\
\midrule
\endfirsthead

\toprule
\endhead

\bottomrule
\endfoot

Number of household members
    & 7.20 (2.98)
    & 6.59 (2.73)
    & 0.0 \\

Number of children 5 and under
    & 2.34 (1.08)
    & 2.15 (0.92)
    & 0.0 \\

Number of mosquito bed nets
    & 0.77 (1.05)
    & 2.10 (0.95)
    & 0.0 \\

Number of rooms used for sleeping
    & 2.28 (1.21)
    & 2.29 (1.13)
    & 0.0 \\

\textbf{Sex of head of household}
    & & & \\
\quad Male
    & 2,676 (76.2)
    & 3,197 (78.2)
    & Ref \\
\quad Female
    & 834 (23.8)
    & 891 (21.8)
    & 0.0 \\

Age of head of household (years)
    & 40.72 (12.76)
    & 38.70 (11.65)
    & 0.0 \\

\textbf{Wealth index}
    & & & \\
\quad Poorest
    & 1,120 (31.9)
    & 941 (23.0)
    & Ref \\
\quad Poorer
    & 766 (21.8)
    & 949 (23.2)
    & 0.0 \\
\quad Middle
    & 645 (18.4)
    & 914 (22.4)
    & 0.0 \\
\quad Richer
    & 555 (15.8)
    & 766 (18.7)
    & 0.0 \\
\quad Richest
    & 424 (12.1)
    & 518 (12.7)
    & 0.0 \\

\textbf{Type of toilet facility}
    & & & \\
\quad Flush toilet
    & 92 (2.6)
    & 150 (3.7)
    & Ref \\
\quad Pit toilet latrine
    & 2,839 (80.9)
    & 3,282 (80.3)
    & 0.0 \\
\quad No facility
    & 573 (16.3)
    & 651 (15.9)
    & 0.2 \\
\quad Composting toilet
    & 6 (0.2)
    & 5 (0.1)
    & 0.0 \\

\textbf{Source of drinking water}
    & & & \\
\quad Piped water
    & 582 (16.6)
    & 706 (17.3)
    & Ref \\
\quad Tube well water
    & 50 (1.4)
    & 56 (1.4)
    & 0.0 \\
\quad Dug well
    & 378 (10.8)
    & 417 (10.2)
    & 0.0 \\
\quad Surface water
    & 2,492 (71.0)
    & 2,895 (70.8)
    & 0.1 \\
\quad Other water sources
    & 8 (0.2)
    & 14 (0.3)
    & 0.0 \\

\textbf{Has electricity}
    & & & \\
\quad No
    & 3,175 (90.5)
    & 3,670 (89.8)
    & Ref \\
\quad Yes
    & 335 (9.5)
    & 418 (10.2)
    & 0.0 \\

\textbf{Main floor material}
    & & & \\
\quad Natural
    & 3,072 (87.5)
    & 3,552 (86.9)
    & Ref \\
\quad Rudimentary
    & 8 (0.2)
    & 13 (0.3)
    & 0.0 \\
\quad Finished
    & 430 (12.3)
    & 523 (12.8)
    & 0.0 \\

\textbf{Main wall material}
    & & & \\
\quad Natural
    & 516 (14.7)
    & 692 (16.9)
    & Ref \\
\quad Rudimentary
    & 2,190 (62.4)
    & 2,483 (60.7)
    & 1.3 \\
\quad Finished
    & 804 (22.9)
    & 913 (22.3)
    & 0.0 \\

\textbf{Main roof material}
    & & & \\
\quad Natural
    & 2,520 (71.8)
    & 2,919 (71.4)
    & Ref \\
\quad Rudimentary
    & 14 (0.4)
    & 22 (0.5)
    & 0.0 \\
\quad Finished
    & 976 (27.8)
    & 1,147 (28.1)
    & 0.0 \\
    
\end{longtable}

\endgroup
\restoregeometry

\begin{table}[p]
\centering
\caption{Results of analyzing the DRC-DHS data using the IPW estimator under specific allocation strategies.
The first and second columns give the allocation strategy and estimand, respectively.
The third and fourth columns give the point estimates and 95\% confidence intervals for the Proposed method and the Full model, respectively.
The confidence intervals are calculated by the sandwich variance estimators.}
\label{tab:realdata_effects}
\small
\setlength{\tabcolsep}{8pt}
\begin{tabular}{@{}clrr@{}}
\toprule
$\theta$ & Estimand
& \multicolumn{1}{c}{Proposed}
& \multicolumn{1}{c}{Full} \\
\midrule
0.1 & DE &  $0.031\;[-0.076,\;0.139]$
         & $-0.030\;[-0.153,\;0.093]$ \\
    & IE & $-0.048\;[-0.128,\;0.033]$
         & $-0.104\;[-0.293,\;0.085]$ \\
    & TE & $-0.017\;[-0.106,\;0.072]$
         & $-0.134\;[-0.283,\;0.015]$ \\
    & OE & $-0.025\;[-0.101,\;0.051]$
         & $-0.113\;[-0.284,\;0.057]$ \\
\addlinespace
0.3 & DE & $-0.045\;[-0.086,\;-0.003]$
         &  $0.072\;[-0.014,\;0.158]$ \\
    & IE &  $0.011\;[-0.024,\;0.046]$
         &  $0.073\;[-0.072,\;0.219]$ \\
    & TE & $-0.034\;[-0.070,\;0.003]$
         &  $0.145\;[0.001,\;0.290]$ \\
    & OE &  $0.018\;[-0.015,\;0.051]$
         &  $0.088\;[-0.040,\;0.215]$ \\
\addlinespace
0.5 & DE & $-0.040\;[-0.069,\;-0.011]$
         &  $0.012\;[-0.040,\;0.063]$ \\
    & IE &  $0.000\;[0.000,\;0.000]$
         &  $0.000\;[0.000,\;0.000]$ \\
    & TE & $-0.040\;[-0.069,\;-0.011]$
         &  $0.012\;[-0.040,\;0.063]$ \\
    & OE &  $0.000\;[0.000,\;0.000]$
         &  $0.000\;[0.000,\;0.000]$ \\
\addlinespace
0.7 & DE & $-0.049\;[-0.075,\;-0.024]$
         &  $0.090\;[0.054,\;0.126]$ \\
    & IE & $-0.046\;[-0.079,\;-0.014]$
         & $-0.147\;[-0.264,\;-0.031]$ \\
    & TE & $-0.096\;[-0.134,\;-0.058]$
         & $-0.058\;[-0.200,\;0.085]$ \\
    & OE & $-0.062\;[-0.090,\;-0.033]$
         & $-0.092\;[-0.204,\;0.021]$ \\
\addlinespace
0.9 & DE & $-0.068\;[-0.133,\;-0.004]$
         &  $0.077\;[0.013,\;0.140]$ \\
    & IE & $-0.097\;[-0.165,\;-0.029]$
         & $-0.132\;[-0.254,\;-0.011]$ \\
    & TE & $-0.166\;[-0.214,\;-0.118]$
         & $-0.056\;[-0.209,\;0.097]$ \\
    & OE & $-0.140\;[-0.184,\;-0.096]$
         & $-0.081\;[-0.202,\;0.040]$ \\
\bottomrule
\end{tabular}

\medskip
\begin{minipage}{0.95\linewidth}
\footnotesize
DE: direct effect; IE: indirect effect; TE: total effect; OE: overall effect.
The reference allocation probability for IE, TE, and OE is $\theta=0.5$.
At the reference allocation probability, IE and OE are zero by definition.
\end{minipage}
\end{table}

\end{document}